%% file: CollaborativeCS_final.tex
\documentclass[final]{siamart1116}
\usepackage[top=1in, bottom=1in, left=1in, right=1in]{geometry}
\usepackage[english]{babel}
\input{ex_shared}

\ifpdf
\hypersetup{
  pdftitle={\TheTitle},
  pdfauthor={\TheAuthors}
}
\fi

\usepackage{tabu}

\newcommand{\argmin}{\mathop{\rm argmin}}
\newcommand{\argmax}{\mathop{\rm argmax}}

\newcommand{\trace}{\operatorname{trace}}
\newcommand{\T}{{\cal T}}

\newcommand{\calU}{\mathcal{U}}
\newcommand{\calO}{\mathcal{O}}
\newcommand{\calS}{\mathcal{S}}
\newcommand{\calP}{\mathcal{P}}

\newcommand{\setB}{\mathbb{B}}

\newcommand{\vct}[1]{\boldsymbol{#1}}
\newcommand{\mtx}[1]{\boldsymbol{#1}}

\newcommand{\vd}{\vct{d}}
\newcommand{\ve}{\vct{e}}

\newcommand{\vh}{\vct{h}}

\newcommand{\vq}{\vct{q}}
\newcommand{\vr}{\vct{r}}
\newcommand{\vs}{\vct{s}}

\newcommand{\vv}{\vct{v}}

\newcommand{\vx}{\vct{x}}
\newcommand{\vy}{\vct{y}}
\newcommand{\vz}{\vct{z}}

\newcommand{\vpsi}{\vct{\psi}}
\newcommand{\vxi}{\vct{\xi}}
\newcommand{\vzero}{\vct{0}}

\newcommand{\mA}{\mtx{A}}
\newcommand{\mB}{\mtx{B}}
\newcommand{\mC}{\mtx{C}}
\newcommand{\mD}{\mtx{D}}
\newcommand{\mE}{\mtx{E}}

\newcommand{\mG}{\mtx{G}}
\newcommand{\mH}{\mtx{H}}

\newcommand{\mP}{\mtx{P}}
\newcommand{\mQ}{\mtx{Q}}
\newcommand{\mR}{\mtx{R}}
\newcommand{\mS}{\mtx{S}}
\newcommand{\mT}{\mtx{T}}
\newcommand{\mU}{\mtx{U}}
\newcommand{\mV}{\mtx{V}}
\newcommand{\mW}{\mtx{W}}
\newcommand{\mX}{\mtx{X}}
\newcommand{\mY}{\mtx{Y}}
\newcommand{\mZ}{\mtx{Z}}

\newcommand{\mOmega}{\mtx{\Omega}}
\newcommand{\mPhi}{\mtx{\Phi}}
\newcommand{\mPsi}{\mtx{\Psi}}

\newcommand{\mSigma}{\mtx{\Sigma}}

\newcommand{\mGamma}{\mtx{\Gamma}}
\newcommand{\mId}{{\bf I}}

\newcommand{\mzero}{{\bf 0}}

\newcommand{\e}{\begin{equation}}
\newcommand{\ee}{\end{equation}}
\newcommand{\eqn}{\begin{eqnarray}}
\newcommand{\eeqn}{\end{eqnarray}}

\newcommand{\MAT}{\left[ \begin{array}}
\newcommand{\mat}{\end{array} \right]}

\usepackage{bm}
\usepackage{enumerate}

\begin{document}

\maketitle

\begin{abstract}
Based on the maximum likelihood estimation principle, we derive a collaborative estimation framework that fuses several different estimators and yields a better estimate. Applying it to compressive sensing (CS), we propose a collaborative CS (CCS) scheme consisting of a bank of $K$ CS systems that share the same sensing matrix but have different sparsifying dictionaries. This CCS system is expected to yield better
performance than each individual CS system, while requiring the same time as that needed for each individual CS system when a parallel
computing strategy is used. We then provide an approach to designing optimal CCS systems by utilizing a measure that involves both the sensing matrix and dictionaries and hence allows us to simultaneously
optimize the sensing matrix and all the $K$ dictionaries. An alternating
minimization-based algorithm is derived for solving the corresponding optimal design problem. With a rigorous convergence analysis, we  show that the proposed algorithm is convergent. Experiments are carried out to confirm the theoretical results and show that
the proposed CCS system yields significant improvements over the existing CS systems in terms of the signal recovery accuracy.
\end{abstract}



\section{Introduction}
Compressive (or compressed) sensing (CS), aiming to sample signals beyond Shannon-Nyquist limit~\cite{CandesRombergTao2006RobustUncertaintyPrinciples,CandesWakin2008IntroductionCS,candes2006stableRecovery,donoho2006compressedSensing},  is a mathematical framework that efficiently compresses a signal vector  $\vx\in\Re^{N}$ by a measurement vector $\vy \in \Re^{M}$ of the form
\begin{align}
\vy = \mPhi \vx,
\label{eq:sensing}\end{align}
where $\mPhi\in\Re^{M\times N}$ ($M\ll N$) is a carefully chosen {\em sensing matrix} capturing the information contained in the signal vector~$\vx$.

As $M\ll N$, we have to exploit additional constraints or structures on the signal vector $\vx$ in order to recover it from the measurement $\vy$ and {\em sparsity} is such a structure. In CS, it is assumed that the original signal $\vx$ can be expressed as a linear combination of few elements/atoms from a $l_2$-normalized set $\{\vpsi_\ell \}$, i.e., $||\vpsi_\ell||_2=1,~\forall~\ell$:
\begin{align}
\vx = \sum_{\ell=1}^L \vpsi_\ell s_\ell = \mPsi \vs,
\label{eq:represet x}\end{align}
where $\mPsi:=\begin{bmatrix}\vpsi_1 & \vpsi_2 & \cdots & \vpsi_L   \end{bmatrix}\in\Re^{N\times L} $ is called the {\em dictionary}  and the entries of $\vs\in\Re^{L}$ are referred to as coefficients.  $\vx$ is said to be $\kappa$-sparse in $\mPsi$ if $\|\vs\|_0 = \kappa$, where $\|\vs\|_0$ is the number of non-zero elements of the vector $\vs$. Under the framework of CS, the original signal $\vx$ can be recovered from $\vx = \mPsi \widehat \vs$ with  $\widehat \vs$ the solution of the following problem when the sparsity level $\kappa$ is small:
\begin{align}
\widehat \vs := \argmin_{\vs}||\vy - \mA \vs||^2_2~~~\text{s.t.}~~\|\vs\|_0\leq \kappa, \label{eq:recov-1}\end{align}
where $\mA:=\mPhi \mPsi$ is called the {\it equivalent dictionary}. A CS system is referred to equations (\ref{eq:sensing}) and (\ref{eq:represet x}) plus an algorithm used to solve (\ref{eq:recov-1}) and its ultimate goal is to reconstruct the original signal $\vx$ from the low-dimensional measurement $\vy$. The latter depends strongly on the properties of $\mPsi$ and $\mPhi$.

In general, the choice of dictionary $\mPsi$ depends on the signal model. Based on whether the dictionary is given in a closed form or learned from data, the previous works on designing dictionary can be roughly classified into two categories. In the first one, one attempts to concisely capture the structure contained in the signals of interest by a well-designed dictionary, like the wavelet dictionary~\cite{mallat1999wavelet} for piecewise regular signals, the Fourier matrix for frequency-sparse signals, and a multiband modulated Discrete Prolate
Spheroidal Sequences (DPSS's) dictionary for sampled multiband signals~\cite{ZhuWakin2015MDPSS}. The second category is to learn the dictionary in an adaptive way to  sparsely represent a set of representative signals (called training data). Typical algorithms for solving this sparsifying dictionary learning problem include the method of optimal directions (MOD)~\cite{engan1999MOD}, K-singular value decompostion (K-SVD)~\cite{aharon2006KSVD} based algorithms, and the method for designing incoherent sparsifying dictionary~\cite{li2017new}. Peyre~\cite{peyre2010best} considered designing signal-adapted orthogonal dictionaries, where the dictionaries are considered to be tree-structured in order to avoid the intractability  involved in the optimization problem and  lead to fast algorithms for solving the problem. Learning a sparsifying dictionary has proved to be extremely useful and has achieved state-of-the-art performance in many signal and image processing applications, like image denoising, impainting, deblurring and compression~\cite{aharon2008sparse,Elad2010book}.

To recover the signal $\vx$ from its low dimensional measurement $\vy$, another important factor in a CS system is to select an appropriate sensing matrix $\mPhi$ that preserves the useful information contained in $\vx$. It has been shown that a sparse signal $\vx$ can be exactly reconstructed from its measurement $\vy$
with greedy algorithms such as those orthogonal matching pursuit (OMP)-based ones~\cite{blumensath2008iterativeThresholding,mallat1993matchingPursuit,needell2009cosamp,tropp2004greed,wang2010sparse} or methods based on convex optimization~\cite{CandesRombergTao2006RobustUncertaintyPrinciples,chen1998atomicDecomposition,donoho2003optimallySparseRepl1}, if the equivalent dictionary $\mA=\mPhi\mPsi$ satisfies the restricted isometry property (RIP)~\cite{baraniuk2008simple,CandesRombergTao2006RobustUncertaintyPrinciples,CandesWakin2008IntroductionCS}, or the sensing matrix $\mPhi$ satisfies the so-called $\mPsi$-RIP \cite{candes2011compressed}. However, despite the fact that a random matrix $\mA$ (or $\mPhi$)  with a specific distribution satisfies the RIP (or $\mPsi$-RIP) with high probability~\cite{baraniuk2008simple}, it is hard to certify the RIP (or $\mPsi$-RIP) for a given sensing matrix that is utilized in practical applications~\cite{bandeira2013certifying}. Thus, Elad~\cite{elad2007optimized} proposed to design a sensing matrix via minimizing the mutual coherence, another property of the equivalent dictionary that is more tractable and  much easier to verify. Since then, it has led to a class of approaches for designing CS systems~\cite{abolghasemi2012gradient,bai2015alternating,duarte2009learning,hong2016efficient,li2013projection,xu2010optimized,zelnik2011sensing}. When the signal is not exactly sparse, which is true for most natural signals such as image signals even with a learned dictionary, it is observed that the sensing matrix obtained via only minimizing the mutual coherence yields poor performance.  Recently, a modified approach to designing sensing matrix was proposed in \cite{li2015designing}, where a measure that takes both the mutual coherence and the sparse representation errors into account is used,  and achieves state-of-the-art performance for image compression.


The ultimate objective of dictionary learning is to determine a $\mPsi \in \Re^{N\times L}$ such that the signals of interest (like image patches extracted from natural images) can be well represented with a sparsity level $\kappa$ given.  To that end, a large $L$ is preferred as increasing $L$ can enrich the components/atoms that are used for sparse representation. One notes, however, that mutual coherence of $\mPsi$,  formally defined as
$$\mu(\mPsi) := \max_{i\neq j}\frac{|\vpsi_i^{\cal T} \vpsi_j|}{\|\vpsi_i\|_2\|\vpsi_j\|_2}$$
with $\cal T$ and $\vpsi_\ell$ denoting respectively the transpose operator and the $\ell$-th column of matrix $\mPsi$, satisfies~\cite{welch1974lower}
\e
\sqrt{\frac{L-N}{N(L-1)}}\leq \mu(\mPsi) \leq 1,
\label{mu-1}\ee
which implies that when $L$ is very large, the dictionary $\mPsi$  will have large mutual coherence. Since the equivalent dictionary $\mA=\mPhi \mPsi$, where $\mPhi$ is of dimension $M\times N$ with $M \ll N$, $\mu(\mA)$ gets in general larger than $\mu(\mPsi)$. Thus, increasing $L$  may affect sparse signal recovery dramatically. Besides, a large scale training data set is required in order to learn richer atoms (or features) and thus it needs a long time for the existing algorithms to learn a dictionary, though it is not a big issue for off-line design.

These motivate us to develop an alternative CS framework that yields a high performance in terms of reconstruction accuracy  as if a high dimensional dictionary were used in the traditional CS framework (see Fig.~\ref{figure:CS systems}(a)), but gets rid of the drawbacks of a high dimensional dictionary as mentioned above. Our argument here is that a signal within some class of interest can be well represented/approximated sparsely in different dictionaries or frames. Given  different training data, we can learn different dictionaries containing different atoms or features for the signals of interest. Though these  dictionaries perform differently one from another for a signal to be represented, a fusion of these representations obtained with these dictionaries may be better than any of the individual representations.   As pointed out in \cite{carrillo2013sparsity}, many natural signals of interest, e.g., images, video sequences,  include several
types of structures and can be well represented sparsely in various frames simultaneously. Based on this observation, the authors of  \cite{carrillo2013sparsity} proposed an approach for improving signal reconstruction using the co-sparsity analysis model \cite{nam2013cosparse}, where the analysis dictionary is designed based on a concatenation of (Parseval) frames  that are used to impose the signal sparsity.\\

The main contributions of this paper are stated as follows:
\begin{itemize}\item
The first contribution is to derive a new framework named collaborative estimator (CE) based on the maximum likelihood estimation (MLE) principle. Such a framework is actually a linear transformation of $K$ estimators used for estimating the same signals/parameters and is expected to yield a better estimate than any individual estimators. It is shown that the optimal transformation is determined by the covariance matrix of the estimation errors and  the optimal CE outperforms the best individual estimator when the estimation errors of the $K$ estimators are uncorrelated;
\item The second contribution regards application of the CE framework to CS, leading to a collaborative CS (CCS) system, in which the set of  individual estimators is specified by a bank of $K$ traditional CS systems that have an identical sensing matrix but different dictionaries of dimension $N\times L_0$. See Fig.~\ref{figure:CS systems}(b). More importantly, we provide an approach to simultaneously designing the sensing matrix and sparsifying dictionaries for the proposed CCS system. Unlike \cite{bai2015alternating,duarte2009learning}, where the sensing matrix and the dictionary are designed independently (though the sensing matrix is involved when the dictionary is updated), we use a measure that is a function of sensing matrix, dictionaries and sparse coefficients and hence allows us to consider both the sensing matrix and dictionaries under one and the same framework. By doing so, it is possible to enhance the system performance and to come up with an algorithm with guaranteed convergence for the optimal CCS design;
    \item The third contribution is to provide an alternating minimization-based algorithm for optimal design of the proposed CCS system, where updating the dictionaries is executed  column by column with an efficient algorithm, the sparse coding is performed using a proposed OMP-based procedure, and the sensing matrix is updated analytically. It should be pointed out that the alternating minimization-based approach has been popularly used in designing sensing matrix and sparsifying dictionary \cite{aharon2006KSVD,duarte2009learning,elad2007optimized,li2013projection,li2015designing}, in which the convergence of the algorithms is usually neither ensured nor seriously considered. As one of the important results in this paper, a rigorous convergence analysis is provided and  the proposed algorithm is ensured to be convergent;
        \item Finally, we provide a convergence guarantee of the OMP algorithm~\cite{tropp2004greed} solving the sparse recovery problem (\ref{eq:recov-1}) without posing any condition on the linear system. In particular, we show that the OMP algorithm always generates a stationary point of the objective function for (\ref{eq:recov-1}), which is expected to be useful for convergence analysis of other OMP-based algorithms for  sensing matrix design and  dictionary learning ~\cite{aharon2006KSVD,duarte2009learning,li2018joint}.\\
\end{itemize}

It should be pointed out that there exist some related works reported. Multimeasurement vector (MMV) problems considering the estimation of jointly sparse signals under the distributed CS (DCS) framework~\cite{baron2009distributed} have been studied in~\cite{ye2015improving,zhu2017performance}. The proposed CCS differs from DCS in that the former only involves one measurement vector and aims to improve the performace of the classical CS performance by using multiple dictionaries, while the latter involves a number of measurement vectors observed for a set of sparse signals sharing common supports which are exploited to yield better performane than the classical approach that individually and independently solves the sparse recovry problem for each signal.
Elad and Yavneh~\cite{elad2009plurality} investigated the problem of detecting the clean signal $\vx$ from the measurement $\vy = \vx + \vv$, where $\vx = \mPsi \vs$ with $\vs$ the sparse representation in dictionary $\mPsi$, and demonstrated with experiments that using a set of sparse representations $\{\vs_i\}$ generated by randomizing the OMP algorithm, a more accurate estimate of $\vx$ can be obtained by fusing $\{\vs_i\}$ with a plain averaging strategy which, as to be seen, corresponds to  a special case of the fusion model we derive. Recently, the same structure of the CCS was proposed independently by Wijewardhana et al. in~\cite{wijewardhana2017signal}, where the same plain averaging model was used with the set of sparse vectors $\{\vs_i\}$ jointly obtained using a customized interior-point method and  the performance improvement was experimentally demonstrated in the context of signal reconstruction from compressive samples.  Besides deriving a more general fusion model, we consider the joint optimal design of the $K$ CS systems in the context of the performance of the overall CCS system and develop an algorithm with guaranteed convergence for it. We will discuss further the differences and similarities between our proposed CCS and these related works in Section~\ref{sec:CCS} when we present our CCS system.
\\

\noindent{\bf Notations}: Throughout this  paper, finite-dimensional vectors and matrices are indicated by bold characters. The symbols $\mId$ and $\mzero$ respectively represent the identity and zero matrices with appropriate sizes. In particular, $\mId_N$  denotes the $N\times N$ identity matrix.  $\calU_N: = \left\{\vh\in\Re^N: \|\vh\|_2 = 1 \right\}$ is the unit sphere. Also, $\calO_{M,N}:=\{\mR\in\Re^{M\times N}: \mR \mR^\T= \mId_M\}$ with  $M \leq N$ - a set of matrices called orthogonal projectors in CS literature, and when $M=N$, $\calO_{N,N}$ is simply denoted as $\calO_{N}$. For any natural number $N$, we let $[N]$ denote the set $\{1,2,...,N\}$. Let $\calS_{k}$ be a set of samples for variable $\vz_k$ for $k\in[K]$ and $\calS_{1}\times \calS_{2}\cdots\times\calS_{K}$ denote the set consisting of joint variables $(\vz_1, \vz_2, \cdots, \vz_K)$. In particular, $\calS^K_{0}$ represents the set $\calS_{0}\times \calS_{0}\cdots\times\calS_{0}$ when $\calS_{k}=\calS_{0},~\forall~k$.  We adopt MATLAB notations for matrix indexing; that is, for a matrix $\mQ$,
its $(i,j)$-th element is denoted by $\mQ(i,j)$, its $i$-th row (or column) is denoted by $\mQ(i,:)$ (or $\mQ(:,i)$). When it is clear from the context, we also use $\vq_i$, $q_{i,j}$ to denote the $i$-th column and $(i,j)$-th element of $\mQ$. Similarly, we use $q_i$ or $\vq(i)$ to denote the $i$-th element of the vector $\vq\in \Re^N$.

The outline of this paper is given as follows. In Section~\ref{sec:colla CS},  a collaborative estimation framework is derived based on the MLE principle and then applied to compressed sensing, resulting in a collaborative CS system that fuses a set of individual CS systems for achieving high precision of signal reconstruction.  Section~\ref{sec:deisgn colla CS} is devoted to designing optimal CCS systems. Such a problem is formulated based on the idea of learning the sensing matrix and the sparsifying dictionaries simultaneously. An algorithm is derived for solving the proposed optimal CCS system design in Section~\ref{sec:algorithm}.  Section~\ref{sec:convergence} is devoted to  performance analysis, including implementation complexity of the proposed CCS and convergence analysis for the proposed algorithm for optimal CCS design. To demonstrate the performance of the proposed CCS systems, experiments are carried out in Section~\ref{sec:experiments}. The paper is concluded in Section~\ref{sec:conclusion}.

\section{A Collaborative Estimation-based CS Framework}
\label{sec:colla CS}
Let us consider a set of estimators, all used  for estimating the identical  signal $\vx\in \Re^{N\times 1}$. Assume that the output of the $i$-th estimator is given by
\e
\vx_i = \vx + \ve_i, \ i = 1,\ldots,K,  \label{ZH-1}\ee
where  $\ve_i$ is the estimation error of the $i$th estimator. How to fuse the set of estimates  $\{\vx_i\}$ to obtain a better estimate  $\widehat \vx = f(\vx_1,\ldots,\vx_K)$ has been an interesting topic in estimation theory \cite{kay1993fundamentals}. In the next subsection, we will derive such a fusion based on the maximum likelihood estimation principle.

\subsection{MLE-based collaborative estimators}
\label{sec:MLE coll estimator}
Denote
\e ~~~\bar \ve^{\cal T} := \MAT{cccccc}\ve^{\cal T}_1& \cdots& \ve^{\cal T}_i &\cdots & \ve^{\cal T}_K\mat,~
 \bar \vx^{\cal T} := \MAT{ccccc}\vx^{\cal T}_1& \cdots& \vx^{\cal T}_i& \cdots& \vx^{\cal T}_K\mat. \label{insert-1}\ee

Lemma~\ref{lem:collaborative-estimator} presented below yields an optimal estimate of $\vx$ with  $\bar \vx$.
\begin{lemma}\label{lem:collaborative-estimator} Let $\{\ve_i\}$ be the set of the estimation errors defined in \eqref{ZH-1}.  Under the assumption that $\bar \ve$  defined in (\ref{insert-1}) obeys a normal distribution $\mathcal N({\bf 0}, \bf \Gamma)$,  the best estimate $\widehat \vx$ of $\vx$, which can be achieved in the MLE sense with the observations $\{\vx_i\}$ in~\eqref{ZH-1}, is given by
\begin{align}
\widehat \vx = ({\bf \bar I}^{\cal T}{\bf \Gamma}^{-1}{\bf \bar I})^{-1}{\bf \bar I}^{\cal T}{\bf \Gamma}^{-1}\bar \vx := \mOmega \bar \vx =\sum^K_{i=1}\mOmega_i \vx_i, \label{collaborative-estimator}\end{align}
where $\bar \vx$  defined in (\ref{insert-1}),  ${\bf \bar I} := \MAT{ccccc}{\bf I}_N&\cdots&{\bf I}_N&\cdots&{\bf I}_N\mat^{\cal T} \in \Re^{NK\times N}$, $\mOmega =\MAT{ccccc} \mOmega_1&\cdots&\mOmega_i&\cdots&\mOmega_K\mat$ with $\mOmega_i \in \Re^{N\times N},~\forall~i$. Consequently,
\begin{align}
\mathbb E [\|\widehat \vx - \vx\|^2_2] = \trace[\mOmega {\bf \Gamma} \mOmega^{\cal T}],
\label{lem:sample mean}
\end{align}
where $\mathbb E [\cdot]$ and $\trace[\cdot]$ denote the mathematical expectation and the trace operation, respectively.
\end{lemma}
\begin{proof}[Proof of Lemma~\ref{lem:collaborative-estimator}] Under the assumptions made in this lemma, the  probability density function (PDF) of $\bar \ve$ is given by
\e
f_{\bar \ve}(\vxi) = \frac{e^{\left(-\frac{1}{2}\vxi^{\cal T}\bf \Gamma^{-1} \vxi\right)}}{\sqrt{2\pi |\bf \Gamma| }},
\label{pdf-1}\ee
where $|\cdot|$ represents the determinant of a matrix. According to the MLE principle~\cite{kay1993fundamentals}, the best estimate $\widehat \vx$ of the ideal signal $\vx$ using the measurements $\bar \vx$ is the one that maximizes the likelihood function in $\vx$, that is $f_{\bar \ve}(\vxi)$ with  $\vxi= \bar \vx- {\bf \bar I} \vx$. Note that $\bf \Gamma$ has nothing to do with $\vx$. The best estimate $\widehat \vx$ is the solution leading the derivative of $\frac{1}{2}{\vxi}^{\cal T}{\bf \Gamma}^{-1}\vxi$ {\it w.r.t.} $\vx$ to zero. Thus, $-{\bf \bar I}^{\cal T}{\bf \Gamma}^{-1}\vxi={\bf 0}$, which leads to  (\ref{collaborative-estimator}) directly.

As to the second part of the lemma, (\ref{lem:sample mean}) follows directly from the fact $\vx = ({\bf \bar I}^{\cal T}{\bf \Gamma}^{-1}{\bf \bar I})^{-1}({\bf \bar I}^{\cal T}{\bf \Gamma}^{-1}{\bf \bar I}) \vx = \mOmega {\bf \bar I} \vx$ and (\ref{collaborative-estimator}) with $\bar \vx={\bf \bar I} \vx + \bar \ve$.
\end{proof}

\vspace{0.5cm}
The estimator given by (\ref{collaborative-estimator}) is called a collaborative estimator of $\vx$ that fuses $K$ individual estimators of the same $\vx$. One notes that under the MLE principle, the collaborative estimate (\ref{collaborative-estimator}) is actually a linear transformation of the observations $\{\vx_i\}$. Such a transformation $\mOmega$ is determined by the covariance matrix $\bf \Gamma$ of the estimation errors $\bar \ve$ of the $K$ individual estimators.

Denote $\mG:=\mGamma^{-1}=\{\mG_{ij}\}$ with $\mG_{ij} \in \Re^{N\times N},~\forall~i,j=1, \cdots, K$. Then, $\mOmega$ can be specified in terms of $\{\mG_{ij}\}$ via
\e \mOmega_i =\left(\sum_{k,l \in [K]}\mG_{kl}\right)^{-1}\sum^K_{k=1}\mG_{ki},~\forall~i.\label{LBJ-1}\ee
Clearly, we have $\sum^K_{i=1}\mOmega_i={\bf I}_N$. Generally speaking, $\mOmega_i$ is fully parametrized with non-zero elements.

Depending on $\mGamma$, $\mOmega_i$ can have different structures, among which there are two interesting models of low implementation complexity:
\begin{itemize}\item $\mOmega_i = \diag(r_{i1}, \cdots, r_{in}, \cdots, r_{iN}):=\mD_i,~\forall~i$. (\ref{collaborative-estimator}) becomes
\begin{align}
\widehat \vx = \sum^K_{i=1}\mD_i \vx_i. \label{collaborative-estimatorII}\end{align}
\item $\mOmega_i = \omega_i {\bf I}_N,~\forall~i$. Consequently, the collaborative estimator is of form
\begin{align}
\widehat \vx = \sum_{i=1}^K \omega_i\vx_i,
\label{collaborative-estimatorX}\end{align}
which implies that the collaborative estimate is just a linear combination of the observations.\\
\end{itemize}

\noindent{\bf Remark~2.1:}
\begin{itemize}
\item (\ref{collaborative-estimator}) yields the optimal estimate with the observations collected from the individual estimators no matter these estimation errors are correlated or not. Though it is difficult to see relationship between the performance of the collaborative estimator and that of each individual estimator from (\ref{lem:sample mean}), extensive experiments showed that the former is much better than the latter. See the experiments results in Section~\ref{sec:experiments};
\item (\ref{collaborative-estimatorII}) and  (\ref{collaborative-estimatorX}) are special cases of (\ref{collaborative-estimator}), corresponding to different circumstances of the estimation errors $\{\ve_i\}$ generated by the $K$ estimators. Though these two models are not necessarily for the situations in which these errors are statistically independent/uncorrelated, very interesting properties of the collaborative estimators can be obtained when the errors are statistically independent, which give us some insights on the perspectives of the collaborative estimation.\\
     \end{itemize}

When $\{\ve_i\}$ are all statistically independent, ${\bf \Gamma} =\text{diag}({\bf \Gamma}_1, \cdots, {\bf \Gamma}_i, \cdots, {\bf \Gamma}_K)$ and hence
\begin{align}
\mOmega_i=\left(\sum^K_{k=1}{\bf \Gamma}^{-1}_k\right)^{-1}{\bf \Gamma}^{-1}_i\label{Omega-independent}\end{align}
and (\ref{lem:sample mean})  yields
\e
\mathbb E [\|\widehat \vx - \vx\|^2_2] = \sum^K_{\ell=1}\trace[\mOmega_{\ell}{\bf \Gamma}_{\ell} \mOmega^{\cal T}_{\ell}] = \trace\left[\left(\sum^K_{\ell=1}{\bf \Gamma}^{-1}_\ell\right)^{-1}\right].  \label{ideal case-0x}
\ee

Furthermore, when ${\bf \Gamma}_i= \diag(\gamma^{(i)}_1,\cdots,\gamma^{(i)}_{n},\cdots,\gamma^{(i)}_N):=\mGamma^{(i)}_d,~\forall~i$, (\ref{Omega-independent}) suggests that $\mOmega_i$ is also diagonal. This means that the corresponding collaborative estimator follows  model  (\ref{collaborative-estimatorII}) and particularly, when ${\bf \Gamma}_i=\epsilon_i^2 {\bf I}_N,~~\forall~i$,  the corresponding collaborative estimator actually  obeys  model  (\ref{collaborative-estimatorX}) with the weighting factors given by
\e \omega_i= \epsilon^{-2}_i/\sum^K_{\ell=1}\epsilon^{-2}_\ell,~\forall~i\label{setting omega}\ee
and
\e
\mathbb E [\|\widehat \vx - \vx\|^2_2] = \sum^K_{\ell=1}\omega^2_{\ell}\mathbb E [\|\vx_{\ell} - \vx\|^2_2] = \frac{N}{\sum^K_{\ell=1}\epsilon^{-2}_\ell}. \label{ideal case-0}
\ee

\noindent{\bf Remark~2.2}:
\begin{itemize}\item
First of all, as the weighting factors $\{\mOmega_i\}$ and $\{\omega_i\}$, given in (\ref{Omega-independent}) and (\ref{setting omega}), respectively,  are decreasing with $K$, it is expected that the performance of these collaborators gets improved with the number of estimators increasing. This claim is supported by (\ref{ideal case-0x}) and (\ref{ideal case-0}).  Furthermore, assume without loss of generality that the first estimator is the best among the $K$ estimators. It follows from (\ref{ideal case-0x}) that
 \eqn \mathbb E [\|\widehat \vx - \vx\|^2_2] & =& \trace\left[ {\bf \Gamma}_1 \left({\bf I}_N +\sum^K_{\ell=2} {\bf \Gamma}_1^{1/2}{\bf \Gamma}^{-1}_\ell {\bf \Gamma}_1^{1/2}\right)^{-1}\right]\nonumber\leq \trace\left[ {\bf \Gamma}_1]=  \mathbb E [\|\vx_1 - \vx\|^2_2] \right],\nonumber \eeqn
where  ${\bf \Gamma}_1^{1/2}$ denotes a (symmetric) square root matrix of  ${\bf \Gamma}_1$, that is ${\bf \Gamma}_1={\bf \Gamma}_1^{1/2}{\bf \Gamma}_1^{1/2}$.
The above equation implies that the variance of  estimation error by the collaborative estimator is smaller than that by the best one of the $K$ individual estimators and that it deceases with the number $K$ of the individual estimators in fusion.
\item
Based on  (\ref{collaborative-estimator}) that is derived under some specific conditions, we can propose different models for practical problems, in which we usually do not have much statistical priori information of the signals and have other issues such as implementation complexity to be considered. (\ref{collaborative-estimatorII}) and  (\ref{collaborative-estimatorX}) are two of such models. As to be seen in Section~\ref{sec:experiments}, with a set-up, in which the estimation errors are correlated, that is the covariance matrix $\mGamma$ is not block-diagonal, we set the general model (\ref{collaborative-estimator}) and  the simplified model (\ref{collaborative-estimatorX}) according to (\ref{Omega-independent}) and $\omega_i=\frac{1}{K},~\forall~i$, respectively, {\it both are not optimal for this case}. It is observed that these two collaborative estimators still outperform significantly all the individual estimators, though yielding a performance not as good as the optimal collaborative estimator whose transformation $\mOmega$ is set using (\ref{LBJ-1}) with the true covariance matrix $\mGamma$.
\end{itemize}

\subsection{Collaborative compressed sensing}\label{sec:CCS}
As known, the ultimate goal of a CS system $(\mPhi, \mPsi)$ is to reconstruct or estimate the original signal $\vx$ from its compressed counterpart $\vy$.
Fig.~\ref{figure:CS systems}(a) depicts the block-diagram of such a system with image compression.

\begin{figure*}[t]
\centering
\includegraphics[width=7in]{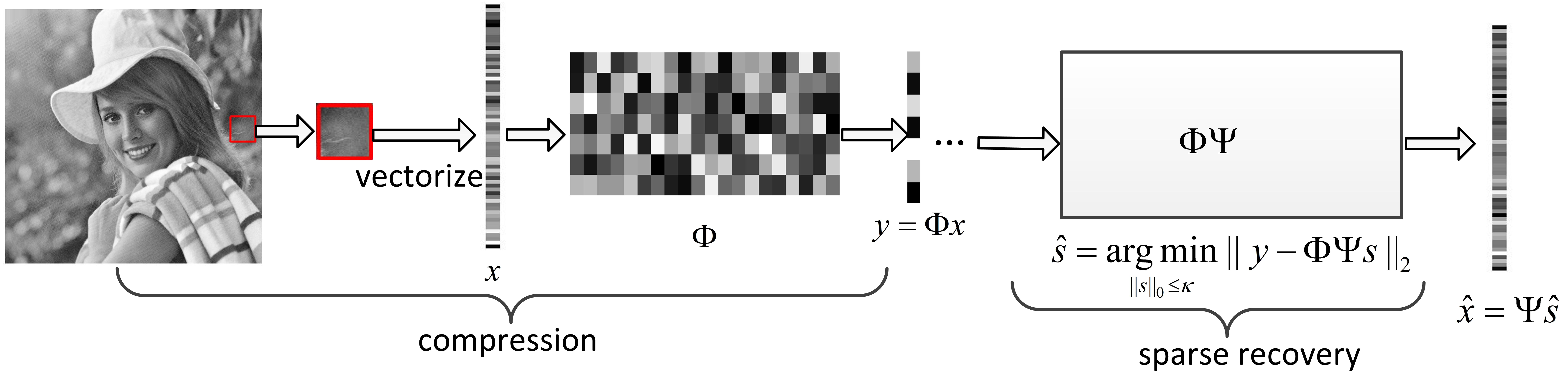}
\centerline{(a)}
\centering
\includegraphics[width=7in]{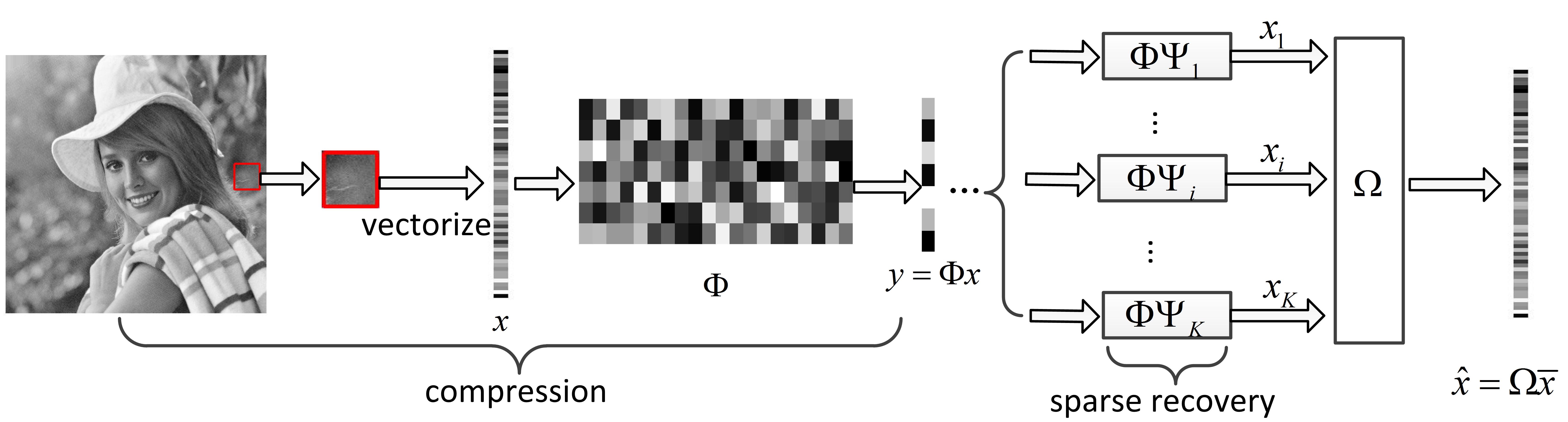}
\centerline{(b)}
\caption{Illustration of (a) the classical CS system, and (b) the proposed collaborative CS system for patch-based image compression. Here, $\vx$ represents one patch of the image.
}\label{figure:CS systems}
\end{figure*}

Specifying the $K$ estimators with a bank of $K$ CS systems $\{(\mPhi, \mPsi_i)\}$, we obtain a CS-based collaborative estimator which is shown in Fig.~\ref{figure:CS systems}(b). Such a system is named collaborative compressive sensing (CCS) in the sequel.

For the same measurement $\vy$~--~the compressed version of $\vx$ via (\ref{eq:sensing})~--~the estimate $\vx_i$ of $\vx$ given by the $i$-th CS system $(\mPhi, \mPsi_i)$ is~\footnote{Here, the reconstruction error $\ve_i=\vx_i -\vx$ depends  mainly on the CS system $(\mPhi, \mPsi_i)$ as well as the algorithm used for signal recovery and is assumed to be statistically independent of the signal $\vx$. A learning strategy may be applied for more complicated situations.}
\e \vx_i = \mPsi_i \widehat \vs_i,~\forall~i,\label{eq:signal model multi dic}\ee
where $\mPsi_i \in \Re^{N\times L_0}, ~\forall~i$ are designed using different training samples and the sparse coefficient vector $\widehat \vs_i$ is obtained using a reconstruction scheme, say (\ref{eq:recov-1}), with the measurement $\vy$ and the equivalent dictionary $\mA_i = \mPhi\mPsi_i$. \\

Combining (\ref{collaborative-estimator}) and (\ref{eq:signal model multi dic}) yields
\[
\widehat \vx =\mOmega\text{diag}(\mPsi_1, \cdots, \mPsi_i, \cdots, \mPsi_{K}) \begin{bmatrix} \widehat\vs_1 \\ \vdots \\ \widehat\vs_i\\ \vdots\\ \widehat\vs_K \end{bmatrix}:= \mPsi_c \vs_c,
\]
which can be viewed as a linear sparse approximation of $\vx$ using a larger dictionary $\mPsi_c$, having $KL_0$ atoms, with a sparsity level of $\kappa K$. A signal $\vx$ can be represented in many different ways. In fact, in the proposed CCS system  $\vx$ is $\kappa$-sparsely represented in each of $K$ different dictionaries and the collaborative estimator (\ref{collaborative-estimator}) is equivalent to a representation in a high dimensional dictionary $\mPsi_c \in \Re^{N\times K L_0}$ and a constrained block sparse coefficient vector $\vs_c$. It is due to this equivalent high dimensional dictionary that enriches the signal representation ability and hence makes the proposed CCS scheme yield an improved performance.\\

As discussed in Section~\ref{sec:MLE coll estimator}, the choice of $\mOmega$ depends on the statistical information of the errors corresponding to the estimate $\vx_i$. As illustrated in Section~\ref{sec:experiments}, when the prior information of  the statistical information of the estimation errors is not available, a practical strategy is to take the average of all the  estimates available. This coincides with the fusion strategy used in \cite{elad2009plurality} and \cite{wijewardhana2017signal}.\\

\noindent{\bf Remark~2.3}:
\begin{itemize}
\item In \cite{elad2009plurality}, the authors consider estimating the clean signal $\vx$ from the measurement $\vy$ of form $\vy=\vx+\vv$, where $\vv$ is assumed to be zero mean independent and identically distributed (i.i.d.) Gaussian and $\vx$ is sparse in a given dictionary $\mPsi$: $\vx = \mPsi \vs$. With a set of sparse representations $\{\vs_i\}$ generated by a randomized OMP algorithm, which is intended to  solve the minimization of $||\vy - \mPsi \tilde \vs||^2_2$ under the sparsity constraint imposed on $\tilde \vs$, an estimator for $\vx$ was proposed by fusing $\{\vs_i\}$ with a plain averaging, which is equivalent to
\begin{align}
\widehat \vx = \frac{1}{K}\sum_{i=1}^K \vx_i, \label{collaborative-estimatorX-plain averageing}\end{align}
where $\vx_i = \mPsi \vs_i = \vx + \mPsi(\vs_i-\vs):= \vx +\ve_i, ~\forall~i$. Clearly, the above fusion is the same as (\ref{collaborative-estimatorX}) with $\omega_i =\frac{1}{K},~\forall~i$. Experiments in~\cite{elad2009plurality} showed that such obtained $\widehat \vx$ outperforms the maximum a posteriori probability (MAP) estimator. In this case, there is nothing to ensure that the corresponding estimation errors $\{\ve_i\}$ are uncorrelated, which is one of the condition assumed when (\ref{collaborative-estimatorX-plain averageing}) is derived under the MLE principle;
\item A practical model is usually  derived with some ideal conditions assumed. This is the case of the fusion model (\ref{collaborative-estimatorX-plain averageing}). Experiments have shown that such a model works very well in some circumstances \cite{elad2009plurality}, \cite{wijewardhana2017signal}. The analysis in \cite{elad2009plurality} and the one provided in Section~\ref{sec:MLE coll estimator} can be served as complementary to each other for studying the theoretical foundation of such a model;
\item The same CCS scheme was independently proposed in \cite{wijewardhana2017signal}, in which the plain averaging fusion (\ref{collaborative-estimatorX-plain averageing}) is adopted and the sensing matrix as well as the dictionaries  are assumed to be given, and the focus is to derive a customized interior-point method to jointly obtain the set of sparse vectors $\{\vs_i\}$. Instead of using fixed dictionaries and random sensing matrix, these elements are simultaneously optimized in our proposed CCS and hence the performance of the system is enhanced. This will be discussed in next section.
    \end{itemize}


\section{Design of Optimal CCS Systems}
\label{sec:deisgn colla CS}
In the classical dictionary learning, the sparsifying dictionary is designed using a collection of available signal samples. For example, in image  application the signal samples, called training samples, are collected from a set of patch-based vectors generated from some database (e.g., LabelMe \cite{russell2008labelme}) that contains a huge number of different images taken a priori.

Traditionally,  a CS system $(\mPhi, \mPsi)$ is designed with the sparsifying dictionary $\mPsi$ learned using the training samples first and the sensing matrix $\mPhi$ then determined based on the learned $\mPsi$, where the two stages involve different design criteria.

Let
\[
\underline\mX =\begin{bmatrix}\mX_1&\cdots&\mX_i&\cdots&\mX_K\end{bmatrix} \in \Re^{N\times J_0K}
\] be  the data matrix of training samples with $\mX_i \in \Re^{N\times J_0},~\forall~i$. The problem we confront with is to learn $\mPsi_i$ using $\mX_i$ for $i=1,2, \cdots, K$ and the same $\mPhi$ shared by all the $K$ dictionaries.

In this section, we provide an approach that allows us to optimize all $\mPhi$ and $\{\mPsi_i\}$ jointly using the same criterion.


\subsection{Structure of sensing matrices}\label{sec:signal model}
Before formulating the optimal CCS design problem, let us consider an alternative reconstruction scheme to (\ref{eq:recov-1}). A class of structures for sensing matrices is then derived under such a scheme.

Consider a CS system $(\mPhi, \mPsi)$ and that the original signal $\vx$ is given by the following more general signal model than (\ref{eq:represet x}):
\e
\vx = \mPsi \vs + \ve,
\label{eq:signal model}\ee
where $\ve$ is the {\em sparse representation noise} (a.k.a {\em signal noise}). The measurement $\vy = \mPhi \vx$ is then of form
\[
\vy = \mA \vs + \vv,
\]
where $\vv := \vy - \mA\vs = \mPhi \ve$  and $\mA = \mPhi \mPsi$ is the equivalent dictionary of the CS system.

Suppose  $\ve$ is of normal distribution with a zero-mean and
\[
\mathbb{E}[\ve\ve^{\cal T}]= \sigma_e^2 \mId_N.
\]
Then, $\vv $ has the multivariate normal distribution with ${\cal N}(\vzero,\sigma_e^2\mPhi\mPhi^{\cal T})$, that is its PDF $f_{\vv}(\vxi)$ is given by (\ref{pdf-1}), where
${\bf \Gamma}=\sigma_e^2\mPhi\mPhi^{\cal T}$
and the sensing matrix $\mPhi$ is assumed of full row rank. As understood, according to the MLE principle along with the sparsity assumption on $\vs$, the best estimate $\widehat \vs$ of the sparse coefficient vector $\vs$ that can be achieved with the given CS system and the measurement $\vy$ is the one that maximizes the likelihood function in $\vs$, that is $f_{\vv}(\vxi)$ with  $\vxi= \vy - \mA\vs$.  This leads to
\e \begin{split}
\widehat \vs := & \argmin_{\vs} \|(\mPhi\mPhi^{\cal T})^{-1/2} (\vy - \mA\vs)\|_2^2,\\
& \text{s.t.} \ \|\vs\|_0 \leq \kappa.
\end{split}\label{eq:MLE for s}\ee
(\ref{eq:MLE for s}) is referred  to as pre-conditioned reconstruction scheme (PRS) \cite{li2018joint}, which is different from the traditional scheme (\ref{eq:recov-1}).

Note that  the objective function in~\eqref{eq:MLE for s} can be rewritten as
\e
\|(\mPhi\mPhi^{\cal T})^{-1/2} (\vy - \mA\vs)\|_2^2 = \|\overline \vy - \overline \mPhi \mPsi \vs\|_2^2,
\label{eq:MLE equivalent to}\ee
where  $\overline \mPhi:= (\mPhi\mPhi^{\cal T})^{-1/2} \mPhi$ and $\overline \vy := (\mPhi\mPhi^{\cal T})^{-1/2} \vy=\bar \mPhi \vx$.

Now, we present some results regarding the RIP properties of the CS systems $(\mPhi, \mPsi)$ and $(\bar \mPhi, \mPsi)$ under the assumption that they utilize the classical reconstruction scheme (\ref{eq:recov-1}) for signal recovery.

A matrix $\mQ$ is said to satisfy the RIP of order $\kappa$ if there exists a constant $\gamma\in(0,1)$ such that
\e
\sqrt{1-\gamma}\leq\frac{\|\mQ\vs\|_2}{\|\vs\|_2}\leq \sqrt{1+\gamma}
\label{eq:RIP}\ee
for all $\|\vs\|_0 \leq \kappa$.  It is shown in \cite[Lemma 4.1]{davenport2012pros} that the matrix $\overline \mPhi$ satisfies the RIP if so does the sensing matrix $\mPhi$.
A generalization of the RIP is $\mPsi$-RIP \cite{candes2011compressed}: a matrix $\mPhi$ is said to be $\mPsi$-RIP of order $\kappa$ if there exists a constant $\delta\in(0,1)$ such that
\e
\sqrt{1-\delta}\leq\frac{\|\mPhi\mPsi\vs\|_2}{\|\mPsi \vs\|_2}\leq \sqrt{1+\delta}
\label{eq:Psi RIP}\ee
holds for all $\vs$ with $\|\vs\|_0\leq \kappa$. We have the following interesting results.
\begin{lemma}\label{lem:RIP of MLE}
Let $(\mPhi, \mPsi)$ be a CS system. Suppose that $\mPsi$ satisfies the RIP of  order $\kappa$ with constant $\gamma\in (0,1)$ and that $\mPhi \in \Re^{M\times N}$ satisfies the $\mPsi$-RIP  of order $\kappa$ with constant $\delta <1$ and $\vx =\mPsi \vs$ with $\|\vs\|_0 \leq \kappa$. Denote $\overline \mPhi= (\mPhi\mPhi^{\cal T})^{-1/2} \mPhi$, $\mA=\mPhi \mPsi$ and $\overline \mA= \overline \mPhi\mPsi$. Then,
\e
\sqrt{(1-\delta)(1-\gamma)} \leq \frac{\|\mA \vs\|_2}{\|\vs\|_2}\leq \sqrt{(1+\delta)(1+\gamma)},
\label{eq:RIP A}\ee
and
\e
\frac{\sqrt{(1-\delta)(1-\gamma)}}{\pi_{\max}(\mPhi)} \leq \frac{\|\overline \mA \vs\|_2}{\|\vs\|_2}\leq \frac{\sqrt{(1+\delta)(1+\gamma)}}{\pi_{\min}(\mPhi)},
\label{eq:RIP tilde A}\ee
where $\pi_{\max}(\mPhi)$ and $\pi_{\min}(\mPhi)$ denotes the largest and smallest nonzero singular values of $\mPhi$ respectively.
\end{lemma}

\begin{proof}[Proof of Lemma~\ref{lem:RIP of MLE}]
Since $\mPsi$ is RIP,  \eqref{eq:RIP} holds with $\mQ =\mPsi$, which combined with  \eqref{eq:Psi RIP} results in \eqref{eq:RIP A} in a straightforward fasion.

Let $\overline\mPhi = \mT^{-1} \mPhi$ for any non-singular matrix $\mT$. Then, $\mPhi = \mT \overline \mPhi$. It follows from matrix analysis that
\eqn \|\overline\mPhi \vx\|_2 \leq \|\mT^{-1}\|_2~ \|\mPhi\vx\|_2,\ \|\mPhi \vx\|_2 \leq \|\mT\|_2~ \|\overline\mPhi\vx\|_2,\nonumber
\eeqn
which implies $
\|\mT\|^{-1}_2~\|\mPhi \vx\|_2 \leq \|\overline\mPhi \vx\|_2 \leq \|\mT^{-1}\|_2 \|\mPhi\vx\|_2$.
With $\mT=(\mPhi\mPhi^\T)^{1/2}$ and $\mPhi$ satisfying $\mPsi$-RIP, we have
\e
\frac{\sqrt{1-\delta}}{\pi_{\max}(\mPhi)} \leq \frac{\|\overline \mPhi \vx\|_2}{\|\vx\|_2}\leq \frac{\sqrt{1+\delta}}{\pi_{\min}(\mPhi)}.
\label{eq:RIP tilde Phi}  \ee
Thus \eqref{eq:RIP tilde A} follows from   \eqref{eq:RIP tilde Phi} and \eqref{eq:RIP A}. This completes the proof.
\end{proof}

\vspace{0.5cm}
\noindent{\bf Remark~3.1}:
\begin{itemize}\item It was shown in \cite{candes2011compressed} that a random matrix $\mPhi$ with a specific distribution satisfies \eqref{eq:Psi RIP} with high probability. The 2nd part of Lemma~\ref{lem:RIP of MLE} implies that the equivalent dictionaries $\mA$ and $\overline\mA$ have RIP of order $\kappa$ if so is the dictionary $\mPsi$;
\item Both \eqref{eq:RIP A} and \eqref{eq:RIP tilde A} imply that $\overline \mPhi$ severs as an isometric operator similar to $\mPhi$ for all $\kappa$-sparse signals, but the constants on the left and right hand sides of them  depend on the condition number of $\mPhi$.
    Specifically, for a given $\mPsi$,  $\overline \mPhi$ makes  $\overline \mA$ achieve the tightest RIP constants if $\mPhi$ is a tight frame. \footnote{We say $\mPhi\in\Re^{M\times N}$ a frame of $\Re^{M}$ if there are positive constants $c_1$ and $c_2$ such that $0< c_1 \leq c_2<\infty$ and for any $\vy\in \Re^{M}$, $c_1\|\vy\|_2^2 \leq \|\mPhi^{\cal T} \vy\|_2^2\leq c_2\|\vy\|_2^2$. A frame is a $c$-tight frame if $c_1 = c_2=c$, which implies $\mPhi\mPhi^{\cal T} = c \mId_M$. A unit tight frame corresponds to $c=1$.}  
\end{itemize}

\vspace{0.25cm}
Let $(\mPhi, \mPsi)$ be any CS system and $(\overline\mPhi, \mPsi)$ be the CS system, where the sensing matrix $\overline\mPhi = (\mPhi \mPhi^{\cal T})^{-1/2}\mPhi$. It is interesting to note that for the same original signal $\vx = \mPhi \vs$,  both $(\mPhi, \mPsi)$ with $\vy =\mPhi \vx$ and $(\overline\mPhi, \mPsi)$ with $\bar \vy = \bar \mPhi \vx$ yield exactly the same estimate of $\vs$ under the PRS (\ref{eq:MLE for s}). In this sense, designing $(\mPhi, \mPsi)$ is totally equivalent to designing $(\overline\mPhi, \mPsi)$ of a structure $\overline\mPhi = (\mPhi \mPhi^{\cal T})^{-1/2}\mPhi$. Furthermore, note that for any full row rank matrix $\mPhi$
\e
\overline \mPhi \overline \mPhi^{\cal T} = (\mPhi\mPhi^{\cal T})^{-1/2}\mPhi \mPhi^{\cal T}(\mPhi\mPhi^{\cal T})^{-1/2} = \mId_M.
\nonumber \ee 

Therefore, one concludes that under the PRS (\ref{eq:MLE for s}), designing a CS system $(\mPhi, \mPsi)$ with the sensing matrix $\mPhi$ searched within the set of full row rank sensing matrices is totally equivalent to designing a CS system $(\mPhi, \mPsi)$ with the sensing matrix $\mPhi$ searched within the set of all unit tight frames, that is~\footnote{With $\mPhi$ constrained by (\ref{eq:tilde Phi tight frame}), $\bar \mPhi =(\mPhi\mPhi^{\cal T})^{-1/2}\mPhi =\mPhi$. Therefore, $\bar \mPhi$ will be no longer used in the sequel.}
\e \mPhi \in \calO_{M,N}. \label{eq:tilde Phi tight frame}\ee
In the optimal CCS design problem to be formulated in the next subsection, the PRS (\ref{eq:MLE for s}) is assumed to be used and hence the optimal sensing matrix $\mPhi$ will be searched under the constraint  \eqref{eq:tilde Phi tight frame}. Note that the unit tight frames have already been utilized for CS. In particular, it has been proved that when the sensing matrix $\mPhi$ is obtained by randomly selecting a number of rows from an orthonormal matrix, it satisfies the RIP under certain condition \cite{rauhut2010compressive}. A typical example is the random partial Fourier matrix~\cite{CandesRombergTao2006RobustUncertaintyPrinciples}. Instead of randomly selecting a number of rows from an orthonormal matrix, we attempt to learn the unit tight frame adaptively to the training data.

\subsection{Learning CCS systems - problem formulation}
Our objective here is to jointly design the sensing matrix and dictionaries for CCS systems with training samples $\underline\mX$. To this end, a proper measure is needed.

First of all, it follows from \eqref{eq:signal model} that
\e
\mX_i = \widetilde\mPsi_i \widetilde \mS_i + \mE_i,
\label{eq:X_i = Phi S + E}\ee
where $\widetilde \mPsi_i$ is the underlying dictionary and $\|\widetilde \mS_i(:,j)\|_0 \leq \kappa$. A widely utilized objective function for dictionary learning is
\e
\varrho_1(\mPsi_i,\mS_i,\mX_i) := \|\mX_i - \mPsi_i \mS_i\|_F^2,
\label{eq:def rho_1}\ee
which represents the variance of the sparse representation error $\mE_i$, where $\|\cdot\|_F$ denotes the {\it Frobenius} norm, and is utilized in state-of-the-art dictionary learning algorithms like the MOD~\cite{engan1999MOD} and the K-SVD~\cite{aharon2006KSVD}.
Similarly, denote
\e
\varrho_2(\mPhi,\mPsi_i,\mS_i,\mX_i) := \|\mPhi (\mX_i - \mPsi_i \mS_i)\|_F^2
\label{eq:def rho_2}\ee
as the variance of the projected signal noise $\mPhi \mE_i$.


It follows from \eqref{eq:X_i = Phi S + E} that the measurements are of the form $\mY_i = \mPhi \mX_i = \mPhi \widetilde\mPsi_i \widetilde \mS_i + \mPhi \mE_i$. Thus besides reducing the projected signal noise variance $\varrho_2$, the sensing matrix $\mPhi$ is expected to sense most of the key ingredients  $\widetilde\mPsi_i \widetilde \mS_i$. This can be done by choosing the sensing matrix $\mPhi$ such that $
\|\mPhi \widetilde\mPsi_i \widetilde \mS_i\|_F^2 $
is maximized. As $\mPhi \mPhi^{\cal T} = \mId_M$ is assumed (see (\ref{eq:tilde Phi tight frame})), maximizing $
\|\mPhi \mPsi_i \mS_i\|_F^2 $  in terms of $\mPhi$ is equivalent to minimizing
\e
\varrho_3(\mPhi,\mPsi_i,\mS_i) := \|(\mId_N - \mPhi^{\cal T}\mPhi) \mPsi_i \mS_i\|_F^2.
\label{eq:def rho_3}\ee
Therefore, the following measure is proposed for each CS in the proposed CCS system:
\eqn \varrho(\mPhi,\mPsi_i,\mS_i,\mX_i) &:=&  \varrho_1(\mPsi,\mS_i,\mX_i) + \alpha  \varrho_2(\mPhi,\mPsi_i,\mS_i,\mX_i)+ \beta \varrho_3(\mPhi,\mPsi_i,\mS_i),\label{the measure}\eeqn
where $\varrho_1(.)$, $\varrho_2(.)$ and $\varrho_3(.)$ are the measures defined in \eqref{eq:def rho_1} - \eqref{eq:def rho_3}, respectively, and $\alpha$ and $\beta$ are the weighting factors to balance the importance of the three  terms.

The problem of designing optimal CCS systems is then formulated as
\e\begin{split}
 &(\widetilde \mPhi, \{\widetilde \mPsi_i, \widetilde \mS_i\}):= \argmin_{\mPhi,\{\mPsi_i, \mS_i\}}  \sum_{i=1}^K \varrho(\mPhi,\mPsi_i,\mS_i,\mX_i),\\
 &\text{s.t.}~~\mPhi\in\calO_{M,N},~\|\mS_i\|_{\infty}\leq b,~~ \| \mS_i(:,j)\|_0\leq \kappa,~\mPsi_i(:,\ell)\in\calU_N, \ \forall \ i,j,\ell,
 \end{split} \label{eq:CS based MLE}\ee
in which both the sensing matrix $\mPhi$ and the bank of dictionaries $\{\mPsi_i\}$ are jointly optimized using the identical measure. Here, $\|\mS_i\|_\infty$ denotes the largest (in absolute value)  entries  in $\mS_i$  and $\|\mS_i\|_{\infty}\leq b$ defines the set of stable solutions~\cite{bao2016dictionary}. We can set $b$ arbitraily large to make the set $\{\mS_i:\|\mS_i\|_{\infty}\leq b\}$ contain the desired solutions for any practical application. We also note that such a constraint assists in the convergence analysis in Section~\ref{sec-convergence} since it results in well bounded solutions.

\section{Algorithms for Designing Optimal CCS Systems}\label{sec:algorithm}
We now propose an algorithm based on alternating minimization for solving \eqref{eq:CS based MLE}. The basic idea is to construct a sequence $\{\mPhi_k,\{\mPsi_{i,k}\},\{\mS_{i,k}\}\}$ such that $\{\sum_{i=1}^K \varrho(\mPhi_k,\mPsi_{i,k},\mS_{i,k},\mX_i)\}$ is a decreasing sequence.

First of all, note that $\varrho(\mPhi,\mPsi_i,\mS_i,\mX_i)$ can be rewritten as
\e \varrho(\mPhi,\mPsi_i,\mS_i,\mX_i) =\left\|\mC(\mPhi,\mX_i) -\mB(\mPhi)\mPsi_i \mS_i\right\|^2_F,  \label{eq:varrho together}\ee
where
 \e \mC(\mPhi,\mX_i) := \MAT{c}\mX_i\\ \sqrt{\alpha}\mPhi \mX_i\\ {\bf 0} \mat,\ \mB(\mPhi):= \MAT{c} \mId_N\\ \sqrt{\alpha}\mPhi \\ \sqrt{\beta}(\mId_N - \mPhi^{\cal T}\mPhi) \mat.\label{AB-def}\ee

The proposed algorithm is then outlined in {\bf Algorithm~\ref{alg:main}}.\\

\begin{algorithm}[htb]
	\caption{Design an optimal CCS system}
	\label{alg:main}
\noindent{\it Input}: the training data $\{\mX_i\}$, the dimension of the projection space $M$, the number of atoms $L_0$, the number of iterations $N_{ite}$, and the parameters $\nu_1,\nu_2, \nu_3$ and $\nu_4$.\\

\noindent{\it Initialization}: set initial dictionaries $\{\mPsi_{i,0}\}$, say each dictionary can be chosen as a DCT matrix or randomly selected from the data, initial sparse coefficient matrixs $\mS_{i,0}$ by solving the sparse coding (i.e., minimizing $\varrho_1$ given by \eqref{eq:def rho_1} with $\mPsi_i = \mPsi_{i,0}$ ) for all $i\in[K]$. \\

\noindent{\it Begin $k=1,2, \cdots, N_{ite}$}
\begin{itemize}
\item Update $\mPhi$ with $\mPsi_{i} = \mPsi_{i,k-1}, \mS_i = \mS_{i,k-1}, \forall~i$ using the proximal method:
\e \mPhi_k = \argmin_{\mPhi\in\calO_{M,N}} \sum^K_{i=1} \varrho(\mPhi,\mPsi_{i},\mS_{i}, \mX_i) + \nu_1 \|\mPhi^\T\mPhi - \mPhi_{k-1}^\T\mPhi_{k-1}\|_F^2, \label{general-ZhuZHx}\ee
which, as to be seen,  can be solved analytically using Lemma~\ref{lem:update Phi}.
\item Update $\mPsi_{i}$ with $\mPhi = \mPhi_{k}, \mS_i = \mS_{i,k-1}, \forall~i\in[K]$ using the proximal method:
\e \mPsi_{i,k} = \argmin_{\mPsi(:,l)\in\calU_N,\forall l}\varrho(\mPhi,\mPsi,\mS_{i}, \mX_i)+ \nu_2 \|\mPsi - \mPsi_{i,k-1}\|_F^2.\label{eq:update Psi prob} \ee
An algorithm will be given later (see {\bf Algorithm} \ref{alg:update Psi}) for solving such a problem.
\item Update $\mS_i$ with $\mPhi = \mPhi_{k}, \mPsi_i = \mPsi_{i,k}, \ \forall i$: Note that \eqref{eq:CS based MLE} becomes
 \e\begin{split}
 &\mS_{i,k}=\argmin_{\mS}  \varrho(\mPhi,\mPsi_i,\mS, \mX_i), \ \text{s.t.}\ \|\mS\|_\infty\leq b, \ \|\mS(:,j)\|_0\leq \kappa, ~\forall~j.
    \end{split}\label{eq:OMP for Sx}\ee
 As to be seen below, such a problem can be addressed using an OMP-based procedure that is parametrized by two positive constants $\nu_3$ and $\nu_4$.
\end{itemize}
\noindent{\it End }\\
\noindent{\it Output:}  $\widetilde \mPhi=\mPhi_{N_{ite}}$ and $\{\widetilde \mPsi_i=\mPsi_{i,N_{ite}}\}$.\\

\end{algorithm}

\noindent{\bf Remark~4.1:}
\begin{itemize}
\item The proximal terms $\nu_1 \|\mPhi^\T\mPhi - \mPhi_{k-1}^\T\mPhi_{k-1}\|_F^2$ and $\nu_2 \|\mPsi - \mPsi_{i,k-1}\|_F^2$ in \eqref{general-ZhuZHx} and \eqref{eq:update Psi prob} are utilized to improve the convergence of the alternating-minimization based algorithms~\cite{attouch2010proximal}. Here the term $\|\mPhi^\T\mPhi - \mPhi_{k-1}^\T\mPhi_{k-1}\|_F^2$ is slightly different than the classical choice which suggests to use $\|\mPhi - \mPhi_{k-1}\|_F^2$. As illustrated in Lemma~\ref{lem:update Phi}, with this term $\|\mPhi^\T\mPhi - \mPhi_{k-1}^\T\mPhi_{k-1}\|_F^2$, we have a closed-form solution to \eqref{general-ZhuZHx}.
\item As seen from the outline of {\bf Algorithm~\ref{alg:main}}, all the dictionaries $\{\mPsi_i\}$ (and the coefficients $\{\mS_i\}$) can be updated concurrently. Thus, utilizing parallel computing strategy, it takes the same time to design a CCS system as that to design a traditional CS system though such efficiency is not an issue for off-line design.\\
\end{itemize}

The proposed {\bf Algorithm~\ref{alg:main}} consists of three minimization problems, specified by \eqref{general-ZhuZHx}, \eqref{eq:update Psi prob} and  (\ref{eq:OMP for Sx}), respectively. In the remainder of this section, we will provide a solver to each problem, while the performance analysis of the whole algorithm will be given in next section.

\subsection{Updating sensing matrix and sparse representations}\label{updating sensing and S}
First of all, let us consider updating the sensing matrix, that is to solve \eqref{general-ZhuZHx}. The solution is given by the following lemma:
\begin{lemma}\label{lem:update Phi} Define $\widetilde \mG :=  \sum_{i=1}^K \mG(\mPsi_{i,k-1},\mS_{i,k-1},\mX_i) +  2\nu_1 \mPhi_{k-1}^\T\mPhi_{k-1}$, where
\eqn
\mG(\mPsi,\mS,\mX) &:=&  \beta \mPsi \mS(\mPsi \mS)^{\cal T} - \alpha (\mX - \mPsi \mS)(\mX - \mPsi \mS)^{\cal T}.
\label{eq:G}\eeqn 
Let $\widetilde \mG = \mV_{\tilde \mG} {\bf \Pi} \mV_{\tilde \mG}^{\cal T}$ be an eigen-decomposition (ED) of the $N\times N$ symmetric $\widetilde \mG$, where ${\bf \Pi}$ is diagonal and its diagonal elements $\{\pi_n\}$ are assumed to satisfy $\pi_1\geq \pi_2 \geq \cdots \geq \pi_N$.
Then the solution to (\ref{general-ZhuZHx})
is given by\footnote{Note that to simplify the notations, we omit the subscript $k$ in $\mPsi_{i,k}$ and $\mS_{i,k}$.}
\e \mPhi_k   = \mU\MAT{cc} \mId_M& {\bf 0}\mat \mV_{\tilde \mG}^{\cal T} \label{optimal-T-1x}\ee
with $\mU\in\calO_M$ arbitrary.
  \end{lemma}

The proof is given in Appendix~\ref{sec:update Phi}.\\

We now consider (\ref{eq:OMP for Sx}).  Under certain conditions (like the mutual coherence or the RIP) on $\mB(\mPhi_k)\mPsi_{i,k}$, the OMP algorithm is guaranteed to solve \eqref{eq:OMP for Sx} exactly~\cite{davenport2010analysis,tropp2004greed} but these conditions are not always satisfied as the iteration goes. Here, we propose a numerical procedure for updating $\{\mS_i\}$ that can ensure the cost function of (\ref{eq:OMP for Sx}) decreasing. This is crucial to guarantee convergence of the proposed  {\bf Algorithm~\ref{alg:main}}.

Denote by $\calS_{\kappa,b}$ the set of $\kappa$-sparse signals that has bounded energy:
\e
\calS_{\kappa,b} = \left\{\vs\in\Re^{L_0}: \|\vs\|_{\infty}\leq b,~~ \| \vs\|_0\leq \kappa \right\}.
\label{eq:calS kappa}\ee
and $\calP_{\calS_{\kappa,b}^{J_0}}[\mS]$ as the orthogonal projection of $\mS \in\Re^{L_0\times J_0}$ onto the set $\calS_{\kappa,b}^{J_0}$.\footnote{The orthogonal projection $\calP_{\calS_{\kappa,b}^{J_0}}$ acts as keeping the largest (in absolute value) $\kappa$ elements for each column of the input matrix and then truncates the entries to $\pm b$ if they are larger (in absolute value) than $b$.
}

Let $\widetilde \mS_{i,k}$ denote the solution obtained by the OMP algorithm~\cite{mallat1993matchingPursuit} and $\zeta(\widetilde \mS_{i,k},\mS_{i,k-1}) := \varrho(\mPhi_{k},\mPsi_{i,k},\mS_{i,k-1},\mX_i) - \varrho(\mPhi_{k},\mPsi_{i,k},\widetilde\mS_{i,k},\mX_i)$. Then, the sparse representations are updated with
\e \mS_{i,k}=\left\{ \begin{array}{l}\widetilde \mS_{i,k},\quad ~~~~~~~~~~~~~~~~~~~~ \text{if} \quad  \zeta(\widetilde \mS_{i,k},\mS_{i,k-1})\geq \nu_3 \|\mS_{i,k-1} - \widetilde\mS_{i,k}\|_F^2 \ \& \ \|\widetilde \mS_{i,k}\|_\infty\leq b,\\
 \calP_{\calS_{\kappa,b}^{J_0}}\left[\mS_{i,k-1} - \nu_4 \nabla_{\mS_i} \varrho(\mPhi_{k},\mPsi_{i,k},\mS_{i,k-1},\mX_i)\right], ~~~~~\text{otherwise}, \end{array}\right.\label{eq:update S 1}\ee
where  $\nabla_{\mZ} f$ denotes the gradient of function $f$ {\em w.r.t.} variable $\mZ$.

The basic idea behind algorithm (\ref{eq:update S 1}) is to ensure that $\mS_{i,k}$ makes the objective function decrease, that is $\zeta( \mS_{i,k},\mS_{i,k-1})>0$,  which is actually a standard requirement for minimization algorithms. As seen from (\ref{eq:update S 1}), if the OMP solution can not meet such requirement the projected gradient descent algorithm~\cite{blumensath2008iterativeThresholding} is utilized to attack \eqref{eq:OMP for Sx}.\footnote{In practical implementations, once the OMP solution can not meet such requirement (i.e., makes the objective function decrease) for successive two times, one can directly utilize the projected gradient descent algorithm to attack \eqref{eq:OMP for Sx} for all the following iterations to reduce the computational complexity.}
A sufficient condition for the gradient descent algorithm to produce a $\mS_{i,k}$  that yields a smaller objective value than the previous one is to choose $\nu_4$ such that~\cite{bolte2014proximal}
\e  \delta(\mS,\mS'):=
\|\nabla_{\mS}\varrho(\mPhi,\mPsi,\mS,\mX) - \nabla_{\mS}\varrho(\mPhi,\mPsi,\mS',\mX)\|_F< \frac{1}{\nu_4}\|\mS - \mS'\|_F
\label{nu-4}\ee
for all $\mPhi\in\calO_{M,N}, \mPsi\in\calU_N^{L_0},\mS,\mS'\in\calS_{\kappa,b}^{J_0}$. To obtain a further simpler guidance for $\nu_4$, we derive an upper bound on the Lipschitz constant regarding the partial gradient $\nabla_{\mS}\varrho(\mPhi,\mPsi,\mS,\mX)$. Towards that end, it follows from (\ref{the measure}) that
\[\delta(\mS,\mS')= 2\left\|\mPsi^\T[\mId + \alpha\mPhi^\T\mPhi + \beta(\mId - \mPhi^\T\mPhi)]\mPsi (\mS - \mS')  \right\|_F,\]
which leads to
\e\begin{split}
\delta(\mS,\mS') &\leq 2 \left\|\mPsi^\T[\mId + \alpha\mPhi^\T\mPhi + \beta(\mId - \mPhi^\T\mPhi)]\mPsi  \right\|_2\left\|\mS - \mS'\right\|_F\\
& \leq 2(1+\alpha + \beta)\left\|\mPsi^\T \mPsi\right\|_2\left\|\mS - \mS'\right\|_F \leq 2(1+\alpha + \beta)L_0  \left\|\mS - \mS'\right\|_F,
\end{split}\label{eq:derive S}\ee
where the second inequality utilizes the fact $\mPhi\in\calO_{M,N}$ and
the last inequality follows because $\mPsi\in\calU_N^{L_0}$. Thus, $\nu_4$ can be simply chosen as
\e \nu_4< \frac{1}{2(1+\alpha + \beta)L_0}:= \frac{1}{L_{cs}}.\label{eq:Lipschitz for S}\ee

\noindent{\bf Remark~4.2}: {\bf Algorithm}~\ref{alg:main} is derived for addressing the optimal CCS problem (\ref{eq:CS based MLE}). There are four parameters $\nu_1, \nu_2, \nu_3$ and $\nu_4$ involved. As to be seen in Section~\ref{sec:convergence}, such an algorithm is ensured to be strictly convergent with a setting, in which the first three  can be chosen as arbitrary positive numbers and $\nu_4$ is set according to (\ref{eq:Lipschitz for S}). For a fixed $\nu_4$, it has been observed from simulations that picking small values for $\nu_1, \nu_2$ and $ \nu_3$ assists in the convergence speed of the algorithm and when they are smaller than a certain value (say $10^{-4}$ as the case in our experiments), the effect of these parameters on the algorithm's performance is not significant. Furthermore, the choice for $\nu_4$, as the step-size of the gradient descent algorithm, is important to ensure the convergence of the proposed  {\bf Algorithm}~\ref{alg:main}. Starting from $\mS_{i,k-1}$, the optimal $\nu_4$, which results in a $\mS_{i,k}$ that minimizes $\varrho(\mPhi_k,\mPsi_{i,k}, \mS, \mX_i)$, depends on the cost function $\varrho$ and the previous point $\mS_{i,k-1}$.  Such an optimal step-size is usually difficult to obtain and a practical strategy is to use line search methods such as the backtracking line search to select an appropriate step-size~\cite{nocedal2006numerical}. Here we utilize a simple strategy in (\ref{eq:Lipschitz for S}) to choose an appropriate step-size which is enough to guarantee the decrease of the objective function (see \eqref{eq:sufficient decrease S 2} in Appendix~\ref{appendix B}) and hence the convergence of the proposed  {\bf Algorithm}~\ref{alg:main}, though it is not the optimal one.  Moreover, as seen from (\ref{eq:Lipschitz for S}), this upper bound $\frac{1}{L_{cs}}$ does not depend on the previous point $\mS_{i,k-1}$  and is just determined by the parameters $L_0$ (the number of atoms in a dictionary $\mPhi_i$) and $(\alpha, \beta)$. We know that $\alpha$ and $\beta$ are the two weighting factors in the cost function  $\varrho$ and they are basically determined empirically.  Some discussions on how to choose the two parameters are given in Section~\ref{sec:experiments}. Given enough computational resources, using a linear search method to addaptively select a better step-size $\nu_4$ for each iteration is expected to increase the convergence of {\bf Algorithm}~\ref{alg:main}. We defer this to our future work.

\subsection{An iterative algorithm for \eqref{eq:update Psi prob}}
\label{sec:update Psi prob}
We now provide an iterative algorithm for \eqref{eq:update Psi prob}. It can be shown with (\ref{eq:varrho together}) that \eqref{eq:update Psi prob} consists of a set of constrained minimizations of form
\e
\widetilde\mPsi = \argmin_{\mPsi(;,\ell)\in {\cal U}_N, \forall~\ell}\{ f(\mPsi):=\|\mC - \mB \mPsi \mS\|_F^2 + \nu_2\|\mPsi - \mPsi'\|_F^2\},
\label{eq:update Psi prob 2}\ee
where $\mC, \mB$ and $\mS$ have nothing to do with $\mPsi$.

The idea behind the algorithm can be explained as follows. Assume that the first $\ell-1$ columns of $\mPsi$ have been updated. Rewrite the objective function as
\begin{align}
f(\mPsi) &= \|\mC - \mB\sum_{j\neq \ell} \vpsi_j \vs_j^\T  - \mB\vpsi_\ell \vs_\ell^\T\|_F^2 + \nu_2\|\mPsi - \mPsi'\|_F^2.
\label{eq:update Psi prob column}\end{align}

Let $\mE_\ell: = \mC - \mB\sum_{j\neq \ell} \vpsi_j \vs_j^\T$. Then minimizing $f$ in terms of $\vpsi_\ell$ is equivalent to solving
\e
\widetilde \vpsi_\ell = \argmin_{{\vpsi\in {\cal U}_N}}\|\mE_\ell - \mB \vpsi \vs_\ell^{\cal T}\|_F^2 + \nu_2 \|\vpsi - \vpsi_\ell'\|^2.
\label{eq:ls 1}\ee
With some manipulations, it can be shown that the above objective function is equal to
\begin{align*}
\vpsi^\T (\vs_\ell^\T \vs_\ell \mB^\T \mB + \nu_2\mId)\vpsi - 2(\vs_\ell^\T \mE_\ell^\T \mB + \nu_2\vpsi_\ell'^{\T}) \vpsi + \|\mE_\ell\|_F^2 + \nu_2 \|\vpsi_\ell'\|^2,
\end{align*}
which implies \eqref{eq:ls 1} is equivalent to
\e
\widetilde \vpsi_\ell = \argmin_{\vpsi\in {\cal U}_N} \{\vpsi^\T \mD_\ell\vpsi - 2  \vd_\ell^{\cal T} \vpsi \},
\label{eq:ls for psi}\ee
where $\mD_\ell:= \|\vs_\ell\|^2_2 \mB^\T \mB + \nu_2\mId$ is a positive semi-definite (PSD) matrix  and $\vd_\ell:= \mB^\T \mE_\ell \vs_\ell + \nu_2\vpsi_\ell'$.

As to be seen in the next subsection, the constrained least square problem (\ref{eq:ls for psi}) can be solved efficiently using an algorithm, denoted as {\bf Alg$_{CQP}$}.

Based on the developments given above, we propose {\bf Algorithm \ref{alg:update Psi}} for addressing \eqref{eq:update Psi prob}.

\begin{algorithm}[htb]
	\caption{Update dictionary $\mPsi_i$ via solving \eqref{eq:update Psi prob}}
\noindent{\it Initialization:}  set $\mC = \mC(\mPhi_{k-1}, \mX_i), \mB=\mB(\mPhi_{k-1})$, $\mS = \mS_{i,k}$, and $\mPsi =\mPsi_{i,k-1}$.

\noindent{\it Begin $\ell=1,2, \cdots, L_0$}, update the $\ell$-th column of $\mPsi$ by
\begin{itemize}
\item Compute the overall representation error by
\[
\mE_\ell = \mC - \mB\sum_{j\neq \ell} \vpsi_j \vs_j^{\cal T}
\]
and $\mD_\ell$ and $\vd_\ell$ as
\begin{align*}
\mD_\ell =  \|\vs_\ell\|^2_2\mB^\T \mB + \nu_2\mId, \ \vd_\ell = \mB^\T \mE_\ell \vs_\ell + \nu_2\vpsi_\ell.
\end{align*}
\item Update the $\ell$-th column by solving \eqref{eq:ls for psi} with {\bf Alg$_{CQP}$} and setting $\mPsi(:,\ell)=\widetilde \vpsi_\ell$.
\end{itemize}
\noindent{\it End }\\
\noindent{\it Output}: $\mPsi_{i,k}=\mPsi$.
	\label{alg:update Psi}\end{algorithm}

\subsection{Algorithm {\bf Alg$_{CQP}$}}
As  seen, (\ref{eq:ls for psi})  requires to solve a set of  constrained quadratic programmings of form
\e
\vh_{opt}:= \argmin_{\vh \in \cal U_N} \{ \xi(\vh):= \vh^\T \mD \vh  - 2 \vd^{\cal T} \vh \},
\label{eq:ls 2}\ee
where $\mD\in \Re^{N\times N}$ is a PSD matrix and $\vd \in \Re^N$, both are not a function  of $\vh$. Let $\mD = \mU\mSigma \mU^\T$ be an ED of $\mD$ where $\mSigma$ is an $N\times N$ diagonal matrix with diagonals ordered as $\sigma_1 \geq \cdots \sigma_N \geq 0$, and $\mU\in\calO_N$. Note that when $\vd = \bf 0$, \eqref{eq:ls 2} is equivalent to finding the smallest eigenvector of $\mD$  and hence the solution to (\ref{eq:ls 2}) is $\vh = \mU(:,N)$. Thus $\vd \neq \bf 0$ is assumed in the sequel. Unlike the unconstrained one, due to  the unit-norm constraint, in general there is no closed-form solution to \eqref{eq:ls 2}. 

Denote the Lagrange function ${\cal L}(\vh,\lambda)$ of problem \eqref{eq:ls 2} as
\[
{\cal L}(\vh,\lambda) = \xi(\vh) - \lambda (\vh^\T \vh -1).
\]
Then, it follows from the Karush-Kuhn-Tucker (KKT) conditions that any solution to  \eqref{eq:ls 2} should satisfy
\e\begin{split}
&\nabla_{\vh} {\cal L}(\vh,\lambda) = 2((\mD -\lambda \mId_N)\vh - \vd) = 0, \ \vh^\T \vh -1 = 0.
\end{split}\label{eq:KKT 1}\ee

Note that the KKT condition in \eqref{eq:KKT 1} is equivalent to
\e
\begin{split}
&(\mSigma -\lambda \mId_N)\overline \vh = \overline \vd, \ \overline \vh^\T \overline \vh -1 = 0,
\end{split}\label{eq:KKT 2}\ee
where $\overline \vh = \mU^\T\vh $ and $\overline \vd = \mU^\T\vd$.

With some manipulations, one can show that for any solution $(\overline \vh, \lambda)$ to \eqref{eq:KKT 2},  the objective function of \eqref{eq:ls 2} is given by
\e \xi(\vh)  = \sum^N_{n=1}\frac{\overline d_n^2}{\lambda - \sigma_n} + \lambda := \tilde \xi(\lambda) \label{eq:xi(lambda)}\ee
for $\lambda \in {\cal S}_{\lambda}$, where ${\cal S}_{\lambda}$ is the solution set of
\e g(\lambda) :=  \|\overline \vh\|^2_2 = \sum_{n =1}^N \frac{\overline d_n^2}{(\sigma_n -\lambda)^2}=1.\label{g-func}\ee
Therefore, the optimal $\lambda$ that corresponds to the solution of \eqref{eq:ls 2} is given by
\e \lambda_{opt}:= \argmin_{\lambda \in {\cal S}_{\lambda}}\tilde \xi(\lambda)\label{optima-lambda}\ee
with which  the solution of \eqref{eq:ls 2} can then be obtained from
$$(\mD -\lambda_{opt} \mId_N)\vh_{opt} =\vd.$$

Some properties of the set ${\cal S}_{\lambda}$ are summarized in the following theorem, which are given in \cite{zhu2017TIT}.
\begin{theorem} Let $g(\lambda)$ be defined in (\ref{g-func}) for all $\lambda \in \Re$ such that $\lambda \neq \sigma_\ell, \forall\ \ell$ unless $\overline d_\ell = 0$. Then
\begin{itemize}
\item $g(\lambda)$ is convex on each of the interval $(-\infty,\sigma_N)$, $(\sigma_N,\sigma_{N-1}), \ldots$, $(\sigma_2,\sigma_1),$ and $(\sigma_1,\infty)$.
\item there exist $\cal N$ solutions $\{\lambda_n\}$ for $g(\lambda) = 1$ with $2\leq {\cal N} \leq 2N$. Suppose $\lambda_1 > \lambda_{2}> \cdots > \lambda_{\cal N}$. Then
    \[
    {\underline\lambda}_1:= \max_{\ell}\left\{\sigma_\ell \pm \overline d_\ell  \right\} \leq \lambda_1 \leq {\overline \lambda}_1 := \sigma_N + \sqrt{\sum_{\ell=1}^N \overline d_\ell^2}
    \]
and $g(\lambda)$ is monotonically increasing within $[{\underline\lambda}_1,{\overline\lambda}_1]$.
\item  $\tilde \xi(\lambda)$ (defined in \eqref{eq:xi(lambda)}) is increasing on $\lambda_n$:
\e
\tilde \xi(\lambda_1) \leq \tilde \xi(\lambda_2) \leq \cdots \leq \tilde \xi(\lambda_{\cal N}),
\label{eq:increas xi}\ee
which implies $\lambda_{opt} = \lambda_1$.\\
\end{itemize}
\end{theorem}

As guaranteed by \eqref{eq:increas xi}, $\lambda_{opt} = \lambda_1$. As $g(\lambda)$ is monotonically increasing within $[{\underline\lambda}_1,{\overline\lambda}_1]$, we can find $\lambda_1$ easily via a standard algorithm, say a bi-section-based algorithm. For convenience, the whole procedure for solving the unit-norm constrained quadratic programming \eqref{eq:ls 1} is denoted as {\bf Alg$_{CQP}$}.\\

\noindent{\bf Remark~4.3}: Our proposed Algorithm~\ref{alg:main} is based on the basic idea of alternating minimization that has led to a class of iterative algorithms for designing sensing matrices and dictionaries. The key difference that makes one algorithm different from another lies in the ways how the iterates are updated. It should be pointed out that though the alternating minimization-based algorithms practically work well for nonconvex problems, there is still lack of rigorous analysis for algorithm properties such as convergence.  To the best of our knowledge, a few of results on this issue have been reported \cite{bao2016dictionary,tropp2005designing}.

\section{Performance Analysis}\label{sec:convergence}
In this section, we will analyze the performance of the proposed CCS and {\bf Algorithm~\ref{alg:main}} in terms of implementation complexity and  convergence.
\subsection{Implementation complexity}
We first note that the collaborative CS system has an identical compression stage as the classical CS system (as shown in Fig.~\ref{figure:CS systems}). With respect to the signal recovery stage, in the collaborative CS system the estimators $\{\vx_i\}$ are independent to each other and can be obtained simultaneously by solving \eqref{eq:signal model multi dic} in parallel. Therefore, with the parallel computing strategy, the time needed for the collaborative CS system obtaining the estimates $\{\vx_i\}$ is similar to that for the classical CS system with a single dictionary of size $N\times L_0$, but is cheaper than that for the classical CS system with a single large dictionary of size $N\times KL_0$.

Once the estimates $\{\vx_i\}$ are computed, the rest procedure in the collaborative CS system is to implement $\mOmega \bar \vx$ as shown in \eqref{collaborative-estimator} with $\mOmega \in \Re^{N\times KN}, \bar \vx\in \Re^{KN\times 1}$. When $\mOmega$ is fully-parameterized,  the computation complexity for computing $\mOmega \bar \vx$ is $O(KN^2)$. When model \eqref{collaborative-estimatorII} is used, where $\mOmega_i$ is diagonal for all $i \in [K]$, the computation complexity is $O(KN)$.  It is interesting to note that model \eqref{collaborative-estimatorX} has the same computation complexity as \eqref{collaborative-estimatorII}  in general but  when $\omega_i=1/K, \forall~i$, the computation complexity is $O(N)$ - the simplest model.

We now discuss the complexity of {\bf Algorithm~\ref{alg:main}} for learning the CCS system. Recall that $\mPhi \in \Re^{M\times N}, \mPsi_i \in \Re^{N\times L_0}$ and $\mX_i\in \Re^{L\times J_0}$.  Our algorithm is an iterative method and in each iteration, there are three main steps. They are sensing matrix updating which requires computing $\widetilde \mG$ in Lemma~\ref{lem:update Phi} with $O(KNL_0J_0)$ operations and ED of $\widetilde \mG$ with $O(N^3)$ operations and the implementation complexity in total is of order $O(KNL_0 J_0)$, dictionary updating which requires $O(K (N+M)L_0^2J_0)$ operations, and sparse coding which requires $O(\kappa^2 K(L+N)J)$ operations. Thus, the total cost in running {\bf Algorithm~1} is of order $O(KNL_0^2J_0N_{iter})$ under the assumption of $\kappa^2\leq M\leq N\leq L_0\ll J_0$.

We finally note that {\bf Algorithm~\ref{alg:main}} can be naturally performed with the parallel computing strategy as all the dictionaries $\{\mPsi_i\}$ and the coefficients $\{\mS_i\}$ can be updated concurrently. We can also utilize similar strategies proposed in \cite{koppel2017d4l,raja2016cloud} for distributed dictionary learning to further reduce the computational time by using more computational resources. We defer this to our future work.

\subsection{Convergence analysis}\label{sec-convergence}
To simplify the notations in the analysis, define \[
\underline\mPsi := \begin{bmatrix}\mPsi_{1}&\cdots&\mPsi_{i}&\cdots&\mPsi_{K}\end{bmatrix}, \ \underline\mS := \begin{bmatrix}\mS_{1}&\cdots&\mS_{i}&\cdots&\mS_{K}\end{bmatrix}.\] Then, the objective function for the optimal CCS system design defined in \eqref{eq:CS based MLE}
can be denoted as
\e
\underline\varrho(\mPhi,\underline\mPsi,\underline\mS):= \sum_{i=1}^K \varrho(\mPhi,\mPsi_i,\mS_i,\mX_i),
\label{eq:underline varrho}\ee
where $\{\mX_i\}$ are dropped as they are fixed during the learning process of the CCS system. Furthermore, define a variable $\mW := (\mPhi,\underline\mPsi,\underline\mS)$. Clearly, running {\bf Algorithm~\ref{alg:main}}, we will have a sequence $\{\mW_k\}=\{(\mPhi_k,\underline\mPsi_k,\underline\mS_k)\}$ generated.

Roughly speaking,  convergence analysis of {\bf Algorithm}~\ref{alg:main} is to study the behaviors of the sequences $\{\mW_k\}$ and $\{\underline\varrho(\mPhi_k,\underline\mPsi_k,\underline\mS_k)\}$. In this subsection, we provide a rigorous convergence analysis for {\bf Algorithm~\ref{alg:main}} which show that both sequences are convergent.

The indicator function $I_{{\cal S}}(\vx)$ on set ${\cal S}$ is defined as
\e I_{{\cal S}}(\vx):= \left\{ \begin{array}{lcr} 0,& \text{if}~~~ \vx\in{\cal S},\\
+\infty,&\text{otherwise}.\end{array}\right. \label{indicator function}\ee
One of advantages in introducing such a function is to simplify the notations as with such a function, one can write a constrained minimization as an unconstrained one.

First of all, we present an interesting result on the solution obtained using the OMP algorithm (see {\bf Algorithm~\ref{alg:OMP}} in Appendix~\ref{appendix AB}) for the classic sparse coding \eqref{eq:recov-1}. With the help of indicator function, \eqref{eq:recov-1} is converted into the following unconstrained minimization.
\e
\min_{\vs}g(\vs) := \|\mA\vs - \vy\|^2_2 + I_{\calS_{\kappa}}(\vs),
\label{eq:sparse recovery}\ee
where  ${\cal S}_\kappa := \left\{\vs\in\Re^{L_0}: \|\vs\|_0 \leq  \kappa \right\}$.\\

\begin{lemma}\label{lem:OMP stationary point}  Let $\widehat \vs$ be an estimate of solution to \eqref{eq:sparse recovery}, obtained using the OMP algorithm. Then,  such $\widehat \vs$ is a stationary point of $g(\vs)$ defined in \eqref{eq:sparse recovery}.\\
\end{lemma}
The proof is given in Appendix~\ref{appendix AB}.

This result is believed of importance for convergence analysis of OMP-based algorithms. In fact, it serves as a crucial tool for proving Theorem~\ref{thm:convergence} to be presented below.

Similarly, utilizing the indicator function, we define
\e
f(\mW):=  \underline\varrho(\mPhi,\underline\mPsi,\underline\mS) + I_{{\cal O}_{M,N}}(\mPhi)+ I_{{\cal U}_{N}^{L_0K}}(\underline\mPsi) + I_{\calS_{\kappa,b}^{KJ_0}}(\underline\mS),
\label{eq:main no regularizer}\ee
where $I_{{\cal S}}(\vx)=0$ if $\vx\in{\cal S}$ and $\infty$ otherwise for ${\cal S}= {\cal O}_{M,N},~ {\cal U}_{N}^{L_0K}$, and $\calS_{\kappa,b}^{KJ_0}$ defined in (\ref{eq:calS kappa}). Therefore, the optimal CCS system design \eqref{eq:CS based MLE}, which is a constrained minimization, is equivalent to the following unconstrained problem
\[
\min_{\mW} f(\mW).
\]

Before presenting our results, let us give the following remark. Let $\mW_k=(\mPhi_k,\underline\mPsi_{k},\underline\mS_{k})$ be output generated by our proposed {\bf Algorithm}~\ref{alg:main} at the $k$th iteration.\\

\noindent{\bf Remark~5.1}: We note that as indicated by \eqref{optimal-T-1x}, $\mU\in\calO_N$ can be chosen arbitrarily\footnote{It can also be observed from \eqref{general-ZhuZHx} that any $\mU\in\calO_M$ yields the same objective value.} when updating $\mPhi_k$. We now exploit this degree of freedom to choose a $\mU$ such that $\mPhi_k$ is closest to $\mPhi_{k-1}$. In particular, instead of choosing an arbitrary $\mU$, we let
\e \mPhi_k   = \mU_k\MAT{cc} \mId_M& {\bf 0}\mat \mV_{\tilde \mG}^{\cal T} \label{optimal-T-1x 2},\ee
where $\mV_{\tilde \mG}^{\cal T}$ is the same as in \eqref{optimal-T-1x} and $\mU_k$ is the minimizer of
\[
\min_{\mU\in\calO_M}\left\|\mU\MAT{cc} \mId_M& {\bf 0}\mat \mV_{\tilde \mG}^{\cal T} - \mPhi_{k-1}\right\|_F.
\]
Let $\MAT{cc}\mId_M& {\bf 0}\mat \mV_{\tilde \mG}^{\cal T}\mPhi_{k-1}^\T = \mP\mSigma\mQ^\T$ be an SVD of $\MAT{cc}\mId_M& {\bf 0}\mat \mV_{\tilde \mG}^{\cal T}\mPhi_{k-1}^\T$. Then, $\mU_k = \mQ\mP^\T$. In the sequel, all $\mPhi_k$ in the sequence $\{\mW_k\}$ by {\bf Algorithm}~\ref{alg:main} are assumed to be given by (\ref{optimal-T-1x 2}).\\

Now, we can present the first set of our results in convergence analysis.

\begin{theorem}\label{thm:convergence}(Sub-sequence convergence of {\bf Algorithm~\ref{alg:main}}) Let $\{\mW_k\}$ be the sequence generated by {\bf Algorithm~\ref{alg:main}}. Suppose that the four parameters in  {\bf Algorithm~\ref{alg:main}} are chosen as: $\nu_1>0,\nu_2>0, \nu_3>0$, and $\nu_4< \frac{1}{L_{cs}}$ with $L_{cs}$ given in \eqref{eq:Lipschitz for S}. Then,  the sequence $\{\mW_k\}$ possesses the following properties:
\begin{enumerate}[(i)]
\item There exists a constant $c_1>0$ such that~\footnote{For a variable $\mW$ defined above, $\|\mW\|_F^2 := \|\mPhi\|_F^2+\|\underline\mPsi\|_F^2+\|\underline\mS\|_F^2$.}
\e
f(\mW_{k-1}) - f(\mW_k)\geq c_1 \|\mW_{k-1} - \mW_k\|_F^2,
\label{eq:sufficient decrease}\ee
 and $\{\mW_k\}$ is regular, i.e.,
\e
\lim_{k\rightarrow \infty} \|\mW_{k-1} - \mW_k\|_F = 0.
\label{eq:regular}\ee
\item For any convergent subsequence $\{\mW_{k_m}\}$, its limit point $\mW^\star$ lies in ${\cal O}_{M,N}\times{\cal U}_N^{L_0K}  \times \calS_{\kappa,b}^{KJ_0}$, is a stationary point of $f(\mW)$, and satisfies
\e
f(\mW^\star) = \lim_{k_m\rightarrow \infty} f(\mW_{k_m}) = \inf_{k}f(\mW_k).
\label{eq:limit f}\ee
\end{enumerate}

\end{theorem}

The proof of Theorem \ref{thm:convergence} is given in Appendix~\ref{appendix B}.\\

As $f(\mW)\geq 0$, the first property of Theorem \ref{thm:convergence} implies that $\{f(\mW_k)\}$ is convergent. 
In a nutshell, the sub-sequence convergence property in Theorem \ref{thm:convergence} guarantees that the sequence generated by {\bf Algorithm~\ref{alg:main}} has at least one convergent subsequence whose limit point is a stationary point of $f(\mW)$. Is the generated sequence $\{\mW_k\}$ itself convergent?\\

\begin{theorem}\label{thm:sequence convergence}(Sequence convergence of {\bf Algorithm~\ref{alg:main}}) With the same setup as in Theorem \ref{thm:convergence}, the generated sequence $\{\mW_k\}$  converges to a stationary point of $f(\mW)$ defined in \eqref{eq:main no regularizer}.\\
\end{theorem}

The proof of Theorem \ref{thm:sequence convergence} is given in Appendix~\ref{sec:prf thm sequence convergence}.\\

The main idea for proving Theorem \ref{thm:sequence convergence} is to utilize the geometrical properties of the objective function $f(\mW)$ around its critical points, i.e., the so-called Kurdyka-\L{ojasiewicz} (KL) inequality~\cite{bolte2007lojasiewicz} (also see Appendix~\ref{sec:prf thm sequence convergence}), which has been widely used to show the convergence of the iterates sequence generated by the proximal alternating algorithms \cite{attouch2010proximal,bao2016dictionary,bolte2007lojasiewicz,bolte2014proximal,zhu2018convergence}.\\

\noindent{\bf Remark~5.2}:
\begin{itemize}\item
 It has been an active research topic to develop  alternating minimization-based algorithms that have a guaranteed convergence (even to a stationary point rather than the global minimizer)~\cite{bolte2014proximal}. A sequence convergence guarantee was provided in~\cite{bao2016dictionary} for a dictionary learning algorithm that  use a similar strategy as the 2nd expression in \eqref{eq:update S 1} but does not utilize the OMP algorithm for sparse coding. We note that the OMP algorithm is widely utilized in dictionary learning algorithms~\cite{aharon2006KSVD,engan1999MOD}. Though one can just utilize the 2nd expression in \eqref{eq:update S 1} for sparse coding and the corresponding algorithm is also convergent, we observe from experiments that the algorithm utilizing OMP as in \eqref{eq:update S 1} converges to a much better solution. Aside from this difference, the method for updating the dictionary in {\bf Algorithm~\ref{alg:main}} is also different than the one in~\cite{bao2016dictionary} since the corresponding problem \eqref{eq:update Psi prob} in {\bf Algorithm~\ref{alg:main}} is more complicated. Finally, as we mentioned before, the proximal term $\|\mPhi^\T\mPhi - \mPhi_{k-1}^\T\mPhi_{k-1}\|_F^2$ utilized in \eqref{general-ZhuZHx} is slightly different than the classical proximal method (utilized in \cite{attouch2010proximal,attouch2013convergence,bao2016dictionary,bolte2014proximal}) which suggests to use $\|\mPhi - \mPhi_{k-1}\|_F^2$.
 In \cite{sun2017complete}, the authors showed that the global minimum of complete (rather than overcomplete) dictionary learning is obtainable when the training data is populated according to certain probabilistic distribution;
 \item Our strategy for simultaneously learning the sensing matrix and dictionary for a CS system differs from the one in \cite{duarte2009learning} in that we provide a unified and identical measure for jointly optimizing the sensing matrix and dictionary.  To the best of our knowledge, our {\bf Algorithm~\ref{alg:main}} is the first work that deals with simultaneous learning of the sensing matrix and the dictionary with guaranteed convergence;
     \item In the proof of Theorem~\ref{thm:convergence}, we establish an interesting result presented in  Lemma~\ref{lem:OMP stationary point} (in Appendix~\ref{appendix B}), which states that without any condition imposed on  $\mB(\mPhi_k)\mPsi_{i,k}$,  any solution obtained using the OMP algorithm for  \eqref{eq:OMP for Sx} is a stationary point of the cost function. We believe this will also be very useful to show the convergence of other algorithms for learning the sensing matrix and the dictionary~\cite{aharon2006KSVD,duarte2009learning,li2018joint}. Finally, we point out that the formulated optimal CCS system design is highly nonconvex and the proposed algorithm can not  guarantee that the convergent point is necessarily the global
solution  of problem. It remains an active research area to develop alternating minimization-based numerical algorithms that are convergent and can yield a solution close to a global one. \end{itemize}



\section{Experiments}\label{sec:experiments}

In this section, we will present a series of experiments to examine the performance of  the MLE-based collaborative estimator, the proposed CCS scheme and approach to designing optimal CCS systems.

\subsection{Demonstration of the MLE-based collaborative estimator}\label{sec:exp for MLE}
In this section, we demonstrate the performance of the MLE-based collaborative estimators derived in Section \ref{sec:MLE coll estimator}. We generate a set of $J=1000$ signal vectors $\{\vx_j\}_{j=1}^J$ where each $\vx_j\in\Re^{N}$ (with $N = 20$) is obtained from a Gaussian distribution of i.i.d. zero-mean and unit variance. For each $\vx_j$, we generate $K = 5$ estimators as $\vy_{j,i} = \vx_j + \ve_{j,i}, ~i=1, \cdots, K$, where $\bar \ve_j = \begin{bmatrix} \ve_{j,1}^\T & \cdots \ve_{j,K}^\T \end{bmatrix}^\T$ obeys a normal distribution $\mathcal N(\vzero,  \mGamma)$. Here $\mGamma\in\Re^{NK\times NK}$ is a positive-definite (PD) matrix generated as $\mGamma = \sigma^2 \mGamma_1\mGamma_1^\T$, where each element in $\mGamma_1\in\Re^{NK\times NK}$ is obtained from an i.i.d. normal distribution. Fig.~\ref{Figure:covariance}(a) displays the correlation matrix $\overline\mGamma = \diag(\gamma_{1,1}^{-1/2},\ldots,\gamma_{NK,NK}^{-1/2})\mGamma \diag(\gamma_{1,1}^{-1/2},\ldots,\gamma_{NK,NK}^{-1/2})$ (i.e., the normalized covariance matrix, which is more appropriate to describe correlations between random variables), where $\gamma_{l,l}^{-1/2}$ is the $l$-th diagonal entry of $\mGamma$. As shown in Fig.~\ref{Figure:covariance}(a),  the noises $\{\ve_{j,i}\}^K_{i=1}$ are not statistically independent since the covariance matrix $\mGamma$ is not a block diagonal matrix.

We use Ind$_i$ to denote the performance of the $i$-th estimator (i.e., the averaged energy of noise contained in $\vy_{j,i}$ for all $j\in[J]$). For convenience, we denote by Ind the best performance of the $K$ estimators Ind$_i$ for $i\in[K]$. We now test  three MLE-based collaborative estimators proposed in Section \ref{sec:MLE coll estimator}. In particular, we use MLE$_1$ to denote the one in \eqref{collaborative-estimator}, where $\mOmega_i$ is set with $\mGamma$ using (\ref{LBJ-1}) for all $i\in [K]$, yielding the optimal CE,  MLE$_2$ to denote the one, where $\mOmega_i$ is obtained using (\ref{Omega-independent}) with $\mGamma_i = \mGamma(i,i)$ for all ~$i \in [K]$, and MLE$_3$ to denote the one that simply averaging the $K$ estimator, i.e., $\widehat \vx_j = \frac{1}{K}\sum_{i=1}^K \vy_{j,i}$ which is the one
in \eqref{collaborative-estimatorX} by setting all the weights to $\frac{1}{K}$. Clearly, MLE$_2$ and MLE$_3$ are not optimal in this case. For convenience, we utilize MLE$_{1_{i}}$, MLE$_{2_{i}}$ and MLE$_{3_{i}}$ to denote the corresponding MLE-based collaborative estimators that fuse the first $i$ estimates.

We use the relative mean squared error (RMSE), denoted as $\sigma_{rmse}$, to measure the noise level and the performance of the MLE-based collaborative estimators:
\e
\sigma_{rmse}:=  \frac{1}{J}\sum_{j=1}^J \frac{\|\vx_j - \widehat \vx_j\|_2}{\|\vx_j\|_2}.
\ee

Fig.~\ref{Figure:covariance}(b) shows the RMSE of the MLE-based collaborative estimators when $\sigma$ is varied from 0.001 to 0.1. Fig.~\ref{Figure:exp_vary_k} illustrates the RMSE of the three MLE-based collaborative estimators MLE$_{1_{i}}$, MLE$_{2_{i}}$ and MLE$_{3_{i}}$ and the individual estimators Ind$_i$ when $\sigma = 0.1$.\\

{\noindent \bf Remark 6.1:}
\begin{itemize}
\item We observe from Fig.~\ref{Figure:covariance}(b) and Fig.~\ref{Figure:exp_vary_k} that the MSE$_1$, as expected,  has the best performance among the three MLE-based collaborative estimators, which coincides with Lemma~\ref{lem:collaborative-estimator};
\item  Fig.~\ref{Figure:exp_vary_k} shows that though not optimal, MLE$_{2_{i}}$ and MLE$_{3_{i}}$ have a much better performance than any of the individual estimators. It is also observed that the performance of all the three collaborative estimators is enhanced with the number of estimators increasing;
    \item It is of interest to see that MLE$_{2_{i}}$ and MLE$_{3_{i}}$ have very similar performance. This suggests that, when the prior information of $\mGamma$ is not available, the simple averaging strategy (i.e., take the mean of all the possible individual estimates) is expected to be a good choice. We will demonstrate the performance of this strategy for collaborative CS systems in the following sections with real images.
\end{itemize}

\begin{figure}[htb!]
\begin{minipage}{0.48\linewidth}
\centerline{
\includegraphics[width=2.8in]{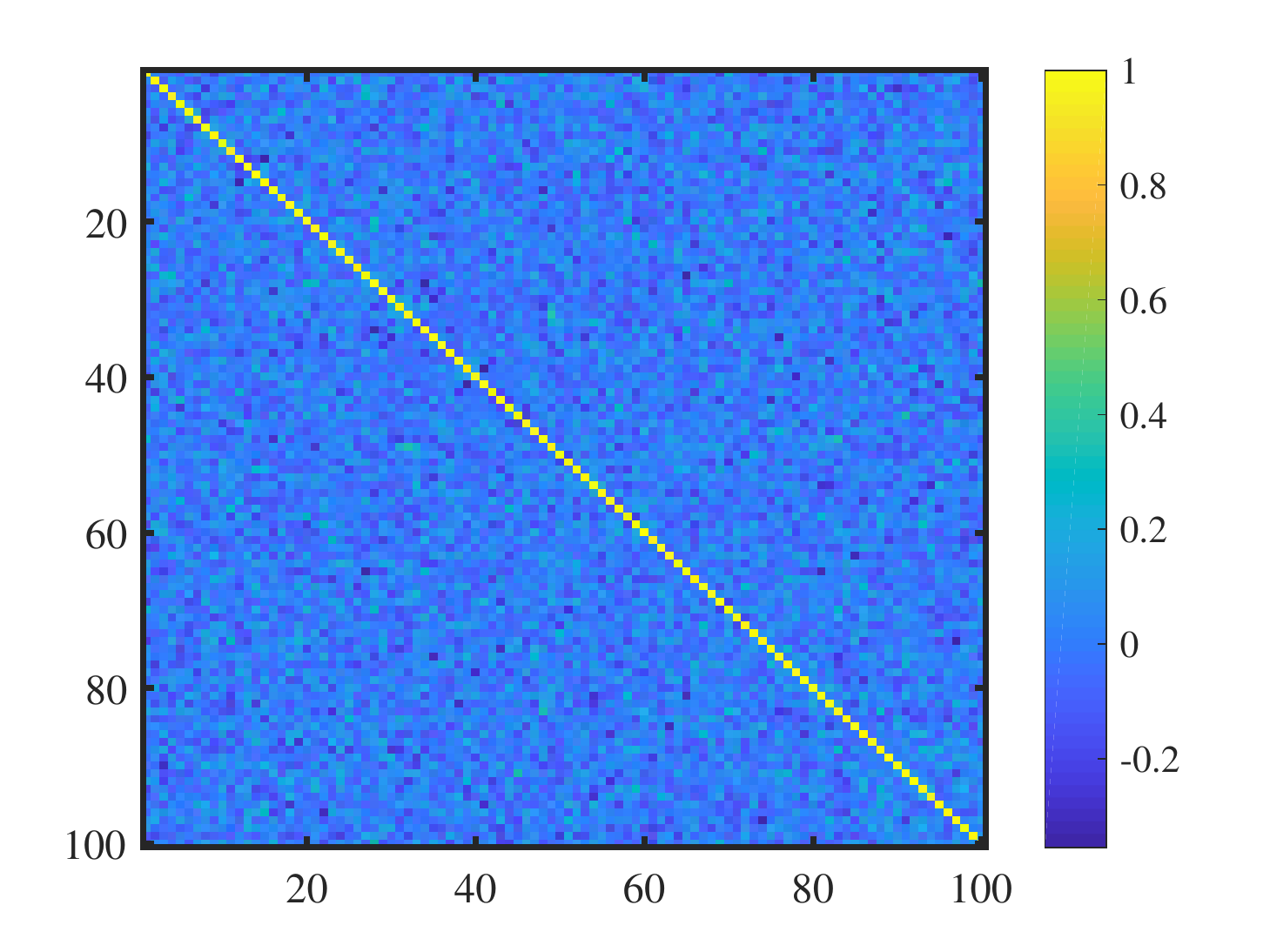}}
\centerline{(a)}
\end{minipage}
\hfill
\begin{minipage}{0.48\linewidth}
\centerline{
\includegraphics[width=2.8in]{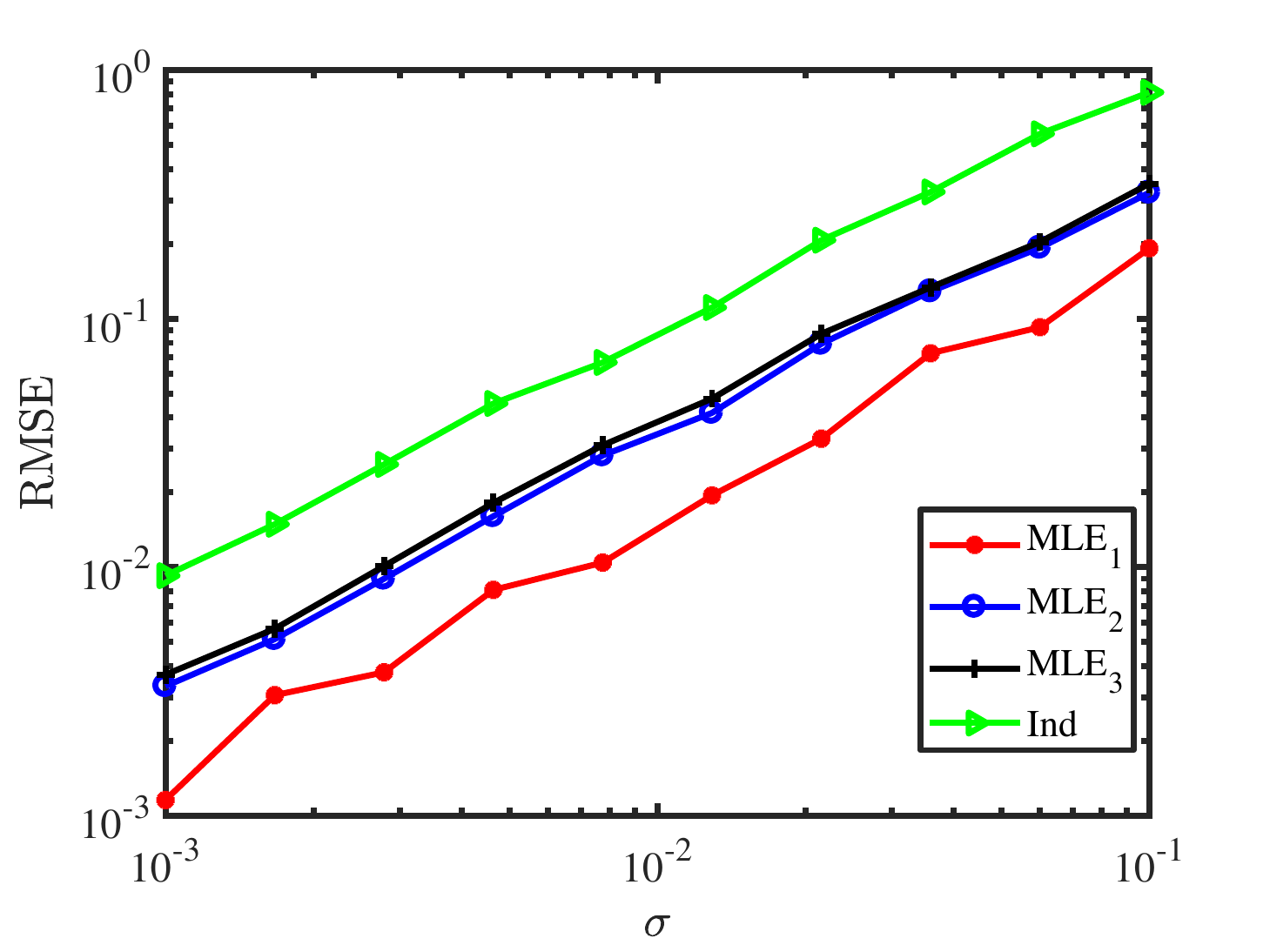}}
\centerline{(b)}
\end{minipage}
\caption{(a) Illustration of the correlation matrix $\overline\mGamma$; (b) RMSE $\sigma_{rmse}$ of the MLE-based collaborative estimators. Here, $K=5$ and $N=20$.} \label{Figure:covariance}
\end{figure}

\begin{figure}[htb!]
\centerline{
\includegraphics[width=4in]{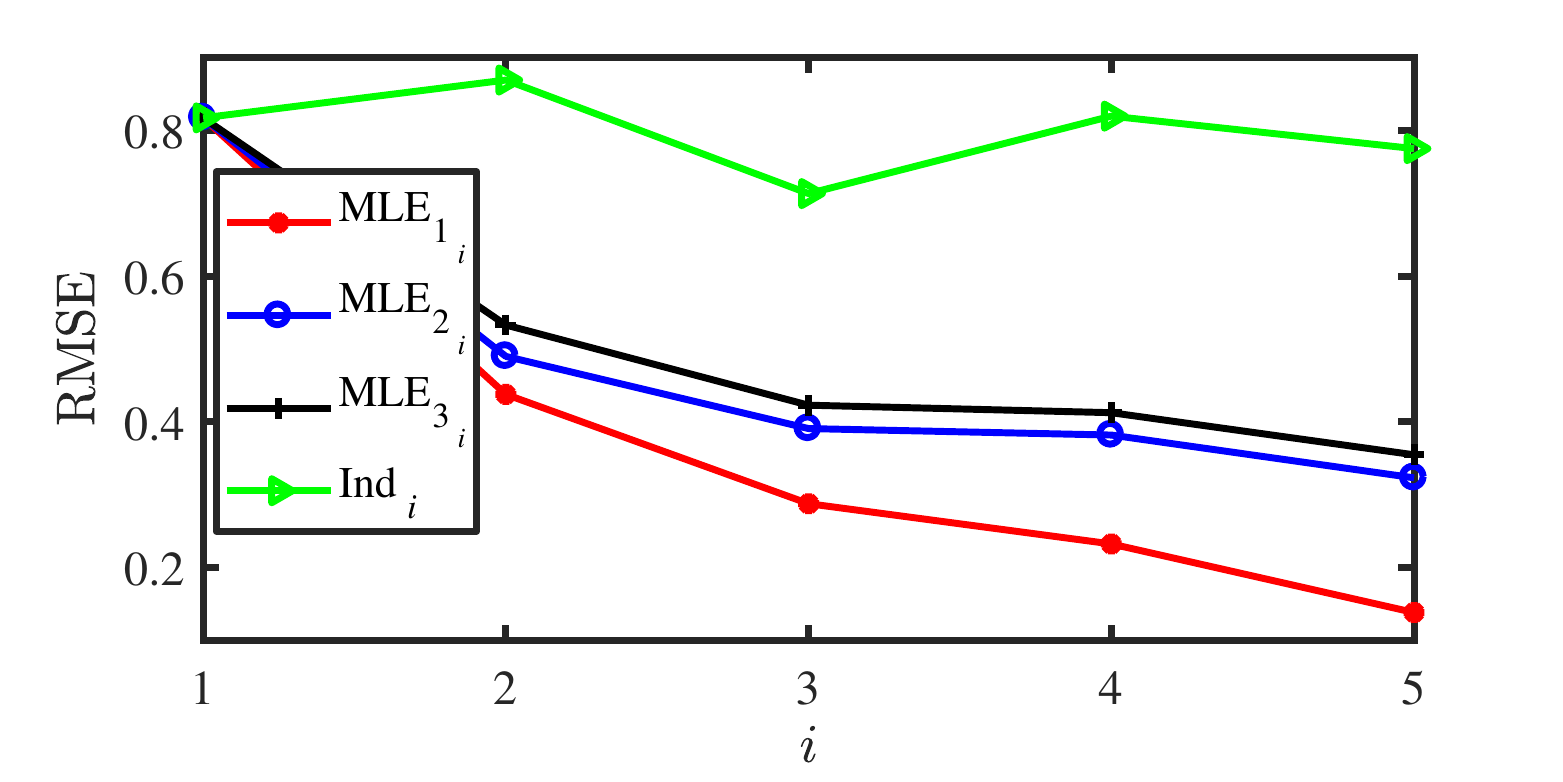}}
\caption{RMSE $\sigma_{rmse}$ of the MLE-based collaborative estimators. Here, $K=5$ and $N=20$.} \label{Figure:exp_vary_k}
\end{figure}

\subsection{Demonstration of the collaborative CS scheme}
In this section, we will demonstrate the performance of the notion of collaborative CS scheme (as shown in Fig.~\ref{figure:CS systems}) and compare it with that of  traditional ones.

Through the experiments for this part, we set the number of dictionaries $K = 5$. Each training data $\mX_i\in\Re^{N\times J_0}$ ($i\in[K]$) is obtained by 1) randomly extracting $15$ non-overlapping patches (the dimension of each patch is $8\times 8$) from each of 400 images in the LabelMe \cite{russell2008labelme} training data set, and 2) arranging each patch of $8\times8$ as a vector of $64\times 1$. Such a setting implies $N=64$ and $J_0=15\times 400=6000$. For each training data set $X_i$, we apply the K-SVD algorithm  to obtain a sparsifying dictionary $\mPsi_{KSVD_i} \in \Re^{N\times L_0}$ with $L_0= 100$ and a given sparsity level $\kappa=4$. We also apply the K-SVD algorithm for the entire training data $\underline \mX=\begin{bmatrix}\mX_1&\cdots&\mX_i&\cdots&\mX_K\end{bmatrix}$ to obtain a sparsifying dictionary $\mPsi_{KSVD} \in \Re^{N\times L}$ with $L=L_0$ and $\mPsi_{\widetilde{KSVD}} \in \Re^{N\times L'}$ with $L'=23^2 = 529$ (such that $L'$ is greater than $KL_0 = 500$).
We then generate a random $M\times N$ sensing matrix $\mPhi$. The CS systems with $(\mPhi,\mPsi_{KSVD_i})$, $(\mPhi,\mPsi_{KSVD})$ and $(\mPhi,\mPsi_{\widetilde{KSVD}})$ are then respectively denoted by $\text{CS}_{Rdm_i}$, $\text{CS}_{Rdm}$ and $\text{CS}_{\widetilde{Rdm}}$. For any image represented by matrix $\mH$, we apply CS system $\text{CS}_{Rdm_i}$ to compress it and use $\mH_i$ to denote the output of each CS system. As demonstrated in the last section, we take the average of each CS system for the collaborative CS system. That is, the collaborative CS system (denoted by $\text{CCS}_{Rdm}$) has the output $\widehat \mH = \frac{1}{K}\sum_{i=1}^K\mH_i$. For convenience, we utilize $\text{CCS}_{Rdm_i}$ to denote the collaborative CS system that fuses the estimates from the first $i$ CS systems, i.e., with the output $\widehat \mH_i = \frac{1}{i}\sum_{m=1}^i\mH_m$. Clearly, $\widehat \mH =\widehat \mH_K$.

The reconstruction accuracy is evaluated in terms of peak signal-to-noise ratio (PSNR), defined
as
\[
\sigma_{psnr} := 10\times \log_{10}\left( \frac{\left(2^r -1\right)^2}{\sigma_{mse}}\right),
\]
where $r=8$ bits per pixel and $\sigma_{mse}$ for images $\mH,\widehat\mH\in\Re^{N_1\times N_2}$ (where $\widehat\mH$ is an estimate of $\mH$) defined as
\[
\sigma_{mse} := \frac{1}{N_1\times N_2}\|\mH - \widehat \mH\|^2_F.
\]

As illustrated in Fig.~\ref{figure:CS systems}, the CS system first divides the input image into a number of non-overlapping $8\times 8$ patches, and then compresses each patch. Once all the patches are recovered, we concatenate them into an image. Fig.~\ref{Figure:CSvsCCS} shows the PSNR $\sigma_{psnr}$ of the CS systems $\text{CS}_{Rdm_i}$ and $\text{CCS}_{Rdm_i}$ when applied to the image `Plane'. Table \ref{table:PSNR with Random Sensing} provides $\sigma_{psnr}$ for the CS systems $\text{CS}_{Rdm_i}$, $\text{CS}_{Rdm}$, $\text{CS}_{\widetilde{Rdm}}$ and $\text{CCS}_{Rdm}$ tested with the twelve images. We also examine the performance with the test data from LabelMe \cite{russell2008labelme}: we randomly extract $15$ non-overlapping patches (the dimension of each patch is $8\times 8$) from each of 400 images in the LabelMe \cite{russell2008labelme} test data set and arrange each patch of $8\times8$ as a vector
of $64\times 1$. Fig.~\ref{Figure:CS random for plane} displays the visual effects of image `Plane' for the CS systems $\text{CS}_{Rdm_i}$ and $\text{CCS}_{Rdm}$. \\

{\noindent \bf Remark 6.2:}
\begin{itemize}
\item It is observed from Fig.~\ref{Figure:CSvsCCS} that the performance of the CCS system gets improved with the number of CS systems fused increasing. This coincides with that observed from the experiments just conducted using synthetic data and the expectations induced from the theoretical results in Section~\ref{sec:MLE coll estimator};
\item Table \ref{table:PSNR with Random Sensing} also demonstrates the advantages of $\text{CCS}_{Rdm}$. It is of great interest to note that $\text{CCS}_{Rdm}$ has still much better performance than $\text{CS}_{Rdm}$ whose dictionary is learned with the same size as those used in fusion and all the training data $\underline \mX$;
\item It is interesting to note from Table \ref{table:PSNR with Random Sensing} that $\text{CS}_{\widetilde{Rdm}}$,   though with a dictionary  $\mPsi_{\widetilde{KSVD}}$ of dimension higher than $KL_0$, even does not outperform  $\text{CS}_{Rdm}$. This, as conjectured before and one of the arguments for the proposed CCS framework, is due to the fact that the mutual coherence of the dictionary $\mPsi$ and hence the equivalent one $\mA$ gets higher when the dimension increases and consequently, the signal reconstruction accuracy decreases.
\end{itemize}

\begin{figure}[htb!]
\centerline{
\includegraphics[width=4in]{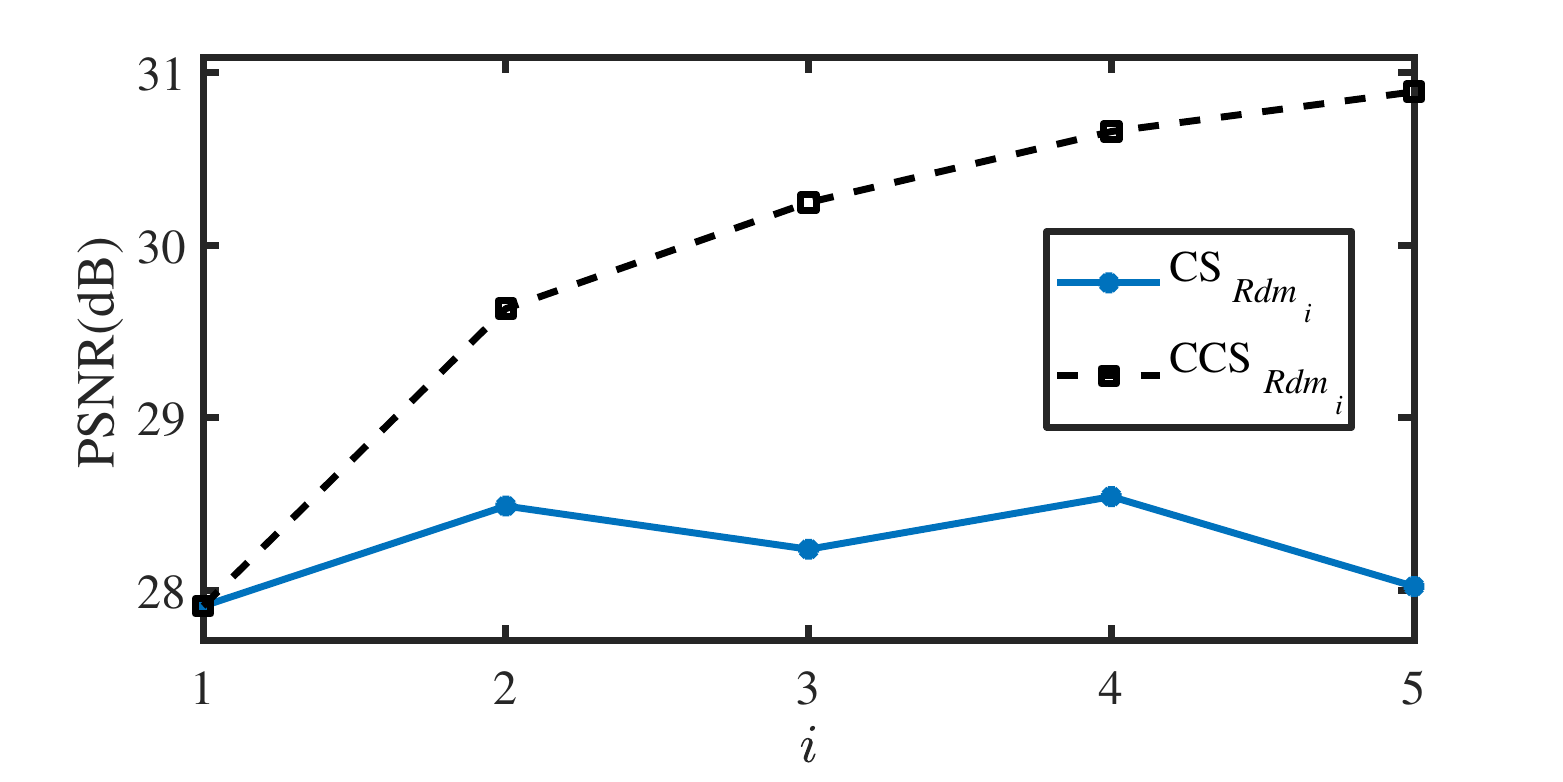}}
\caption{PSNR $\sigma_{psnr}$ of CS systems and collaborative CS systems for image `Plane'.} \label{Figure:CSvsCCS}
\end{figure}

\begin{figure}[htb!]
\begin{minipage}{0.48\linewidth}
\centerline{
\includegraphics[width=1.7in]{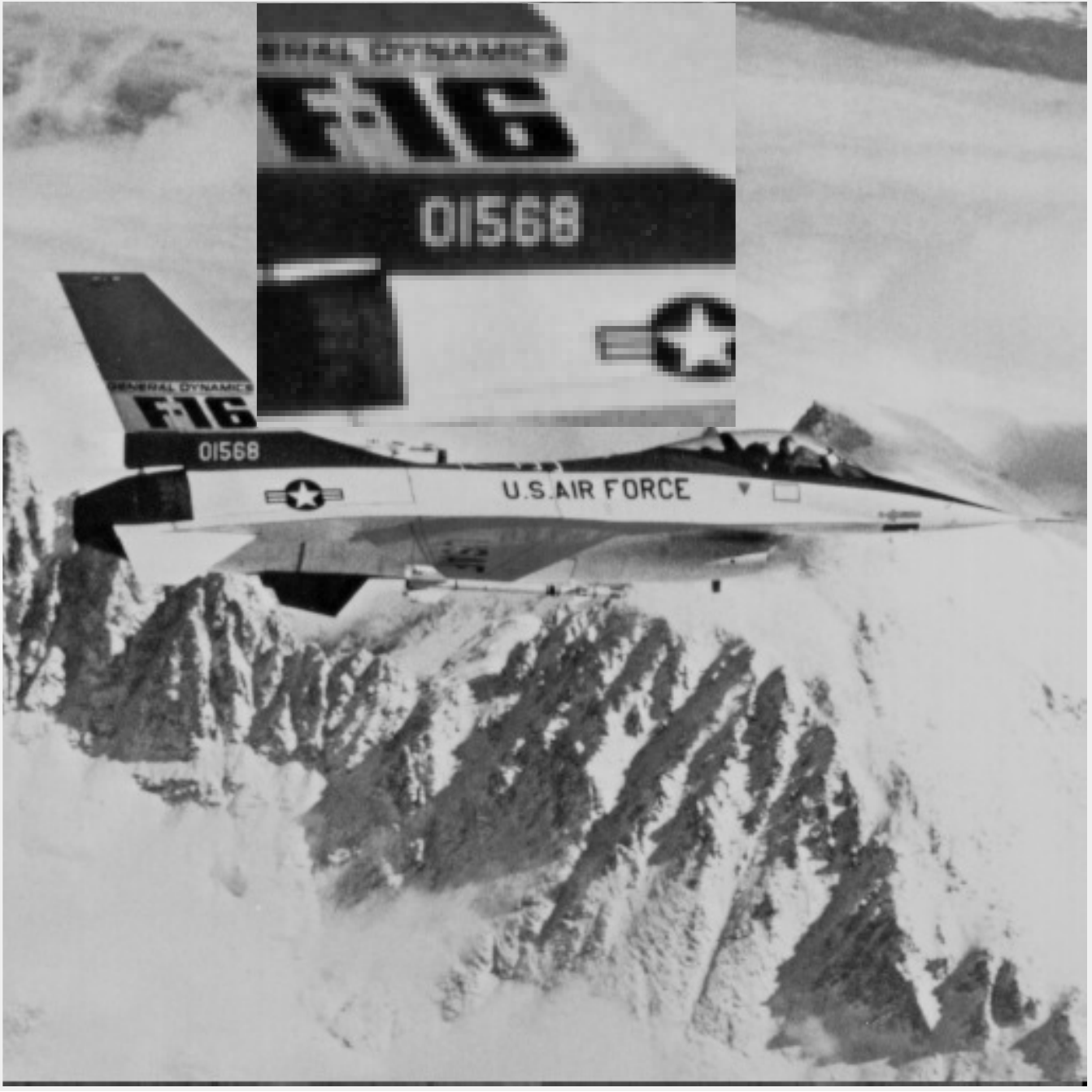}}
\centerline{(a)}
\end{minipage}
\hfill
\begin{minipage}{0.48\linewidth}
\centerline{
\includegraphics[width=1.7in]{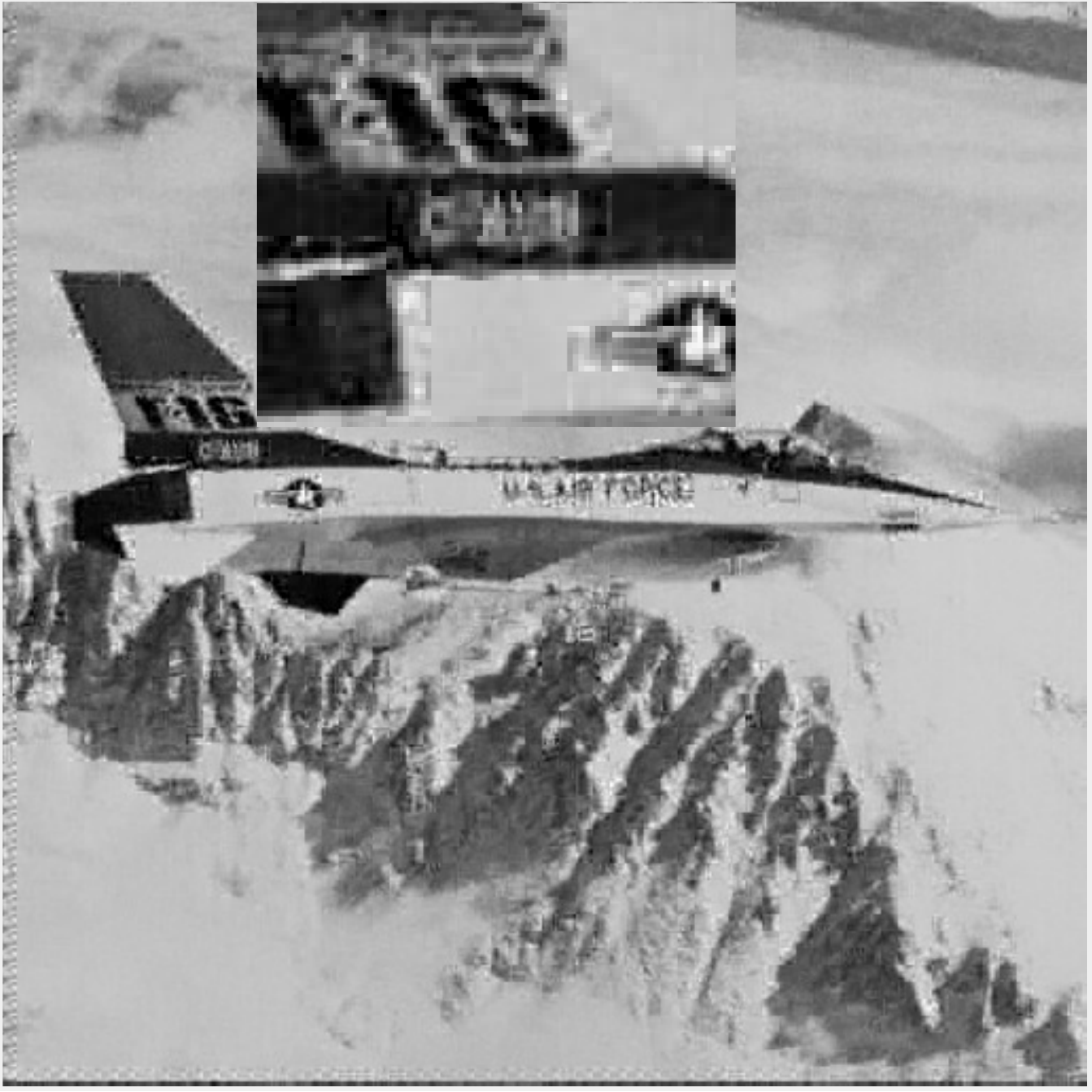}}
\centerline{(b)}
\end{minipage}
\vfill
\begin{minipage}{0.48\linewidth}
\centering
\includegraphics[width=1.7in]{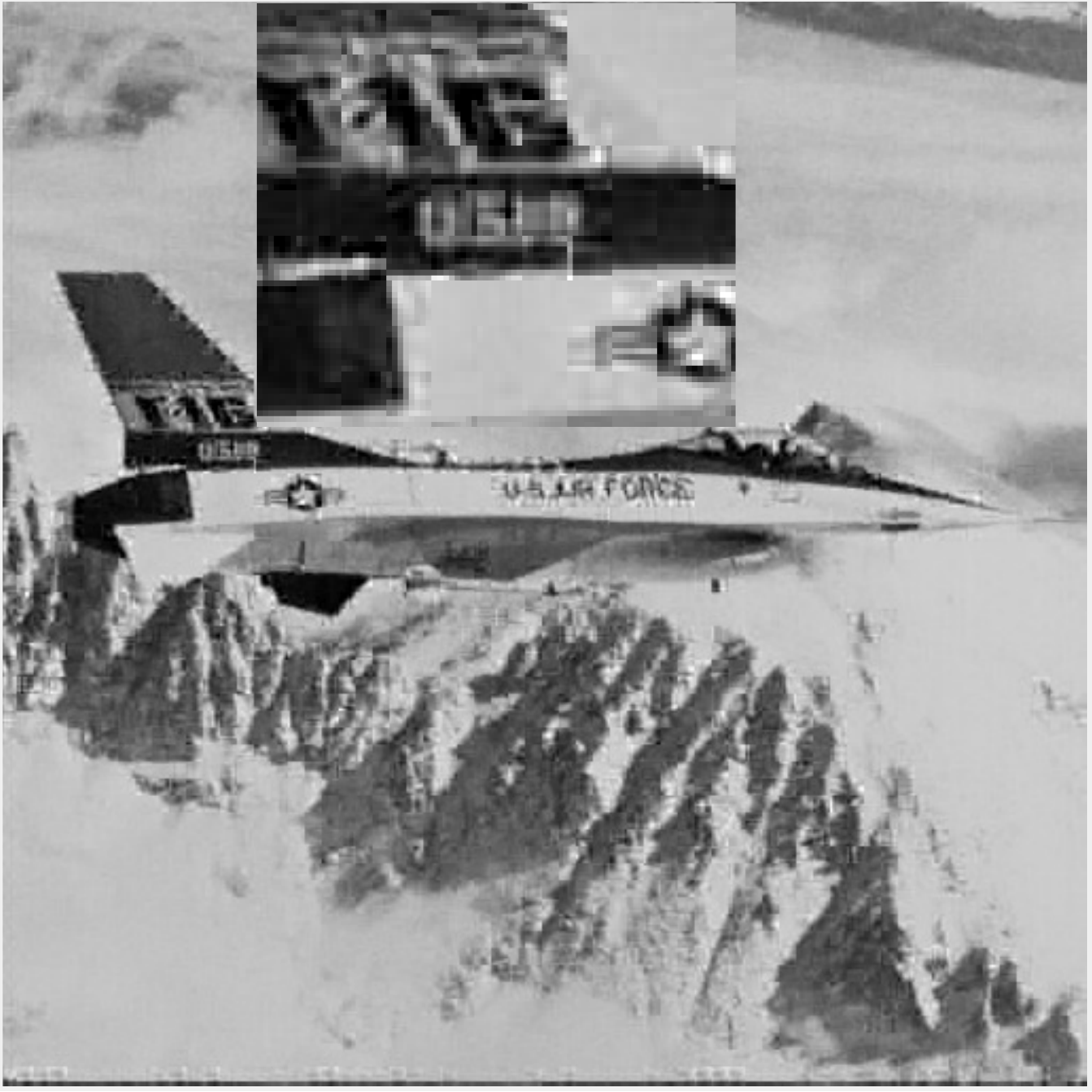}
\centerline{(c)}
\end{minipage}
\hfill
\begin{minipage}{0.48\linewidth}
\centering
\includegraphics[width=1.7in]{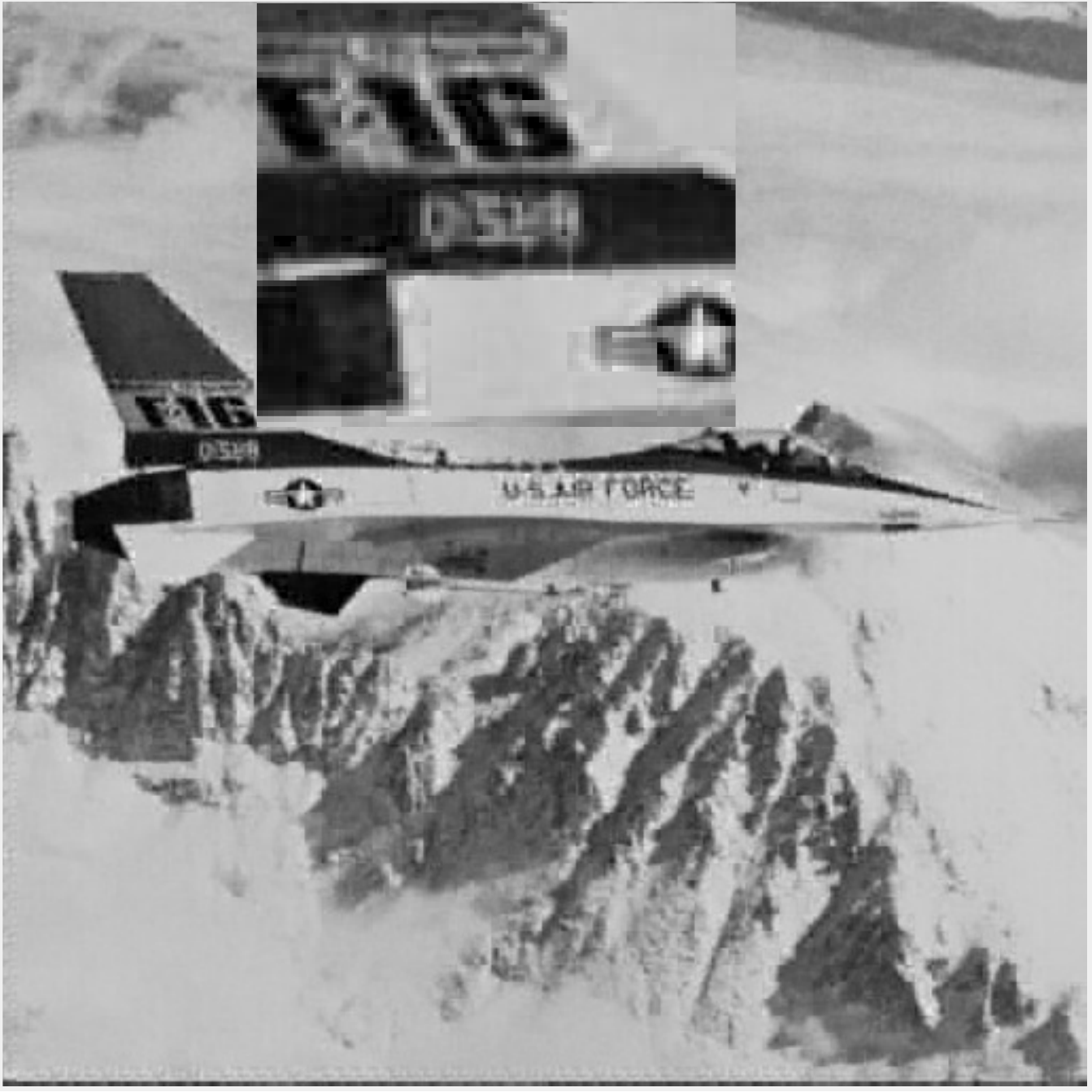}
\centerline{(d)}
\end{minipage}
\caption{The original 'Plane' and reconstructed images from their CS samples with $M = 20$ and $\kappa=4$. (a) The original; (b) $\text{CS}_{Rdm_1}$; (c) $\text{CS}_{Rdm}$; (d) $\text{CCS}_{Rdm}$. }
\label{Figure:CS random for plane}
\end{figure}

\setlength{\tabcolsep}{3pt}
\begin{table*}[thb!]\caption{Statistics of $\sigma_{psnr}$(dB) for images processed with $M=20,N=64,L_0=100,\kappa=4$, and $K=5$.}\label{table:PSNR with Random Sensing}
\begin{center}
\begin{tabu}{c|c|c|c|c|c|c|c|c|c|c|c|c|c}
\hline  &Baboon&Boat&Child&Couple&Crowd&Elanie&Finger&Lake&Lax&Lena&Man&Plane&Test data \cite{russell2008labelme}\\\cline{1-14}
$\text{CS}_{Rdm_1}$ &  22.11  & 26.79  & 31.03 &  26.90  & 27.51  & 29.45  & 23.38  & 25.83  & 22.49 &  29.65 &  27.36  & 27.91 & 27.17 \\\cline{1-14}
$\text{CS}_{Rdm_2}$ & 22.06  &  26.63 &  31.08&   26.95 &  27.55 &  29.37 &  23.17  & 25.87 &  22.51  & 29.43  & 27.37  & 28.49 & 27.19  \\\cline{1-14}
$\text{CS}_{Rdm_3}$ & 22.03 &  26.73  & 31.04 &  26.98  & 27.45 &  29.28  & 23.45 &  25.81  & 22.50  & 29.59  & 27.28  & 28.24  & 27.25\\\cline{1-14}
$\text{CS}_{Rdm_4}$ & 22.21  &  26.81 &  31.05  & 26.98 &  27.53 &  29.39  & 23.81  & 26.00  & 22.55 &  29.74  & 27.45  & 28.54 & 27.18 \\\cline{1-14}
$\text{CS}_{Rdm_5}$ & 22.10  & 26.67  & 31.10 &  26.92 &  27.59 &  29.56  & 23.51  & 25.92 &  22.54  & 29.61&   27.41 &  28.02 & 27.12 \\\cline{1-14}
$\text{CS}_{Rdm}$&
22.39  &  27.24 &  31.70  & 27.51 &  28.31  & 29.78 &  24.31  & 26.46 &  22.77  & 30.32 &  27.89   & 29.22 & 27.78 \\\cline{1-14}
$\text{CS}_{\widetilde{Rdm}}$&
22.14 &  26.86 &  31.16 &  27.08 &  27.84 &  29.40 &  23.68  & 26.02  & 22.51  & 29.75  & 27.49  & 28.83 & 27.45\\\cline{1-14}
$\text{CCS}_{Rdm}$ & {\bf 23.78}   & {\bf 28.86}  & {\bf 33.45} &  {\bf 29.09} &  {\bf 30.04}  & {\bf 31.10} &  {\bf 26.48} & {\bf  28.07} &  {\bf 24.14}  & {\bf 31.93} &  {\bf 29.41}  & {\bf 30.89} & {\bf 29.38}\\\cline{1-14}
\end{tabu}
\end{center}
\end{table*}

The excellent performance of CCS system can also be explained as follows. The model (\ref{collaborative-estimatorX}) suggests that for any patch $\vx$, the following holds
$$||\widehat \vx - \vx||^2_2= ||\sum^K_{i=1}\omega_i(\vx_i-\vx)||^2_2\leq\frac{\sum^K_{i=1} ||\vx_i-\vx||^2_2}{K}$$
when setting $\omega_i =1/K,~\forall~i$.  As each dictionary has limited  capacity of representation, the corresponding CS system may yield excellent performance for some of the patches (of the image), but very poor one for the others. The inequality above implies that for any patch,  the variance of the reconstructed one by the CCS system is always smaller than the average of those by the $K$ individual CS systems. As PSNR is evaluated over all the patches, from a statistical point of view the PSNR of the proposed CCS scheme is  better than that of any single CS system.

\subsection{Performance of the optimized collaborative CS systems}
With the obtained training data $\underline\mX$, we now examine the performance of the CCS system with sensing matrix and dictionaries simultaneously learned by solving \eqref{eq:CS based MLE} with\footnote{
Though there is no systematic way to find the best $\alpha$ and $\beta$ (which depend on specific applications and are also not the main focus of this paper), we provide a rough guide to choose these parameters. With respect to $\alpha$ which weights the importance of the projected signal noise, similar to what is suggested in \cite{duarte2009learning}, $\alpha<1$ is preferred since we need to highlight the sparse representation error $\mE$. In terms of $\beta$, $\beta>\alpha$ is suggested since in dictionary learning for image processing, the sparse representation error $\mE$ is not very small and capturing most of the key information in $\bm \Psi \bm S$ is more important for a sensing matrix.} $\alpha = 0.2$, $\beta = 1$, $\nu_1 = \nu_2 = \nu_3 = 10^{-4}$, $b = 10\|\underline \mX\|_\infty$ and $\nu_4 = \frac{1}{3(\alpha + \beta+1)L_0}$ as suggested by \eqref{eq:Lipschitz for S}. We run the proposed {\bf Algorithm}~\ref{alg:main} with $N_{ite} = 30$ iterations to solve \eqref{eq:CS based MLE}. As \eqref{eq:CS based MLE} is highly nonconvex, a suitable choice of initialization for {\bf Algorithm}~\ref{alg:main} may result in a better solution. We utilize the DCT matrix as the initialization for each dictionary $\mPsi_i$.\footnote{Similar to~\cite{aharon2006KSVD}, we construct an overcomplete separable version of the DCT dictionary by sampling the cosine wave in different frequencies as follows. Construct $\mT_{DCT}\in\Re^{\sqrt{N}\times\sqrt{L_0}}$ as
 $\mT_{DCT}(k,j)= \cos\left(\frac{(k-1)(j-1)\pi}{\sqrt{L_0}}\right), \forall~j\in[\sqrt{L_0}], ~\forall~k\in[\sqrt{N}]$ and normalize each of its column to unit norm to obtain $\overline\mT_{DCT}$ . We then take $\mPsi_{i,0}=\overline\mT_{DCT}\otimes \overline\mT_{DCT}$, the Kronecker product between $\overline\mT_{DCT}$ and itself as the initial for each $\mPsi_i$. Here, the Kronecker product between $\mA\in\Re^{m\times n}$ and $\mB\in\Re^{p\times q}$ results in $\mC\in\Re^{mp\times nq}$ with $\mC(u,v) = \mA(i,j)\mB(k,l)$ for all $u = p(i-1)+k, v = q(j-1)+l$.} We then take the first $M$ left singular vectors of the dictionary as the initialization for the sensing matrix. We conducted the experiments with several type of initializations (like the DCT matrix and the one by randomly selecting a number of training signals for the initialization of the dictionary), and we observed the resulted CCS systems have very similar performance. We defer the thorough investigation of the robustness to the initialization of {\bf Algorithm~\ref{alg:main}} to  future work.

In this section, this learned CCS system along with the plain averaging for fusion strategy is simply denoted by CCS. Fig.~\ref{Figure:obj}(a) shows the convergence behavior of the objective function $\underline\varrho$ as iteration $k$ goes, while Fig.~\ref{Figure:obj}(b) displays $\sigma_{psnr}$(dB) for the test data from LabelMe \cite{russell2008labelme} as iteration $k$ goes. Here, $\sigma_{psnr}(k)$ is computed with the test data that is processed using the CCS system obtained at the $k$th iteration of the proposed algorithm. Note that we display Fig.~\ref{Figure:obj}(b) only to demonstrate the potential relationship between the objective value $\underline\varrho$ defied in~\eqref{eq:underline varrho} and the actual performance of the CCS system and that the test data is not involved in the procedure of training the CCS system. We also show the iterates change $\|\mW_k - \mW_{k-1}\|_F^2$ (recall that $\mW_k = (\mPhi_k,\underline\mPsi_{k},\underline\mS_{k})$) as the iteration $k$ goes in Fig.~\ref{Figure:obj}(c). In Figs.~\ref{Figure:obj}, we set the number of iterations $N_{ite} = 50$ to better illustrate the convergence of {\bf Algorithm~\ref{alg:main}}. Finally, with the learned CCS system, we display the normalized  version $\overline\mGamma$ of the correlation matrix $\mGamma$ estimated using the test data from LabelMe \cite{russell2008labelme}. In particular, let $\mH \in \Re^{64\times 6000}$ denote the set of test data from LabelMe \cite{russell2008labelme} and $\widehat\mH_1,\ldots,\widehat\mH_K$ be $K$ estimates of $\mH$ with the learned CCS system. Then the covariance matrix $\mGamma$ is estimated using the errors $\begin{bmatrix}(\mH - \widehat\mH_1)^\T & \cdots & (\mH - \widehat\mH_K)^\T  \end{bmatrix}^\T$.
As seen from Fig.~\ref{Figure:CorrelationImage}, the estimation errors $\{\mH - \widehat\mH_i\}$, each consisting of sparse representation errors of image signals  and these caused by the reconstruction algorithm,  are generally correlated to each other.

We now compare the CCS with other CS systems. The CS systems CS$_{Elad}$, CS$_{Classic}$, CS$_{TKK}$, CS$_{LZYCB}$, and CS$_f$  are the ones with the learned dictionary $\mPsi_{KSVD}$ obtained using the K-SVD method \cite{aharon2006KSVD} and the sensing matrix $\mPhi \in \Re^{M\times N}$ designed respectively via the methods in \cite{elad2007optimized}, \cite{li2015designing}, \cite{li2013projection}, \cite{tsiligianni2014construction}, and \cite{xu2010optimized}, among which the first four methods are all based on  minimizing the average coherence of the equivalent dictionary, while the last one given in \cite{li2015designing} designs the sensing matrix based on an alternative measure that also takes the sparse representation error into account and hence leads to a more robust CS system than the others. The CS system $\text{CS}_{DCS}$ is the one with the sensing matrix and dictionary  that are simultaneously designed  with the training data $\underline\mX$ using the method\footnote{We set the coupling factor in $\text{CS}_{DCS}$ to $0.5$, which is suggested in \cite{duarte2009learning}.} in \cite{duarte2009learning}.

Table \ref{table:PSNR 1} provides the performance of these CS systems for the same twelve images and the test data from LabelMe \cite{russell2008labelme}. Fig. \ref{Figure:CS for plane} displays the visual effects of image `Plane' for these CS systems. Finally, we list the time used for the algorithms to simultaneously learning the sensing matrix and the dictionary.\footnote{All of the
experiments are carried out on a laptop with Intel(R) i5-3230 CPU @ 2.6GHz and RAM 8G.} The CPU time for learning the CCS system with {\bf Algorithm~\ref{alg:main}} is $565$ seconds, while it takes $502$ seconds for learning CS$_{DCS}$. As a comparison, the K-SVD algorithm spends $424$ seconds to learn the dictionary $\mPsi_{KSVD}$. In other words, the other CS systems (like  CS$_{Elad}$ and CS$_{Classic}$) that first learn the dictionary (with the K-SVD algorithm) and then design the sensing matrix require at least $424$ seconds. We note that all the experiments are not conducted with the parallel computing strategy.\\


\begin{figure}[htb!]
\begin{minipage}{0.32\linewidth}
\centering
\includegraphics[width=2.4in]{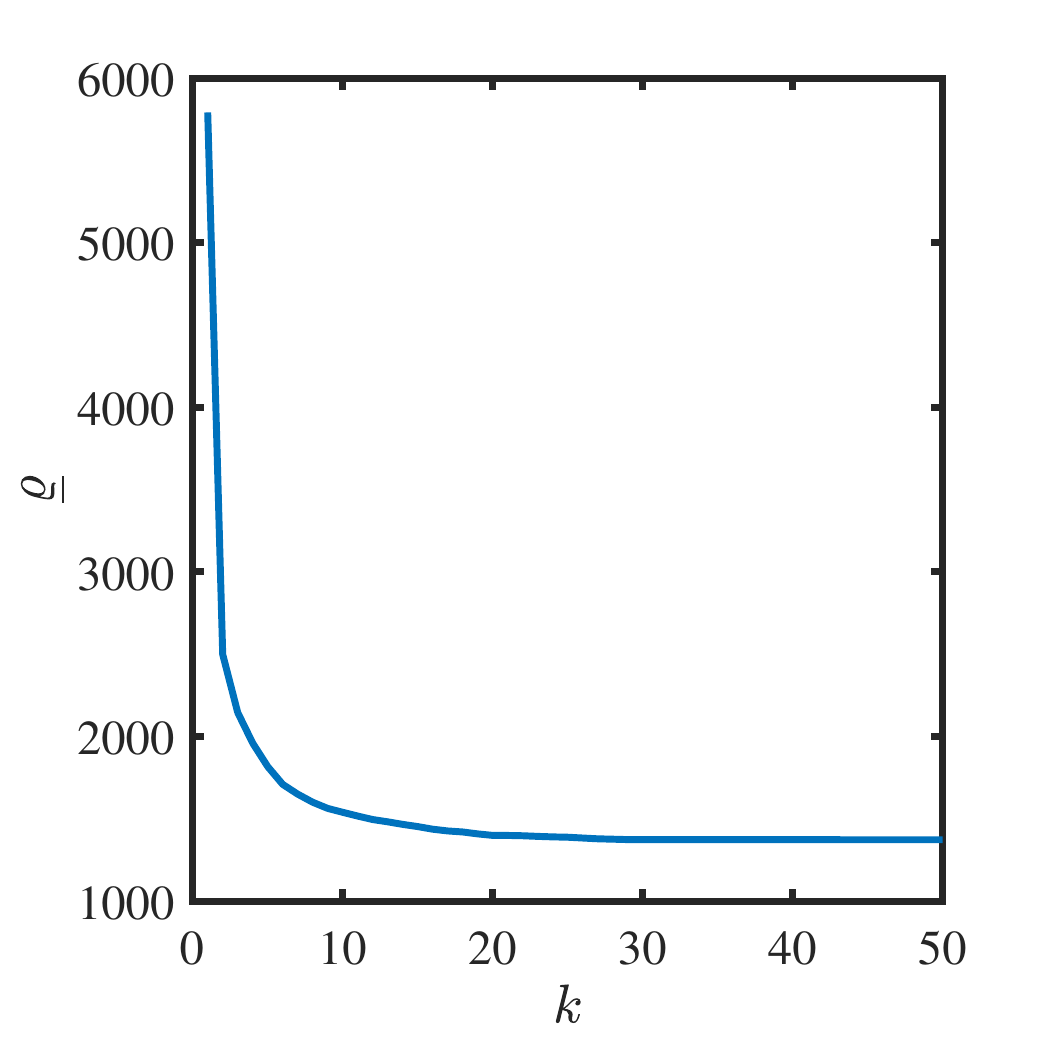}
\centerline{(a)}
\end{minipage}
\hfill
\begin{minipage}{0.32\linewidth}
\centering
\includegraphics[width=2.4in]{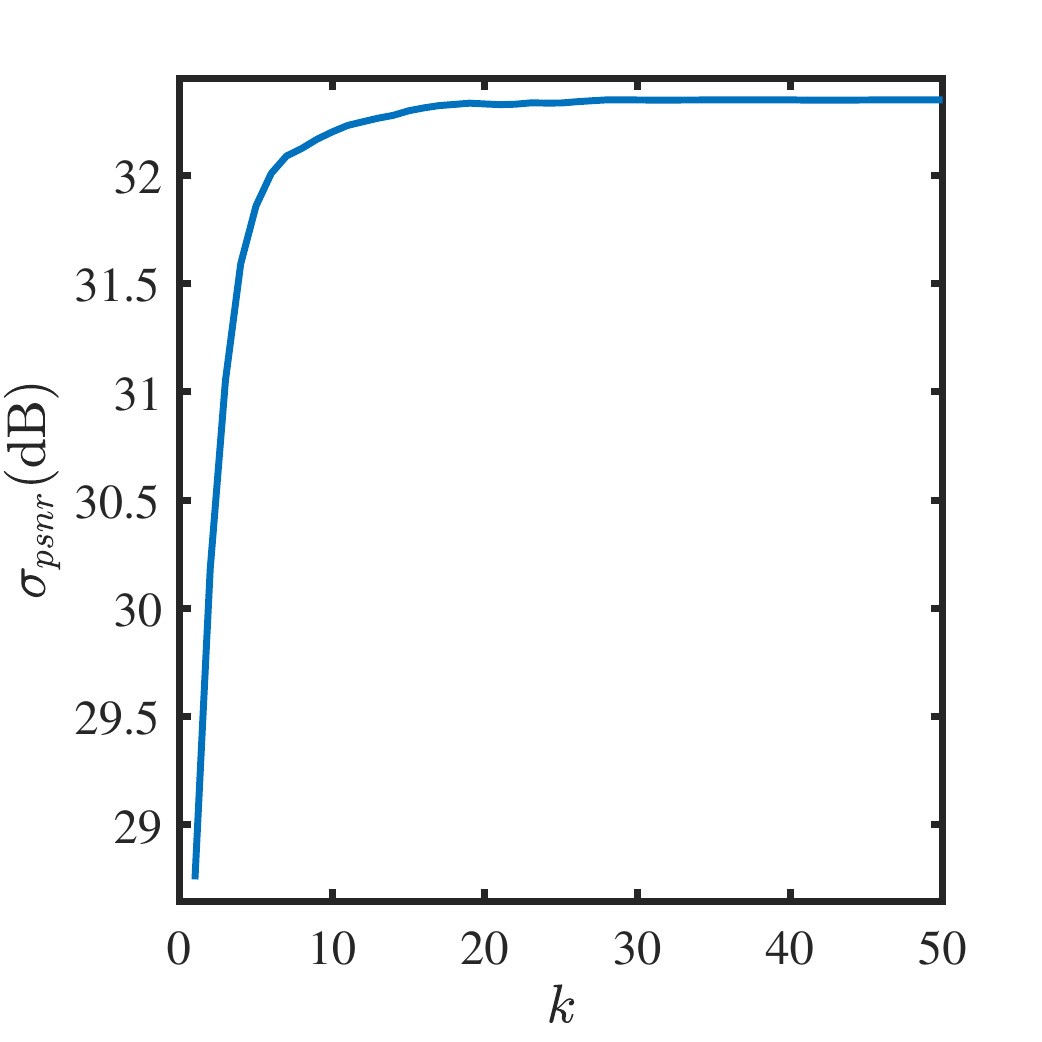}
\centerline{(b)}
\end{minipage}
\hfill
\begin{minipage}{0.32\linewidth}
\centering
\includegraphics[width=2.4in]{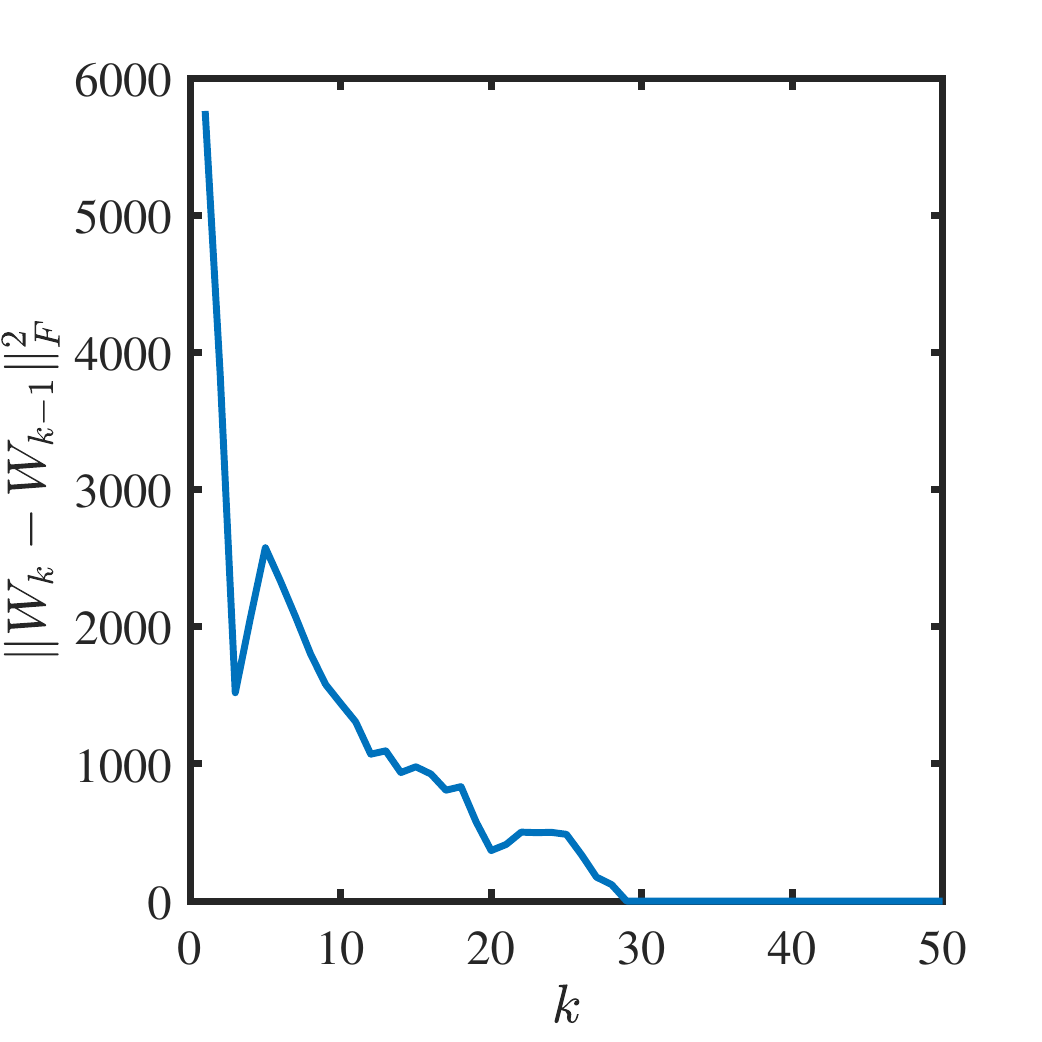}
\centerline{(c)}
\end{minipage}
\caption{Evolution of (a) the cost function $\underline\varrho(\mW_k)$ where $\mW_k = (\mPhi_k,\underline\mPsi_k,\underline \mS_k)$, (b) the PSNR $\sigma_{psnr}$(dB) for the test data from LabelMe \cite{russell2008labelme}, and (c) the iterates change $\|\mW_k - \mW_{k-1}\|_F^2$ with $K=5$.}
\label{Figure:obj}
\end{figure}

\begin{figure}[htb!]
\centering\includegraphics[width=2.8in]{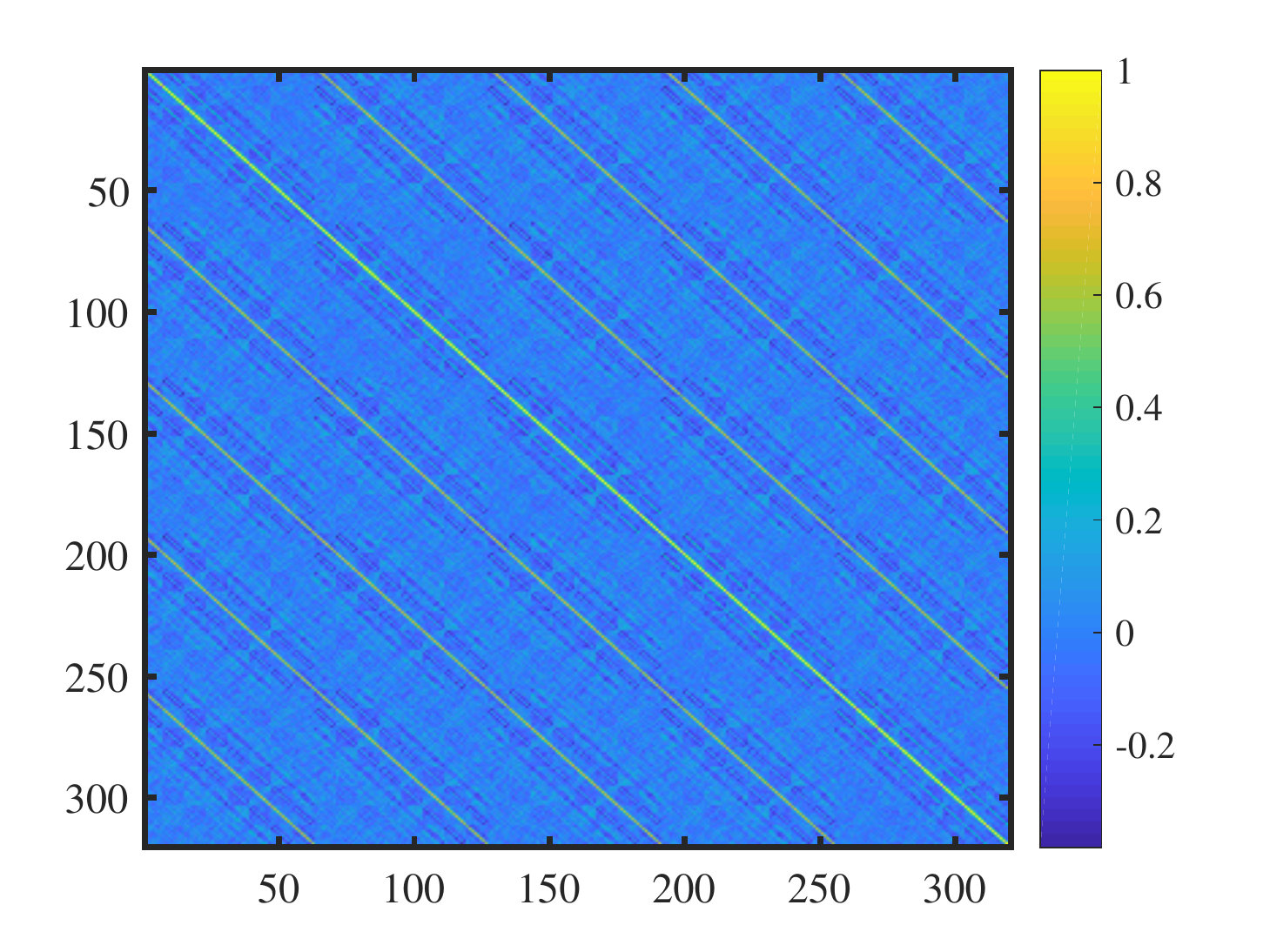}
\caption{Illustration of the correlation matrix $\overline\mGamma$ obtained with the test data from LabelMe \cite{russell2008labelme}.} \label{Figure:CorrelationImage}
\end{figure}

\setlength{\tabcolsep}{3pt}
\begin{table*}[thb!]\caption{Statistics of $\sigma_{psnr}$(dB) for images processed with $M=20, N=64, L_0=100,\kappa=4$, and $K=5$.}\label{table:PSNR 1}
\begin{center}
\begin{tabu}{c|c|c|c|c|c|c|c|c|c|c|c|c|c}
\hline  &Baboon&Boat&Child&Couple&Crowd&Elanie&Finger&Lake&Lax&Lena&Man&Plane&Test data \cite{russell2008labelme}\\\cline{1-14}
$\text{CS}_{Elad}$ & 12.11  & 16.55 &  25.38 &  19.40 &  22.67 &  14.51 &  20.94  & 14.98 &  10.86  & 21.52  & 18.68 &  23.69 & 16.57   \\\cline{1-14}
$\text{CS}_{Classic}$ &9.39  &  13.78 &  22.79 &  16.69  & 20.27 &  11.83 &  18.91 &  12.26   & 8.11 &  18.78  & 15.88  & 21.26 & 12.87  \\\cline{1-14}
$\text{CS}_{TKK}$ &11.93  &  16.31 &  25.03  & 19.02  & 22.28 &  14.32 &  20.39 &  14.60 &  10.68  & 21.19  & 18.34  & 23.32  & 14.85 \\\cline{1-14}
$\text{CS}_{LZYCB}$ & 8.15   & 12.54  & 21.64  & 15.52 &  19.22 &  10.73  & 17.98 &  11.08  &  6.90 &  17.67  & 14.82  & 20.11 & 12.19  \\\cline{1-14}
$\text{CS}_{DCS}$ &25.35  & 30.39 &  34.74  & 30.58 &  31.45  & 32.62 &  27.39 &  29.59  & 25.74  & 33.38  & 30.87  & 32.50 & 30.39  \\\cline{1-14}
$\text{CS}_{f}$ & 24.85  &  30.04  & 34.70 &  30.16  & 31.46  & 32.18 &  27.44  & 29.42  & 25.27 &  33.01  & 30.61  & 32.45 & 30.19 \\\cline{1-14}
$\text{CCS}$ &26.25  & 31.57  & 36.42 &  31.92 &  33.31 &  33.36 &  29.98 &  31.07 &  26.44  & 34.73  & 32.03 &  34.18 & 32.32\\\cline{1-14}
$\text{CCS}_8$ &  {\bf 26.34} &  {\bf 31.70} &  {\bf 36.62} &  {\bf 32.05} &  {\bf 33.50} &  {\bf 33.42}  & {\bf 30.29} &  {\bf 31.22}  & {\bf 26.55} &  {\bf 34.89} &  {\bf 32.16} &  {\bf 34.42} & {\bf 32.45}\\\cline{1-14}
\end{tabu}
\end{center}
\end{table*}

\begin{figure}[htb!]
\begin{minipage}{0.48\linewidth}
\centering
\includegraphics[width=1.7in]{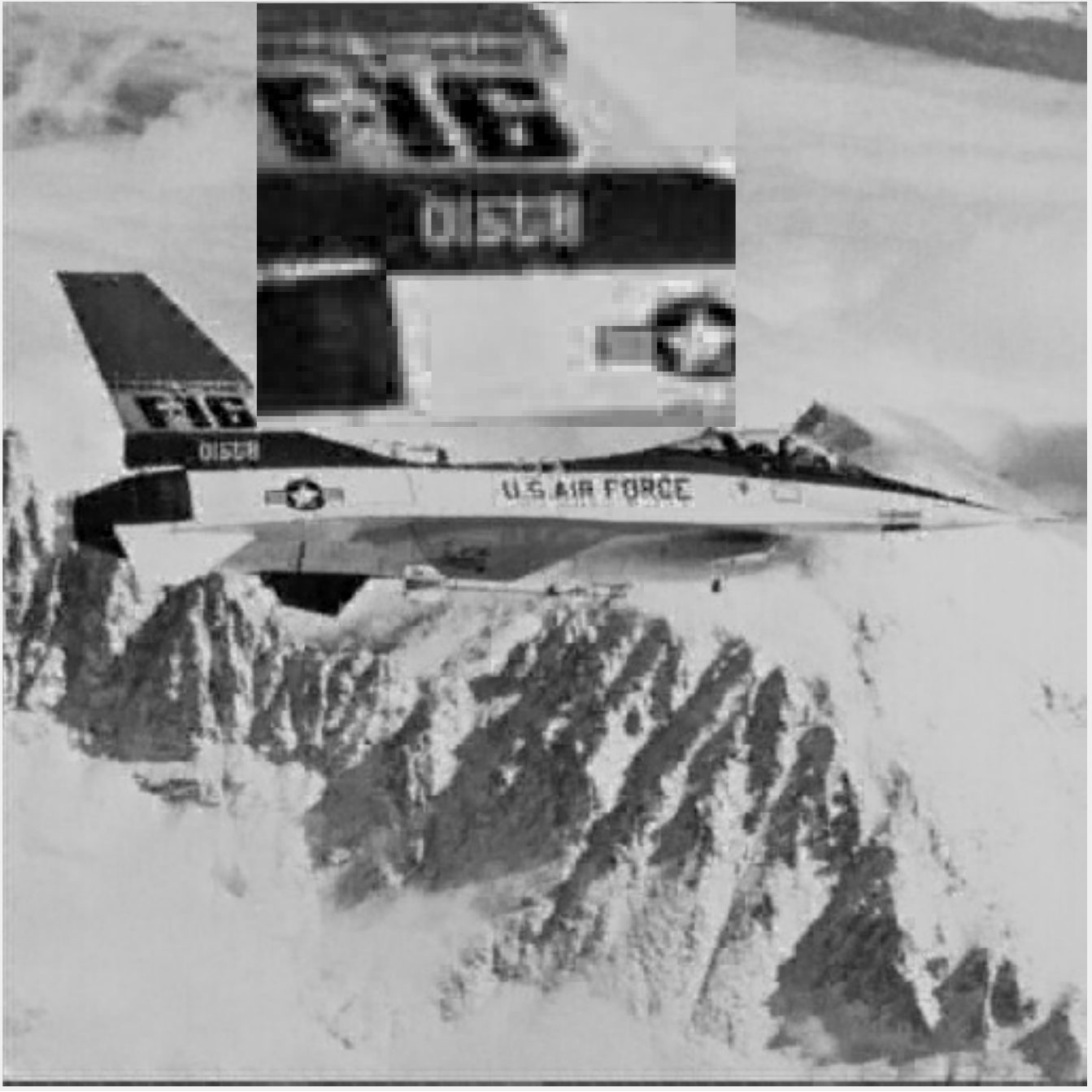}
\centerline{(a)}
\end{minipage}
\hfill
\begin{minipage}{0.48\linewidth}
\centering
\includegraphics[width=1.7in]{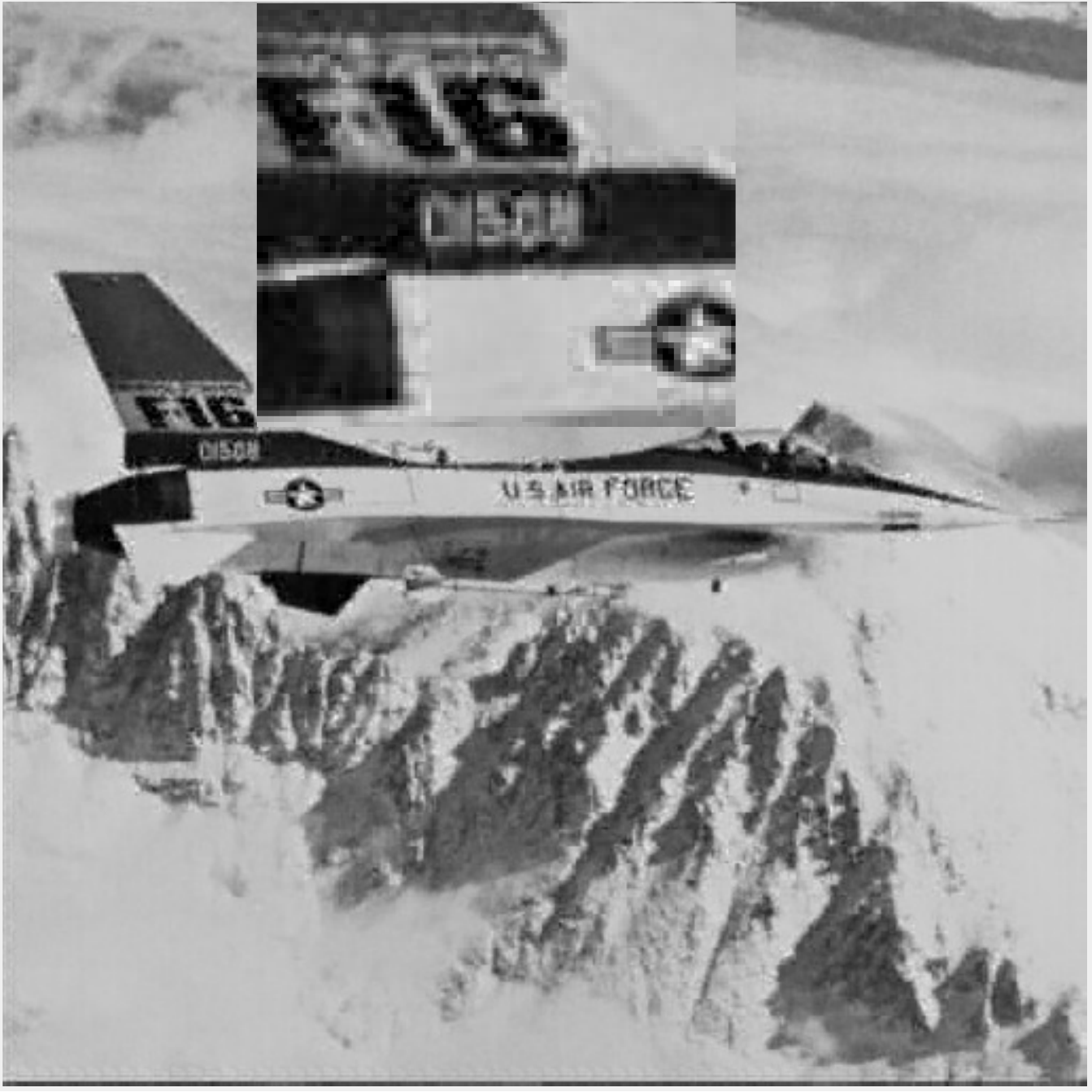}
\centerline{(b)}
\end{minipage}
\vfill
\begin{minipage}{0.48\linewidth}
\centering
\includegraphics[width=1.7in]{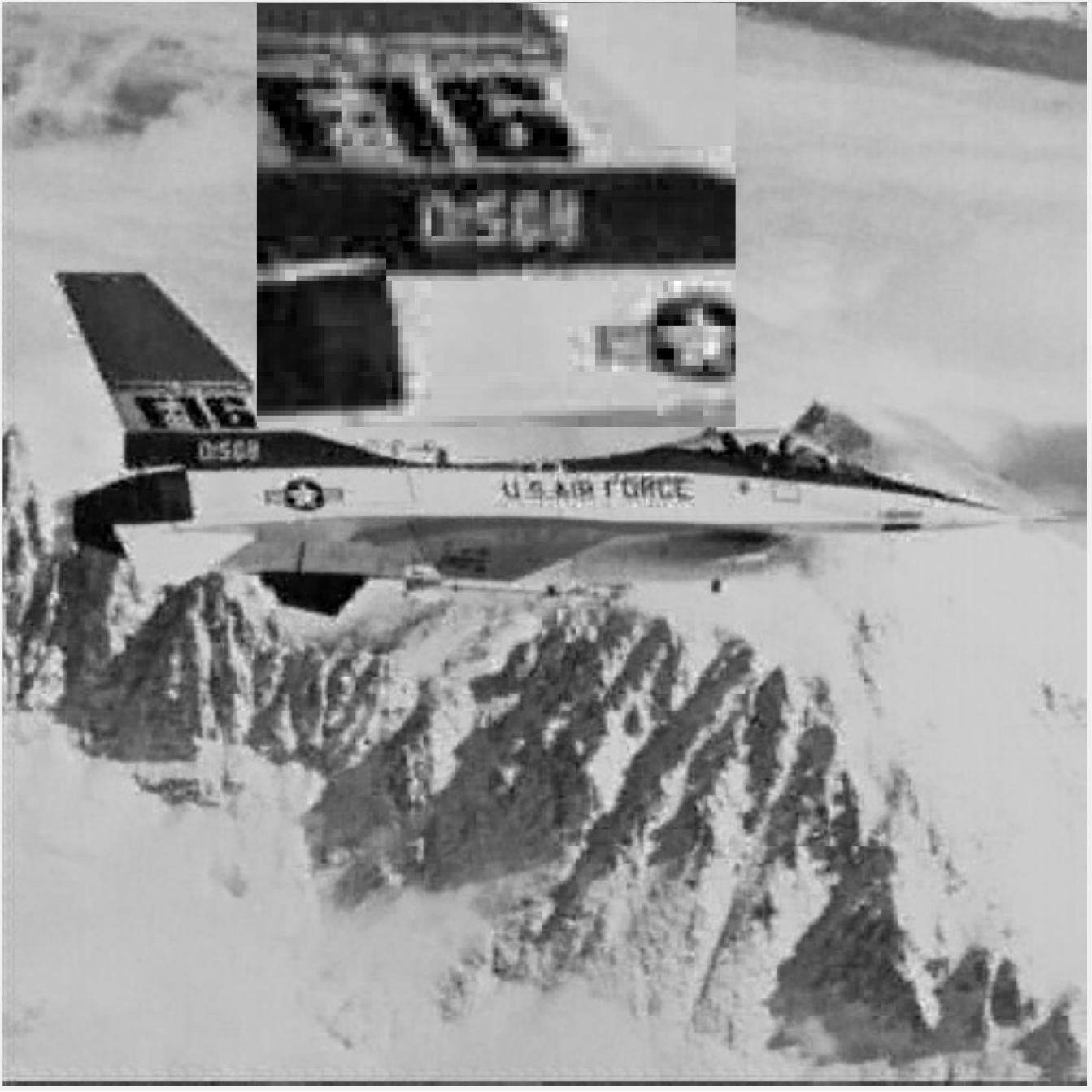}
\centerline{(c)}
\end{minipage}
\hfill
\begin{minipage}{0.48\linewidth}
\centering
\includegraphics[width=1.7in]{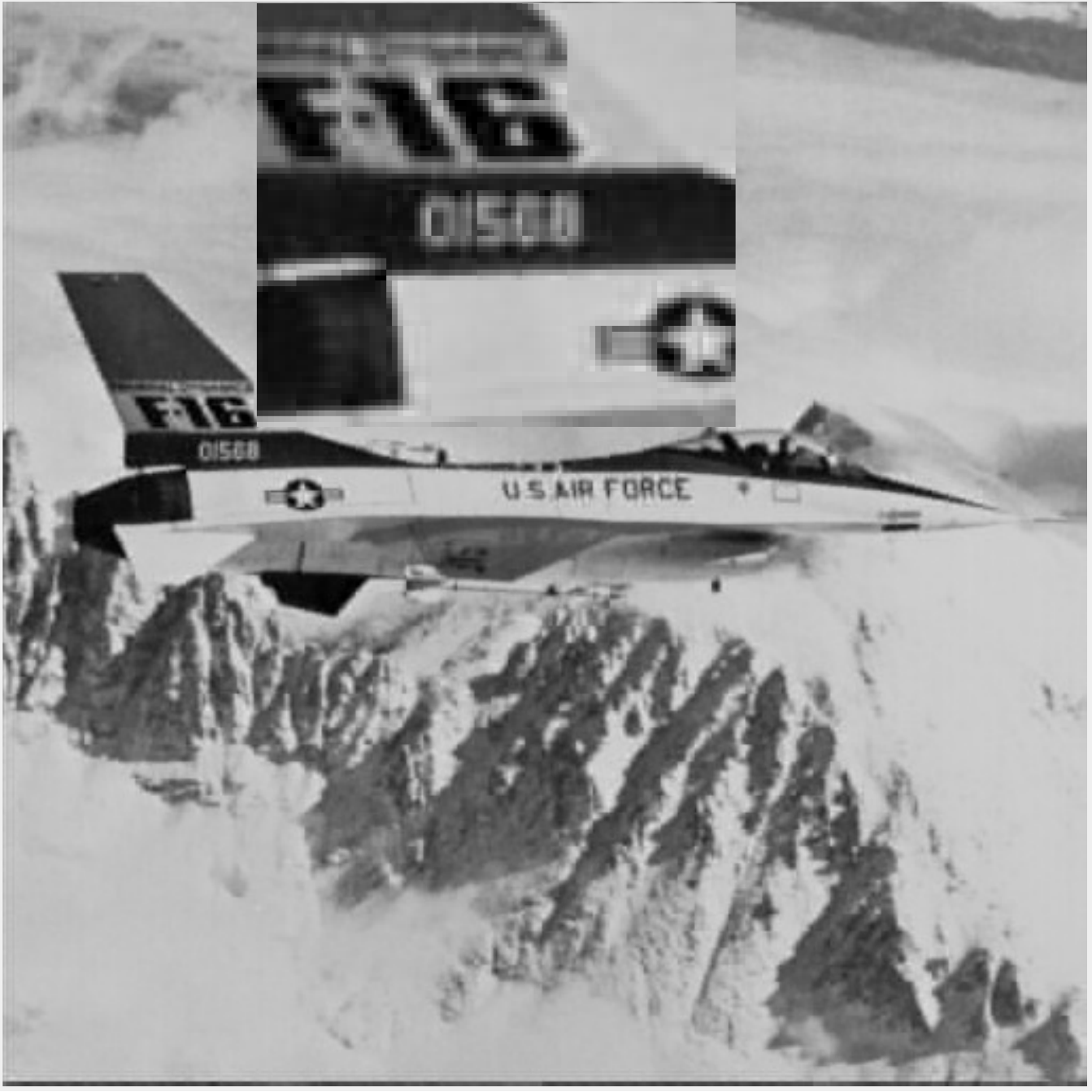}
\centerline{(d)}
\end{minipage}
\caption{The  reconstructed 'Plane' images from their CS samples with $M = 20$ and $\kappa=4$. (a) $\text{CS}_{DCS}$; (b) $\text{CS}_{f}$; (c) $\text{CS}_{etf}$; (d) $\text{CCS}$.}
\label{Figure:CS for plane}
\end{figure}

\vspace{0.5cm}
{\noindent \bf Remark~6.3:}
\begin{itemize}
\item Comparing Fig.~\ref{Figure:obj}(a) and Fig.~\ref{Figure:obj}(b), we observe that the PSNR increases when the objective function decreases. We also observe from Fig.~\ref{Figure:obj}(c) that the change of iterates $\|\mW_k - \mW_{k-1}\|_F^2$ converges (almost) to 0 after 50 iterations. These confirm the convergence analysis of the proposed {\bf Algorithm~\ref{alg:main}} given in Section~\ref{sec-convergence} and demonstrates that the objective function $\underline\varrho$ is a good surrogate for the performance of the CS systems  in terms of reconstruction accuracy;
\item It is observed from Tables \ref{table:PSNR with Random Sensing} and \ref{table:PSNR 1} that the CS systems CS$_{Elad}$, CS$_{Classic}$, CS$_{TKK}$, and CS$_{LZYCB}$ yield  a performance even poorer than those that use a random sensing matrix. This is due to the fact that the sensing matrices in the former are optimized with the assumption that the sparse representation errors of signals are very small, which is not the case for images at all;
    \item As the sensing matrices used in CS$_f$  and CS$_{DCS}$ are optimized with the sparse representation error taken into account, they perform much better than those using a random sensing matrix and even better than the CCS system CS$_{Rdm}$ in which the sensing matrix is randomized and all the dictionaries are independently designed using the K-SVD algorithm;
  \item  The superiority of our optimized CCS system CCS is again demonstrated clearly. Compared with CS$_{Rdm}$ (see TABLE~1), it is basically 2.5 dB better in terms of $\sigma_{psnr}$. This demonstrates clearly the importance of the proposed joint optimization in designing CCS systems. Furthermore, the  optimized CCS system CCS yields a performance much better than CS$_f$ and CS$_{DCS}$. The amount of improvement can be further increased if the number of CS systems fused in the CCS system increases. With the $5\times 6000=30,000$ samples in the obtained training data matrix $\underline \mX$, a CCS system  with $K=8$, denoted as  CCS$_8$, is also obtained and its performance is better than the optimized CCS system CCS with $K=5$ by a amount of 0.15 dB on average. See Table~\ref{table:PSNR 1}.
\end{itemize}

\section{Conclusion}\label{sec:conclusion}
Based on the MLE principle, a collaborative estimation framework has been derived. It is shown that by fusing several estimators the collaborative estimator yields a better an estimate than any of the individual estimators when the estimation errors of these estimators are not very correlated.    As an application to CS, the CCS scheme has been raised and analyzed, which consists of a bank of $K$ CS systems that share the same sensing matrix.  A measure has been proposed, which allows us to simultaneously optimize the sensing matrix and all the $K$ dictionaries. For solving the corresponding optimal CCS system problem, an alternating minimization-based algorithm has been derived with guaranteed convergence. Experiments with synthetic data and real images showed that the proposed optimized CCS systems outperform significantly the traditional CS systems in terms of signal reconstruction accuracy. 

In practice, the exact statistical properties can not be obtained. It is an interesting topic to investigate how to learn a collaborative estimator. It is also of interest to develop distributed methods for efficiently implementing the proposed algorithm~\cite{koppel2017d4l,raja2016cloud}. Developing alternating minimization-based algorithms with guaranteed convergence is still an important topic in the circumstances, where non-convex optimizations are confronted. The results achieved in the paper seem useful to this topic. Works in these directions are on-going.

\section*{\bf Acknowledgment}
The authors would like to thank the reviewers for their detailed comments and constructive suggestions that help us improve the  quality of this paper.\\

\appendix
\section{Proof of Lemma~\ref{lem:update Phi}}\label{sec:update Phi}
\begin{proof}[Proof of Lemma~\ref{lem:update Phi}] First of all, let $\varrho(\mPhi,\mPsi,\mS,\mX)$ and $\mG(\mPsi,\mS,\mX)$ be the same as defined before in \eqref{the measure} and \eqref{eq:G}, respectively. With $\mPhi \in {\cal O}_{M,N}$, that is  $\mPhi\mPhi^{\cal T} = \mId_M$, it can be shown that
$$\varrho(\mPhi,\mPsi,\mS,\mX) = - \trace[\mPhi^{\cal T}\mPhi \mG(\mPsi,\mS,\mX)]+ct,$$
where $ct$  is independent of $\mPhi$. Thus, (\ref{general-ZhuZHx}) is equivalent to
\begin{align*}
 \mPhi_k
&= \argmax_{\mPhi \in {\cal O}_{M,N}} \left\{\trace\left[\mPhi^{\cal T}\mPhi\left(\sum_{i=1}^K \mG\left(\mPsi_{i,k-1},\mS_{i,k-1},\mX_{i}\right) + 2\nu\mPhi_{k-1}^\T\mPhi_{k-1}\right)\right]:= \eta\right\}.
\end{align*}
It then follows from the ED of $\widetilde \mG :=  \sum_{i=1}^K \mG(\mPsi_{i,k-1},\mS_{i,k-1},\mX_i) +  2\nu \mPhi_{k-1}^\T\mPhi_{k-1} =  \mV_{\tilde \mG} {\bf \Pi} \mV_{\tilde \mG}^{\cal T}$ and $\mPhi = \mU\MAT{cc} \mId_M& {\bf 0}\mat \mV^{\cal T}$ - an SVD of $\mPhi$  due to $\mPhi \in {\cal O}_{M,N}$ that
 \begin{align*}
 \eta &= \trace\left[\tilde \mV^\T{\bf \Pi}  \tilde \mV  \MAT{cc}\mId_M& {\bf 0}\\{\bf 0} & {\bf 0}\mat \right]=\sum^M_{n=1}\mQ(n,n),
 \end{align*}
 where  $\tilde \mV := {\bf V}^{\cal T}_{\tilde \mG} \mV$ and  $\mQ := \tilde \mV^\T{\bf \Pi}  \tilde \mV$.

 Without loss of generality, assume that the diagonal elements $\{\pi_n\}$ of ${\bf \Pi}$ satisfy $\pi_n\geq \pi_{n+1},~\forall~n$. According to \cite{li2013projection} (see the proof of Theorem 3 there), one has
 \begin{align*}
 \eta &=\sum^M_{n=1}\mQ(n,n)\leq \sum^M_{n=1}\pi_n.
 \end{align*}
 Thus, $\eta$ is maximized if $\tilde \mV$ is of the following form
 \e \tilde \mV =\MAT{cc}\tilde \mV_{11}&\bf 0\\ \bf 0& \tilde \mV_{22}\mat, \label{optimal-V-1}\ee
 where both $\tilde \mV_{11} \in {\cal O}_{M}$ and $\tilde \mV_{22} \in {\cal O}_{N-M}$ are arbitrary. Consequently, the optimal sensing matrix is given by
 $$\mPhi_k =\mU\MAT{cc} \mId_M& {\bf 0}\mat \tilde \mV^{\cal T} \mV^\T_{\tilde \mG},$$
 where $\tilde \mV $ is of form given by \eqref{optimal-V-1} and $\mU \in {\cal O}_M$ is arbitrary. Clearly, $\tilde \mV =\mId_N$ yields one of the solutions, that is \eqref{optimal-T-1x}.
\end{proof}

\section{Proof of Lemma~\ref{lem:OMP stationary point}}\label{appendix AB}
The proof needs the concept of Fr\'{e}chet subdifferential, which is given in the following definition.

\begin{definition}\label{def:stationary point}
Let $g(\vx):\Re^N\rightarrow \Re\cup\{\infty\}$ be a proper lower semi-continuous function in $\vx \in \Re^{N}$.
\begin{itemize}
\item The domain of $g(\vx)$ is defined by $\text{dom}_g:=\{\vx\in\Re^N:g(\vx)<\infty\}$.
\item The Fr\'{e}chet subdifferential $\partial g$ is defined by
    \[
    \partial g(\vx) = \left\{\vz:\lim_{\vy\rightarrow \vx}\inf\frac{g(\vy) - g(\vx) - \langle \vz, \vy - \vx\rangle}{\|\vx - \vy\|_2}\geq 0\right\}
    \]
for any $\vx\in \text{dom}_g$ and $\partial g(\vx) = \emptyset$ if $\vx\notin \text{dom}_g$, where $<,>$ denotes the inner product of two vectors.\footnote{When $g(\mX)$ is function of a matrix variable $\mX \in \Re^{N\times M}$, $\partial g(\mX)$ can be defined in the same way but with  $\|\cdot\|_2$  replaced with $\|\cdot\|_F$ and `matrix inner product' $\langle\mZ,\mY-\mX\rangle$ is defined as $\sum^M_{m=1}\langle\mZ(:,m),\mY(:,m)-\mX(:,m)\rangle$, respectively.}
\item For each $\vx\in \text{dom}_g$, $\vx$ is called the stationary point of $g(\vx)$ if it satisfies $\vzero\in \partial g(\vx)$.\\
\end{itemize}
\end{definition}

When $g(\vx)$ is differentiable at $\vx$, it is clear that the gradient of $g(\vx)$ satisfies: $\nabla g(\vx)\in \partial g(\vx)$.\\

The OMP algorithm is practically used for solving \eqref{eq:sparse recovery}. Under certain conditions (like the mutual coherence or the RIP) on $\mA$, such an algorithm can guarantee to recover the $\kappa$-sparse solution $\vs^\star$ if $\vy = \mA \vs^\star$~\cite{davenport2010analysis, tropp2004greed}. It is outlined in {\bf Algorithm}~\ref{alg:OMP}, where the set of positions of the nonzero entries of $\vs$ is referred to as the {\em support} of $\vs$, denoted by $\text{supp}(\vs)$, and $\vs(\Lambda)$  denotes the elements of $\vs$ indexed by $\Lambda$.

\begin{algorithm}[htb]
	\caption{OMP algorithm for solving \eqref{eq:sparse recovery}}
	\label{alg:OMP}
\noindent{\it Initialization:}  set $\vs_0 = \vzero,\vr_0 = \vy,\Lambda_0 = \emptyset$.\\
\noindent{\it Begin $j=0,1, \cdots, \kappa-1$}
\begin{itemize}
\item identify the support:
\e\begin{split}\vh_j = \mA^\T\vr_j,\ \Lambda_{j+1} = \Lambda_j \cup \{\arg \max_l|\vh_j(l)|\}~\text{(choose one if multiple maxima exist).}
\end{split}\label{eq:OMP 1}\ee
\item update the estimate and the residual:
\e
\vs_{j+1} = \argmin_{\vz:\text{supp}(\vz)\subset \Lambda_{j+1}}\|\vy - \mA\vz\|^2_2,\ \vr_{j+1} = \vy - \mA \vs_{j+1}.
\label{eq2OMP 1}\ee
\end{itemize}
\noindent{\it End }\\
\noindent{\it Output}: $\widehat \vs=\vs_\kappa$.
\end{algorithm}

\begin{proof}[Proof of Lemma \ref{lem:OMP stationary point}] First note that the OMP algorithm iteratively finds an atom from $\mA$ that has not been selected before, is not within the subspace spanned by the selected atoms, and best fits the residual vector $\vr$. Thus, if $\vh_j\neq \vzero$, we know that the atoms of $\mA$ indexed by $\Lambda_{j+1}$ are linearly independent and $\vs_{j+1}(\Lambda_{j+1})$ has no zero element. 

We first suppose $\vh_{\kappa -1}\neq \vzero$, which implies that $\vh_{\kappa -1}$ has at least one non-zero elements. In this case, we know $\widehat \vs(\Lambda_\kappa)$ has no zero element. Thus, when $\vs\rightarrow \widehat \vs$, we have $\Lambda_\kappa\subset\text{supp}(\vs)$. Now if $\Lambda_\kappa = \text{supp}(\vs)$, we know $g(\vs) \geq g(\widehat\vs)$ since
$\widehat\vs = \argmin_{\vz:\text{supp}(\vz)= \Lambda_{\kappa}}\|\vy - \mA\vz\|^2_2$. Otherwise $\Lambda_\kappa \subset \text{supp}(\vs)$, we have $g(\vs) = \infty$ since $\vs$ has more than $\kappa$ number of non-zero elements. Thus, we conclude
\[
\lim_{\vs\rightarrow \widehat\vs}\inf\frac{g(\vs) - g(\widehat\vs) - \langle \vzero, \vs - \widehat\vs\rangle}{\|\vs - \widehat\vs\|}\geq 0,
\]
which implies $\vzero \in \partial  g(\widehat \vs)$.

Now if $\vh_{\kappa -1}= \vzero$, which means that $\vr_{\kappa -1}$ is orthogonal to $\mA$. We can decompose $\vy$ into two parts, $\vy = \vy_1 + \vy_2$, where $\vy_1$ is within the subspace spanned by $\mA(\Lambda_{\kappa-1})$, and $\vy_2$ is orthogonal to the subspace spanned by $\mA$.  Thus, in this case $\vs_{\kappa -1}$ is the global solution to \eqref{eq:sparse recovery} (and also to $\min_{\vs}\|\mA\vs- \vy\|^2_2$), which implies that $\widehat \vs = \vs_{\kappa}$ is also the global solution to \eqref{eq:sparse recovery}. Thus, it follows from the optimality condition that $\vzero \in \partial  g(\widehat \vs)$.
\end{proof}

\section{Proof of Theorem \ref{thm:convergence}}\label{appendix B}
\begin{proof} Note that $\underline\varrho(\mPhi,\underline\mPsi,\underline\mS)$ defined before can be considered as a function of three matrix-variables, and it has
a (Fr\'{e}chet) subdifferential of $\underline\varrho$ with respect to each variable, say $\underline\mPsi$. Such subdifferential is denoted as $\partial_{\underline\mPsi}\underline\varrho$. In the same way, $\partial_{\vpsi_{i,\ell}} \underline \varrho$ denotes the (Fr\'{e}chet) subdifferential of $\underline\varrho$ with respect to the $\ell$-th column in the dictionary $\mPsi_i$. All these apply to the notations for gradient functions. Note that due to $\underline\varrho(\mPhi,\underline\mPsi,\underline\mS) =\sum_{i=1}^K \varrho(\mPhi,\mPsi_i,\mS_i)$, we have $\nabla_{\vpsi_{i,\ell}}\underline\varrho (\mPhi,\underline\mPsi,\underline\mS) = \nabla_{\vpsi_{\ell}} \varrho(\mPhi,\mPsi_i,\mS_i)$.

Let $\mW_k = (\mPhi_k,\underline\mPsi_k,\underline\mS_k)$ be the point generated with our proposed algorithm at the $k$-th iteration and $\vpsi_{i,k,\ell}$ denote the $\ell$-th column of $\mPsi_{i,k}$ - the $i$th dictionary corresponding to $\underline\mPsi_k$. Define
\[
\mPsi_{i,k,\ell}/\vpsi: = \begin{bmatrix}\vpsi_{i,k,1} & \cdots \vpsi_{i,k,\ell-1} & \vpsi & \vpsi_{i,k-1,\ell+1} & \vpsi_{i,k-1,L_0} \end{bmatrix}.
\]
With this definition, we have $\mPsi_{i,k,L_0}/\vpsi_{i,k,L_0} = \mPsi_{i,k}$ and $\mPsi_{i,k,1}/\vpsi_{i,k-1,1} = \mPsi_{i,k-1}$.

Since the sets $\calO_{M,N},\calU_N,\calS_{\kappa,b}$ all are bounded and $\underline\varrho$ is smooth, it is clear that $\underline\varrho$ has Lipschitz continuous gradient~\cite{bolte2014proximal}, i.e., there exists a constant $L_c>0$ such that\footnote{Here, $\nabla \underline\varrho(\mPhi,\underline\mPsi,\underline\mS)$ is considered as a variable of the same structure as $\mW$ with the three elements replaced with the corresponding gradients $\nabla_{\mPhi} \underline\varrho$, $\nabla_{\underline\mPsi} \underline\varrho$, and $\nabla_{\underline\mS} \underline\varrho$.}
\e
\left\|\nabla \underline\varrho(\mPhi,\underline\mPsi,\underline\mS) - \nabla \underline\varrho(\mPhi',\underline\mPsi',\underline\mS')\right\|_F \leq L_c \left\|(\mPhi,\underline\mPsi,\underline\mS) - (\mPhi',\underline\mPsi',\underline\mS')\right\|_F
\label{eq:Lipschitz gradient}\ee
for all $\mPhi,\mPhi'\in \calO_{M,N},\underline\mPsi,\underline\mPsi'\in \calU_N^{KL_0}$ and $\underline\mS,\underline\mS'\in \calS_{\kappa,b}^{KJ_0}$.\footnote{We can explicitly estimate $L_c$ as we derive $L_{cs}$ in \eqref{eq:derive S}. We omit the detail since the explicit expression of $L_c$ is not required in the following proof. But we note that $L_c$ is expected to be larger than $L_{cs}$ since the later is the Lipschitz constant regarding the partial gradient (i.e., by letting $\mPhi = \mPhi'$ and $\underline\mPsi = \underline\mPsi'$, \eqref{eq:Lipschitz gradient} reduces to \eqref{nu-4}).}\\

\noindent{\it Proof of Part~$(i)$ of Theorem \ref{thm:convergence}}: As  $\mPhi_k$ is the solution of \eqref{general-ZhuZHx}, we have
$$
\underline\varrho(\mPhi_k,\underline\mPsi_{k-1},\underline\mS_{k-1}) + \nu_1 \|\mPhi_k^\T\mPhi_k - \mPhi_{k-1}^\T\mPhi_{k-1}\|_F^2 \leq  \underline\varrho(\mPhi_{k-1},\underline\mPsi_{k-1},\underline\mS_{k-1}),$$ 
which implies that
\e\begin{split}
\underline\varrho(\mPhi_{k-1},\underline\mPsi_{k-1},\underline\mS_{k-1}) - \underline\varrho(\mPhi_k,\underline\mPsi_{k-1},\underline\mS_{k-1}) &\geq \nu_1 \|\mPhi_k^\T\mPhi_k - \mPhi_{k-1}^\T\mPhi_{k-1}\|_F^2\\
&\geq \nu_1 \|\mPhi_k - \mPhi_{k-1}\|_F^2,
\end{split}\label{eq:sufficient decrease Phi}\ee
where the last inequality follows from \cite[Lemma 3.4]{li2018} by noting that $\mPhi_k$ is updated as \eqref{optimal-T-1x 2}.

Keeping the definition of $\vpsi_{i,k,\ell}$ in mind, we obtain\footnote{Like $\underline \varrho(\mPhi, \underline \mPsi, \underline \mS)$, we drop the data matrix $\mX$ in $\varrho(\mPhi, \mPsi, \mS, \mX)$ for simplifying the notation in the sequel.}
\begin{align*}
\varrho(\mPhi_{k},\mPsi_{i,k,\ell}/\vpsi_{i,k,\ell},\mS_{i,k-1}) + \nu_2\|\vpsi_{i,k,\ell} - \vpsi_{i,k-1,\ell}\|^2\leq \varrho(\mPhi_{k},\mPsi_{i,k,\ell}/\vpsi_{i,k-1,\ell},\mS_{i,k-1})
\end{align*}
for all $i\in[K]$ and $\ell\in[L_0]$. Repeating the above equations for all $\ell\in[L_0]$ and summing them up give
\e
\varrho(\mPhi_{k},\mPsi_{i,k-1},\mS_{i,k-1}) - \varrho(\mPhi_{k},\mPsi_{i,k},\mS_{i,k-1}) \geq \nu_2\|\mPsi_{i,k}- \mPsi_{i,k-1} \|^2_F
\label{eq:sufficient decrease Psi}\ee
for all $i\in[K]$.

The rest is to show similar sufficient decrease property in terms of $\mS_i$. If $\mS_{i,k}$ is updated via the first expression of \eqref{eq:update S 1}, it is clear that
\e
\varrho(\mPhi_{k},\mPsi_{i,k},\mS_{i,k-1}) -\varrho(\mPhi_{k},\mPsi_{i,k},\mS_{i,k}) \geq \nu_3\|\mPsi_{i,k}- \mPsi_{i,k-1} \|^2_F.
\label{eq:sufficient decrease S 1}\ee

Now, we consider the other case that $\mS_{i,k}$ is updated via  the 2nd expression of \eqref{eq:update S 1}. We define the proximal regularization of $\varrho(\mPhi,\mPsi,\mS)$ linearized at $\mS'$ as
\[
\varrho_{\nu'}(\mPhi,\mPsi,\mS,\mS') = \varrho(\mPhi,\mPsi,\mS') + \langle \nabla_{\mS} \varrho(\mPhi,\mPsi,\mS'), \mS - \mS' \rangle + \frac{\nu'}{2}\|\mS - \mS'\|_F^2.
\]
Due to~\eqref{eq:derive S}, we have
\e
\varrho(\mPhi,\mPsi,\mS) \leq \varrho_{2(1+\alpha + \beta)L_0}(\mPhi,\mPsi,\mS,\mS'), \ \forall \ \mS'\in\calS_{\kappa,b}^{J_0}.
\label{eq:suff decrease S proximal 1}\ee
Note that
\e\begin{split}
&\mS_{i,k} =\calP_{\mS\in\calS_{\kappa,b}^{J_0}}\left[ \mS_{i,k-1} - \nu_4 \nabla_{\mS} \varrho(\mPhi_{k},\mPsi_{i,k},\mS_{i,k-1}) \right]\\
 & =\argmin_{\mS\in\calS_{\kappa,b}^{J_0}}\left\|\mS -\left( \mS_{i,k-1}- \nu_4 \nabla_{\mS} \varrho(\mPhi_{k},\mPsi_{i,k},\mS_{i,k-1}) \right) \right\|_F^2\\
& = \argmin_{\mS\in\calS_{\kappa,b}^{J_0}} \varrho(\mPhi_{k},\mPsi_{i,k},\mS_{i,k-1}) + \langle \nabla_{\mS} \varrho(\mPhi_{k},\mPsi_{i,k},\mS_{i,k-1}), \mS - \mS_{i,k-1} \rangle + \frac{1}{2\nu_4}\|\mS - \mS_{i,k-1}\|_F^2\\
&=\argmin_{\mS\in\calS_{\kappa,b}^{J_0}}\varrho_{\frac{1}{\nu_4}}(\mPhi_{k},\mPsi_{i,k},\mS,\mS_{i,k-1}),
\end{split}
\label{eq:rewrite Sik}\ee
which implies that
\[
\varrho_{\frac{1}{\nu_4}}(\mPhi_{k},\mPsi_{i,k},\mS_{i,k},\mS_{i,k-1})\leq \varrho(\mPhi_{k},\mPsi_{i,k},\mS_{i,k-1}).
\]
This, together with \eqref{eq:suff decrease S proximal 1} (by setting $\mPhi = \mPhi_k$, $\mPsi = \mPsi_{i,k}$, $\mS = \mS_{i,k}$ and $\mS' = \mS_{i,k-1}$),  gives
\e\begin{split}
&\varrho(\mPhi_{k},\mPsi_{i,k},\mS_{i,k-1}) - \varrho(\mPhi_{k},\mPsi_{i,k},\mS_{i,k})\geq \varrho_{\frac{1}{\nu_4}}(\mPhi_{k},\mPsi_{i,k},\mS_{i,k},\mS_{i,k-1}) - \varrho_{2(1+\alpha + \beta)L_0}(\mPhi_{k},\mPsi_{i,k},\mS_{i,k},\mS_{i,k-1})\\
& \geq \frac{1}{2}(\frac{1}{\nu_4} - 2(1+\alpha + \beta)L_0)\|\mS_{i,k} - \mS_{i,k-1}\|_F^2.
\end{split}
\label{eq:sufficient decrease S 2}\ee
Let
\e
c_1 = \min\{\nu_1,\nu_2,\nu_3,\frac{1}{2}(\frac{1}{\nu_4} - 2(1+\alpha + \beta)L_0)\}.
\label{eq:c1}
\ee
Combining \eqref{eq:sufficient decrease Phi}-\eqref{eq:sufficient decrease S 2}, we conclude
\[
\varrho(\mPhi_{k-1},\mPsi_{i,k-1},\mS_{i,k-1}) - \varrho(\mPhi_{k},\mPsi_{i,k},\mS_{i,k})\geq c_1 \|(\mPhi_{k-1},\mPsi_{i,k-1},\mS_{i,k-1}) - (\mPhi_{k},\mPsi_{i,k},\mS_{i,k})\|_F^2,
\]
which implies
\e
\underline\varrho(\mPhi_{k-1}, \underline \mPsi_{k-1},\underline \mS_{k-1}) - \underline\varrho(\mPhi_{k}, \underline \mPsi_{k},\underline \mS_{k}) \geq c_1 \|\mW_{k-1} - \mW_k\|_F^2.
\label{eq:sufficient decrease 2}\ee
By repeating the above equation for all $k\geq 1$ and summing them, we have
\[
c_1 \sum_{k=1}^\infty \|\mW_{k-1} - \mW_k\|_F^2 \leq \underline\varrho(\mPhi_{0}, \underline \mPsi_{0},\underline \mS_{0})
\]
since $\underline\varrho(\mW_k)\geq 0$ for all $k$. Thus, we obtain $\lim_{k\rightarrow \infty} \|\mW_{k-1} - \mW_k\|_F = 0$.

The proof is completed by noting the fact that $\mW_k\in{\cal O}_{M,N}\times{\cal U}_N^{L_0K}  \times \calS_{\kappa,b}^{KJ_0}$ for all $k$ and hence $f(\mW_k)=\underline\varrho(\mPhi_{k}, \underline \mPsi_{k},\underline \mS_{k})$.\\

\noindent{\it Proof of Part~(ii) of Theorem \ref{thm:convergence}}:  Since the sets $\calO_{M,N},\calU_N,\calS_{\kappa,b}$ are compact and $\mW_k\in{\cal O}_{M,N}\times{\cal U}_N^{L_0K}  \times \calS_{\kappa,b}^{KJ_0}$ for all $k$, this bounded $\{\mW_k\}$, as suggested by  the Bolzano-Weiestrass Theorem \cite{bartle2000introduction}, has at least a convergent subsequence $\{\mW_{k_m}\}$ with limit point denoted as $\mW^\star$.

Note that $f(\mW_k)=\underline \varrho(\mPhi_k, \underline\mPsi_k, \underline\mS_k)$ and
\[
\lim_{k\rightarrow \infty} I_{{\cal O}_{M,N}}(\mPhi_k) = \lim_{k\rightarrow \infty} I_{{\cal U}_{N}^{L_0K}}(\underline\mPsi_k) = \lim_{k\rightarrow \infty} I_{\calS_{\kappa,b}^{KJ_0}}(\underline\mS_k) =0
\]
since the sets $\calO_{M,N},\calU_N,\calS_{\kappa,b}$ are compact, which also implies that $\mW^\star\in{\cal O}_{M,N}\times{\cal U}_N^{L_0K}  \times \calS_{\kappa,b}^{KJ_0}$. Recall that $\{\underline \varrho(\mPhi_k, \underline \mPsi_k, \underline \mS_k)\}$, as just shown above, is a non-increasing sequence and is convergent. As well known,  its subsequence $\{\underline \varrho(\mPhi_{k_m}, \underline \mPsi_{k_m}, \underline \mS_{k_m})\}$ is also convergent (since $\varrho$ is  continuous) and converges to the same limit point:
\[
 \lim_{k\rightarrow \infty} f(\mW_{k}) = \lim_{k_m\rightarrow \infty}f(\mW_{k_m}) = f(\mW^\star).
\]

We now show that $\mW^\star$ is a critical point  of $f(\mW)$.

Since $\mPhi_k$ is a solution of \eqref{general-ZhuZHx}, it should be a stationary point of the problem in the RHS of \eqref{general-ZhuZHx}, that is
\e
\vzero\in\nabla_{\mPhi}\underline\varrho(\mPhi_k,\underline\mPsi_{k-1},\underline\mS_{k-1}) + 4\nu_1 \mPhi_k(\mPhi_k^\T\mPhi_k - \mPhi_{k-1}^\T\mPhi_{k-1}) + \partial I_{{\cal O}_{M,N}}(\mPhi_k).\nonumber\ee
Equivalently,
\[
\underbrace{\nabla_{\mPhi}\underline\varrho(\mPhi_k,\underline\mPsi_{k},\underline\mS_{k}) -\nabla_{\mPhi}\underline\varrho(\mPhi_k,\underline\mPsi_{k-1},\underline\mS_{k-1}) -4\nu_1 \mPhi_k(\mPhi_k^\T\mPhi_k - \mPhi_{k-1}^\T\mPhi_{k-1})}_{\mD_{\mPhi_k}}\in \partial_{\mPhi} f(\mPhi_k,\underline\mPsi_{k},\underline\mS_{k}).
\]
To bound $\|\mD_{\mPhi_k}\|_F$, first note that
\e\begin{split}
&\|\mPhi_k(\mPhi_k^\T\mPhi_k - \mPhi_{k-1}^\T\mPhi_{k-1})\|_F \leq \|\mPhi_k^\T\mPhi_k - \mPhi_{k-1}^\T\mPhi_{k-1}\|_F\\
&= \|\mPhi_k^\T(\mPhi_k - \mPhi_{k-1}) + (\mPhi_k -\mPhi_{k-1})^\T\mPhi_{k-1}\|_F\\&\leq  \|\mPhi_k^\T(\mPhi_k - \mPhi_{k-1}) \|_F + \| (\mPhi_k -\mPhi_{k-1})^\T\mPhi_{k-1}\|_F \leq 2 \|\mPhi_k - \mPhi_{k-1}\|_F,
\end{split}\nonumber\ee
where we utilize the inequality $\|\mC\mD\|_F \leq \|\mC\|_2\|\mD\|_F$ and the fact that both $\mPhi_k$ and $\mPhi_{k-1}$ has spectral norm (i.e., largest singular value) $1$.
Also,  it follows from the global Lipschitz property of $\underline\varrho$ in \eqref{eq:Lipschitz gradient} that
\e\begin{split}
\|\nabla_{\mPhi}\underline\varrho(\mPhi_k,\underline\mPsi_{k},\underline\mS_{k}) -\nabla_{\mPhi}\underline\varrho(\mPhi_k,\underline\mPsi_{k-1},\underline\mS_{k-1})\|_F \leq L_c \|\mPhi_k - \mPhi_{k-1}\|_F.\end{split}
\nonumber\ee
Thus, the term $\|\mD_{\mPhi_k}\|_F$ can be bounded as
\e\begin{split}
\|\mD_{\mPhi_k}\|_F \leq (8\nu_1 + L_c)\|\mPhi_k - \mPhi_{k-1}\|_F.
\end{split}
\label{eq:safeguard Phi}\ee

Similarly, it follows from the definition of $\vpsi_{i,k,\ell}$ that the following holds
\e\begin{split}
\vzero \in \nabla_{\vpsi_\ell}\varrho(\mPhi_{k},\mPsi_{i,k,\ell}/\vpsi_{i,k,l},\mS_{i,k-1}) + 2\nu_2\left(\vpsi_{i,k,\ell} - \vpsi_{i,k-1,\ell}\right) + \partial I_{\calU_N}(\vpsi_{i,k,\ell}).
\end{split}\nonumber\ee
Noting that
\[
\partial_{\vpsi_{i,\ell}} f(\mPhi_k,\underline\mPsi_k,\underline\mS_k) = \nabla_{\vpsi_{\ell}} \varrho(\mPhi_k,\mPsi_{i,k},\mS_{i,k}) + \partial I_{\calU_N}(\vpsi_{i,k,\ell}),
\]
which, together with the above equation, gives
\e\begin{split}
\underbrace{\nabla_{\vpsi_{\ell}} \varrho(\mPhi_k,\mPsi_{i,k},\mS_{i,k}) - \nabla_{\vpsi_\ell}\varrho(\mPhi_{k},\mPsi_{i,k,\ell}/\vpsi_{i,k,\ell},\mS_{i,k-1}) - 2\nu_2\left(\vpsi_{i,k,\ell} - \vpsi_{i,k-1,\ell}\right)}_{\mD_{\vpsi_{i,k,\ell}}}\in \partial_{\vpsi_{i,\ell}} f(\mPhi_k,\underline\mPsi_k,\underline\mS_k).
\end{split}\nonumber\ee
Due to the global Lipschitz property of $\underline\varrho$ in \eqref{eq:Lipschitz gradient}, we have
\e\begin{split}
\left\|\nabla_{\vpsi_{\ell}} \varrho(\mPhi_k,\mPsi_{i,k},\mS_{i,k}) - \nabla_{\vpsi_l}\varrho(\mPhi_{k},\mPsi_{i,k,\ell}/\vpsi_{i,k,\ell},\mS_{i,k-1})\right\|_F&\leq L_c \|(\mPhi_k,\mPsi_{i,k},\mS_{i,k}) - (\mPhi_{k},\mPsi_{i,k,\ell}/\vpsi_{i,k,\ell},\mS_{i,k-1})\|_F \\&\leq L_c(\|\mPsi_{i,k} - \mPsi_{i,k-1}\|_F + \|\mS_{i,k}- \mS_{i,k-1}\|_F),
\end{split}\nonumber\ee
which implies that
\e\begin{split}
\|\mD_{\vpsi_{i,k,l}}\|_F&\leq L_c(\|\mPsi_{i,k} - \mPsi_{i,k-1}\|_F + \|\mS_{i,k}- \mS_{i,k-1}\|_F) + 2\nu_2\left\|\vpsi_{i,k,l} - \vpsi_{i,k-1,l}\right\|_F\\
&\leq (L_c+2\nu_2)\|\mPsi_{i,k} - \mPsi_{i,k-1}\|_F+ L_c \|\mS_{i,k}- \mS_{i,k-1}\|_F.
\end{split}
\label{eq:safeguard Psi}\ee

The rest is to show similar properties for the subdifferential in terms of $\mS_i$. If $\mS_{i,k}$ is updated via the first expression of \eqref{eq:update S 1}, by utilizing Lemma~\ref{lem:OMP stationary point} (the constraint $\|\mS_i\|_\infty\leq b$ can be simply incorporated into the analysis), we have
\e
\vzero \in \partial_{\mS_i}f(\mPhi_k,\underline\mPsi_k,\underline\mS_k).
\label{eq:safeguard Psi S1}\ee

When $\mS_{i,k}$ is updated via  the 2nd expression of \eqref{eq:update S 1}, the optimality
condition of \eqref{eq:rewrite Sik}  gives
\[
\vzero \in \nabla_{\mS} \varrho(\mPhi_{k},\mPsi_{i,k},\mS_{i,k-1}) + \frac{1}{\nu_4}(\mS_{i,k} - \mS_{i,k-1}) + \partial I_{\calS_{\kappa,b}^{J_0}}(\mS_{i,k}),
\]
which implies that
\[
\underbrace{\nabla_{\mS} \varrho(\mPhi_{k},\mPsi_{i,k},\mS_{i,k}) - \nabla_{\mS} \varrho(\mPhi_{k},\mPsi_{i,k},\mS_{i,k-1}) - \frac{1}{\nu_4}(\mS_{i,k} - \mS_{i,k-1})}_{\mD_{\mS_{i,k}}}\in \partial_{\mS_i}f(\mPhi_k,\underline\mPsi_k,\underline\mS_k)
\]
Similar to \eqref{eq:safeguard Psi}, we bound $\mD_{\mS_{i,k}}$ as
\e
\|\mD_{\mS_{i,k}}\|_F\leq (L_c+ \frac{1}{\nu_4})\|\mS_{i,k} - \mS_{i,k-1}\|_F .
\label{eq:safeguard Psi S2}\ee

We stack $\mD_{\mPhi_k},\mD_{\vpsi_{i,k,l}}$ and $\vzero$ (or $\mD_{\mS_{i,k}}$) into
\e\begin{split}
\mD_k &= \left(\mD_{\mPhi_k},\mD_{\vpsi_{1,k,1}}, \ldots, \mD_{\vpsi_{K,k,L_0}}, \vzero  \right),\\
\mD_k' &= \left(\mD_{\mPhi_k},\mD_{\vpsi_{1,k,1}}, \ldots, \mD_{\vpsi_{K,k,L_0}}, \mD_{\mS_{1,k}}, \ldots, \mD_{\mS_{K,k}}  \right).
\end{split}\nonumber\ee
We have either $\mD_k\in \partial f(\mPhi_k,\underline\mPsi_k,\underline\mS_k)$ or $\mD'_k\in \partial f(\mPhi_k,\underline\mPsi_k,\underline\mS_k)$. Let
\e
c_2 = 2(L_0 + 1)L_c +2L_0 \nu_2 + 8\nu_1+\frac{1}{\nu_4}.
\label{eq:c2}\ee
It follows from \eqref{eq:safeguard Phi}-\eqref{eq:safeguard Psi S2} that
\e
\|\mD_k\|_F,\|\mD'_k\|_F \leq c_{2}\|\mW_k - \mW_{k-1}\|_F.
\label{eq:gradient Psi bounded}\ee
Recalling from \eqref{eq:sufficient decrease 2} that $\lim_{k_m\rightarrow\infty}\|\mW_{k_m} - \mW_{k_m-1}\|_F = 0$, we have $\lim_{k_m\rightarrow\infty} \mD_{k_m} = \lim_{k_m\rightarrow\infty} \mD'_{k_m} = \vzero$.
Owing to the closedness properties of $\partial f(\mW_k)$, we finally obtain $\vzero\in \partial f(\mW^\star)$. Thus, $\mW^\star$ is a critical point of $f(\mW)$. This completes the proof of Theorem~\ref{thm:convergence}.\\
\end{proof}

\section{Proof of Theorem \ref{thm:sequence convergence}}
\label{sec:prf thm sequence convergence}

We first state the definition of Kurdyka-Lojasiewicz (KL) inequality, which is proved to be useful for  convergence analysis of the iterates sequence~\cite{attouch2010proximal,attouch2013convergence,bao2016dictionary,bao2015convergence,bolte2014proximal}.
\begin{definition}\label{def:KL}(Kurdyka-Lojasiewicz (KL) inequality \cite{attouch2010proximal,attouch2013convergence,bolte2014proximal}.)
	A proper semi-continuous  function $g(\vx):\Re^N\rightarrow \Re$ is said  to  satisfy the Kurdyka-Lojasiewicz (KL) inequality,  if for any stationary point $\vx^\star$  of $g(\vx)$, there exist constants $\delta>0, \theta\in[0,1), c>0$ such that
\[
\left|g(\vx) - g(\vx^\star)\right|^{\theta}  \leq c~ \textup{dist}[\partial g(\vx),0],~~\forall~\vx\in \setB(\vx^\star, \delta),
\]
where $\setB(\vx^\star, \delta) = \{\vx\in\Re^n:\|\vx - \vx^\star\|_2\leq \delta\}$ and $\textup{dist}[\partial g(\vx),\vzero] = \min_{\vd\in \partial g(\vx)}\|\vd\|_2$ denotes the distance of $\vzero$ to the set $\partial g(\vx)$.\\
\end{definition}

\begin{proof}
It is clear that $f(\mW)$ is a KL function since $\underline\varrho(\mW)$ is a KL function (as it is an analytical function) and $I_{{\cal O}_{M,N}}(\mPhi)$, $I_{{\cal U}_{N}^{L_0K}}(\underline\mPsi)$ and $I_{\calS_{\kappa,b}^{KJ_0}}(\underline\mS)$ are also KL functions as they are the indicator functions on the sets ${\cal O}_{M,N}$, ${\cal U}_{N}^{L_0K}$ and $\calS_{\kappa,b}^{KJ_0}$ respectively~\cite{bolte2014proximal}.

Let ${\cal H}_{\mW}$ denote the set of accumulation points of $\{\mW_k\}$, i.e.,
\[
{\cal H}_{\mW} = \left\{\overline\mW: \exists \ \textup{an increasing sequence of intergers}\ \{k_m\}_{m\in\mathbb{N}}, \ \textup{s.t.} \lim_{m\rightarrow \infty}\mW_{k_m} = \overline \mW \right\}.
\]
Recall part($ii$) of Theorem~\ref{thm:convergence}, which states that any element $\mW^\star\in{\cal H}_{\mW}$ is a stationary point of $f(\mW)$ and $f(\mW)$ has the same value at the set ${\cal H}_{\mW}$. Also since the sequence $\{\mW_k\}$ is bounded and regular (see \eqref{eq:regular}), we know that ${\cal H}_{\mW}$ is a nonempty compact connected set and satisfies \cite[Lemma 3.5]{bolte2014proximal}
\e
\lim_{k\rightarrow \infty} \text{dist}[\mW_k,{\cal H}_{\mW}] = 0.
\label{eq:W to L}\ee

First, note that  if there exists an integer $k_0$ for which $f(\mW_{k_0}) = f(\mW^\star)$ with $\mW^\star\in{\cal H}_{\mW}$, it follows from \eqref{eq:sufficient decrease} and \eqref{eq:limit f} that $\mW_{k} = \mW_{k_0}$ for all $k\geq k_0$. Thus, $\{\mW_k\}$ is convergent.

Now we assume $f(\mW_k)>f(\mW^\star)$.  Invoking \eqref{eq:W to L}, we know for any $\delta$, there exists $k_1$ such that $\text{dist}[\mW_k, {\cal H}_{\mW}]\leq \delta$ for all $k\geq k_1$. Since $f$ satisfies the KL inequality, there exists constants $\theta\in[0,1), c_3>0$ such that~\cite{bolte2014proximal}
\e
(f(\mW_k) - f(\mW^\star))^\theta \leq c_3 \text{dist}[\partial f(\mW_k),\vzero]
\label{eq:KL consequence}\ee
for all $k\geq k_1$.
From  the concavity of  the function $h(y)= y^{1-\theta}$ with domain $y>0$, we have
\begin{align*}
&\left( f(\mW_k) - f(\mW^\star)   \right)^{1-\theta} -\left(  f(\mW_{k+1}) - f(\mW^\star)    \right)^{1-\theta}\geq  (1-\theta) \frac{f(\mW_k)-f(\mW_{k+1})}{ \left( f(\mW_k) - f(\mW^\star)    \right)^{\theta}}\\
 &\geq  (1-\theta)\frac{c_1\|\mW_k - \mW_{k+1}\|_F^2}{c_3 \text{dist}[\partial f(\mW_k),\vzero]} \geq (1-\theta)\frac{c_1\|\mW_k - \mW_{k+1}\|_F^2}{c_2 c_3\|\mW_{k-1} - \mW_{k}\|_F },
\end{align*}
where $c_1$, $c_2$ and $c_3$ are respectively defined in \eqref{eq:c1}, \eqref{eq:c2} and \eqref{eq:KL consequence}, the second inequality utilizes \eqref{eq:sufficient decrease} and \eqref{eq:KL consequence}, and the last inequality follows from  \eqref{eq:gradient Psi bounded} and the fact $\text{dist}[\partial f(\mW_k),\vzero]\leq \max(\|\mD_k\|_F,\|\mD'_k\|_F)$ since $\mD_k\in\partial f(\mW)$ or $\mD_k'\in\partial f(\mW)$. Repeating the above equations for all $k\geq k_1$ and summing them give
\begin{align*}
&\frac{c_2c_3}{c_1(1-\theta)}\left( f(\mW_{k_1}) - f(\mW^\star)\right)^{(1-\theta)}\geq\frac{c_2c_3}{c_1(1-\theta)}\sum_{k=k_1}^\infty
\left( f(\mW_k) - f(\mW^\star) \right)^{1-\theta} -\left(  f(\mW_{k+1}) - f(\mW^\star) \right)^{1-\theta}\\
&\geq \sum_{k=k_1}^\infty \frac{\|\mW_k - \mW_{k+1}\|_F^2}{\|\mW_{k-1} - \mW_{k}\|_F }= \sum_{k=k_1}^\infty \frac{\|\mW_k - \mW_{k+1}\|_F^2}{\|\mW_{k-1} - \mW_{k}\|_F } + \|\mW_{k-1} - \mW_{k}\|_F - \|\mW_{k-1} - \mW_{k}\|_F\\
& \geq \sum_{k=k_1}^\infty 2\|\mW_k - \mW_{k+1}\|_F - \|\mW_{k-1} - \mW_{k}\|_F\geq \left(\sum_{k=k_1}^\infty \|\mW_k - \mW_{k+1}\|_F \right) - \|\mW_{k_1 - 1} - \mW_{k_1}\|_F,
\end{align*}
which implies
\[
\sum_{k=k_1}^\infty \|\mW_k - \mW_{k+1}\|_F <\infty.
\]
Thus, $\{\mW_k\}$ is a Cauchy sequence and hence $\{\mW_k\}$ is convergent.
\end{proof}
\bibliographystyle{siamplain}
\bibliography{RefSensingMtx}

\end{document}

%% file: ex_shared.tex
\usepackage{lipsum}
\usepackage{amsfonts}
\usepackage{graphicx}
\usepackage{epstopdf}
\usepackage{algorithmic}
\ifpdf
  \DeclareGraphicsExtensions{.eps,.pdf,.png,.jpg}
\else
  \DeclareGraphicsExtensions{.eps}
\fi

\numberwithin{theorem}{section}

\newcommand{\TheTitle}{On Collaborative Compressive Sensing Systems:\\ The Framework, Design  and Algorithm}
\newcommand{\TheAuthors}{Zhihui Zhu, Gang Li, Jiajun Ding,  Qiuwei Li, and
       Xiongxiong He }

\headers{Collaborative Compressive Sensing Systems}{\TheAuthors}

\title{{\TheTitle}\thanks{This work was funded by the Grants of NSFCs 61273195, 61304124, \&
 61473262.}}

\author{
  Zhihui Zhu\thanks{Department of Electrical Engineering at the Colorado School of Mines, 1500 Illinois St., Golden, CO 80401
    (\email{zzhu@mines.edu, qiuli@mines.edu}).}
  \and
  Gang Li \thanks{College of Automation \& Electrical Engineering, Zhejiang University of Science \& Technology,  Hangzhou, 310023, Zhejiang, P.R. China. (\email{ieligang@zjut.edu.cn}).}
  \and
  Jiajun Ding\thanks{College of Information Engineering, Zhejiang University of Technology,  Hangzhou, 310014, Zhejiang, P.R. China.}
  \and
  Qiuwei Li\footnotemark[2]
  \and
  Xiongxiong He \footnotemark[4]
}

\usepackage{amsopn}
\DeclareMathOperator{\diag}{diag}
